\documentclass[reqno]{amsart}
\usepackage[margin=1.5in]{geometry}

\usepackage{amsmath,amsthm,amscd,amssymb}
\usepackage{latexsym}
\usepackage{color}

\usepackage{todonotes}
\usepackage[all]{xy}

\usepackage{calrsfs}
\DeclareMathAlphabet{\pazocal}{OMS}{zplm}{m}{n}
\numberwithin{equation}{section}

\newtheorem{Thm}{Theorem}[section]

\newtheorem{Cor}[Thm]{Corollary}
\newtheorem{Lem}[Thm]{Lemma}
\newtheorem{Prop}[Thm]{Proposition}

\newtheorem{Cond}[Thm]{Condition}

\newcommand{\p}{\varphi}

\newcommand{\e}{\varepsilon}

\renewcommand{\L}{\Lambda}

\newcommand{\C}{{\mathcal{C}}}
\newcommand{\N}{\mathbb{N}}
\newcommand{\Z}{\mathbb{Z}}
\renewcommand{\P}{\mathbb{P}}
\newcommand{\R}{\mathbb{R}}
\newcommand{\E}{\mathbb{E}}

\renewcommand{\L}{\Lambda}

\newcommand{\Eb}{\pazocal{E}}

\newcommand{\Mb}{\pazocal{M}}

\newcommand{\Pa}{\mathcal{P}}
\newcommand{\Pam}{\mathcal{R}}

\newcommand{\Wb}{\pazocal{W}}

\DeclareMathOperator*{\argmax}{arg\,max}
\DeclareMathOperator*{\argmin}{arg\,min}

\newcommand{\dsp}{\displaystyle}

\title[Hydrodynamic Limits of Young Diagrams] {On Hydrodynamic Limits of Young Diagrams}

\author{Ibrahim Fatkullin, Sunder Sethuraman, and Jianfei Xue}

\address{Department of Mathematics, University of Arizona,  Tucson, AZ 85721}
\email{ibrahim@math.arizona.edu\\ sethuram@math.arizona.edu\\ jxue@math.arizona.edu}

\begin{document}

\begin{abstract}

We consider a family of stochastic models of evolving two-dimensional Young diagrams, given in terms of certain energies, with Gibbs invariant measures.
`Static' scaling limits of the shape functions, under these Gibbs measures, have been shown by several over the years.  The purpose of this article is to study corresponding `dynamical' limits of which less is understood. We show that the hydrodynamic scaling limits of the diagram shape functions may be described by different types parabolic PDEs, depending on the energy structure.

\end{abstract}




\maketitle


\section{Introduction}

Young diagrams or tableaux, originally introduced in the context of combinatorics and representation theory (cf. \cite{Ful}, \cite{Yo}), have proved to be useful in a variety of disciplines ranging from mathematical physics to genetics. In particular, language involving Young diagrams and their shape functions may be used to describe phenomena such as Bose-Einstein condensation \cite{EJU}, polymerization and molecular assembly \cite{CGH}, \cite{KSS},
and random partitions in coagulation-fragmentation processes \cite{B}, \cite{P}, and references therein, among others.

In this paper, we present a class of stochastic evolutions of two-dimensional Young diagrams, given in terms of certain microscopic energy structures, and show that the hydrodynamic scaling limits of the associated shape functions obey different types of parabolic PDEs, reflecting the type of the energy formulations.  Previously, there seems to be only a small literature studying dynamical Young diagrams, for instance \cite{ES} and \cite{FS1}, which treat processes where there is birth and death evolution of squares in the diagrams.  See also the monograph \cite{Fu} which reviews some of this work.   The purpose of this article is to analyze a natural, but different class of models, through new and robust techniques.  Later, we give a brief comparison with the results in \cite{ES} and \cite{Fu}, \cite{FS1}, the latter pair closest to ours in spirit.

To describe our results, we first discuss certain `static' limits, which set the stage.
Let $\p = (p_1,p_2,\dots,p_n)$ with $p_k\geq p_{k+1}$ be a partition of the integer $M(\p) := \sum_{k=1}^n p_k$.
For example, $\p = (4,2,2,1)$ corresponds to $9 = 4+2+2+1$.
We call $\xi = (\xi(k;\p))_{k\in \N}$, where $\xi(k;\p) = \#\left\{ m : p_m = k \right\}$, the size density of the partition $\p$. Vice versa, given $\xi$, one can reconstruct $\p$, and so in a sense they are interchangeable. In terms of $\xi$, $M(\p) = \sum_{k\geq 1} k\xi(k;\p)$.
Denote by $\psi(x)$ the associated shape (height) function: 
\begin{equation*}
	\psi(x) = \sum_{k\geq x} \xi(k;\p). 
\end{equation*}
The graph of $\psi$ is the Young diagram of $\varphi$.
Since $\xi(k;\p) = \psi(k) - \psi(k+1)$, the numbers $\xi$ can be viewed as the gradient particle description of the associated partition $\p$.  
See Fig. \ref {from Young Diagram to Particles}.

\begin{figure}
{
\resizebox{3.5cm}{!}
{
\begin{tikzpicture}
\draw [->] (0,0) -- (5,0);
\draw [->] (0,0) -- (0,5);
\draw [-,thick] (0,4) -- (1,4); 
\draw [-,thick] (1,3) -- (2,3); 
\draw [-,thick] (2,1) -- (4,1); 
\draw [-,thick] (4,0) -- (5,0);

\draw [dashed,thick] (0,1) -- (2,1);
\draw [dashed,thick] (0,2) -- (2,2);
\draw [dashed,thick] (0,3) -- (1,3);
\draw [dashed,thick] (1,0) -- (1,3);
\draw [dashed,thick] (2,0) -- (2,1);
\draw [dashed,thick] (3,0) -- (3,1);

\draw [dashed,thick] (4,1) -- (4,0); 
\draw [dashed,thick] (2,3) -- (2,1); 
\draw [dashed,thick] (1,4) -- (1,3); 
\draw [fill] (1,4) circle [radius=0.05];
\draw [fill] (2,3) circle [radius=0.05];
\draw [fill] (4,1) circle [radius=0.05];
\draw [fill,white] (1,3) circle [radius=0.05]; 
\draw (1,3) circle [radius=0.05];
\draw [fill,white] (2,1) circle [radius=0.05];
\draw (2,1) circle [radius=0.05];
\draw [fill,white] (4,0) circle [radius=0.05];
\draw (4,0) circle [radius=0.05];
\coordinate [label=below: { $\psi(x)$} ] () at (2,-0.5);
\draw  (1,0) node[below]{$1$} -- (1,0.1);
\draw  (2,0) node[below]{$2$} -- (2,0.1);
\draw  (3,0) node[below]{$3$} -- (3,0.1);
\draw  (4,0) node[below]{$4$} -- (4,0.1);
\draw  (0,1) node[left]{$1$} -- (0.1,1);
\draw  (0,2) node[left]{$2$} -- (0.1,2);
\draw  (0,3) node[left]{$3$} -- (0.1,3);
\draw  (0,4) node[left]{$4$} -- (0.1,4);
\end{tikzpicture}
}
}
\hspace{0.51in}
{
\resizebox{3.5cm}{!}
{
\begin{tikzpicture}
\draw [->] (0,0) -- (5,0);

\draw [fill] (1,0.5) circle [radius=0.075];
\draw [fill] (2,0.5) circle [radius=0.075];
\draw [fill] (2,1) circle [radius=0.075];
\draw [fill] (4,0.5) circle [radius=0.075];

\coordinate [label=below: {$\xi = (1,2,0,1,0,\ldots)$} ] () at (2,-0.5);
\draw  (1,0) node[below]{$1$} -- (1,0.1);
\draw  (2,0) node[below]{$2$} -- (2,0.1);
\draw  (3,0) node[below]{$3$} -- (3,0.1);
\draw  (4,0) node[below]{$4$} -- (4,0.1);
\end{tikzpicture}
}
}
\caption{The Young diagram and particle description  associated with the partition $(4,2,2,1)$.}
 \label{from Young Diagram to Particles}
\end{figure}
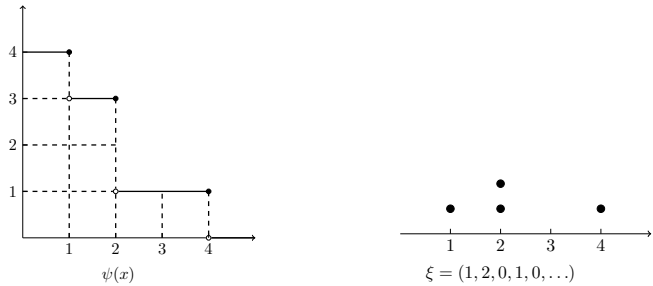

Let $\Pa_M$ be the uniform probability measure on all partitions of an integer $M$.
A classical result of A.\,Vershik \cite{V} states that in the limit as $M\to\infty$, the rescaled shape functions \mbox{$\psi_M(x) :=\psi(x\sqrt M)/\sqrt{M}$} converge in probability with respect to the canonical measure $\Pa_M$ to the curve
\begin{equation} \label {vershik curve}
\psi(x) = -\dfrac{\sqrt 6}{\pi} \ln \left( 1- e^{-\pi x/\sqrt 6} \right).
\end{equation}
Namely, for every $\epsilon>0$ and $a>0$, 
$$\lim_{M\to \infty} \Pa_M\left( \sup_{x\geq a}\big| \psi_M(x) - \psi(x)\big| >\epsilon\right) = 0.$$

Such results have a long history, and limits and phenomena different than the one above may appear if other ensembles, such as those with respect to Haar statistics, the Plancherel measure or Ewens measure are employed: see \cite{BOO}, \cite{EG},
\cite{FS2}, \cite{LS}, \cite{KV}, \cite{SV}, \cite{V}, \cite{VY} \cite{Y}, and references therein.

In this article, we will consider grand canonical ensembles of sizes $\{\xi(k):k\geq 1\}$, including those prescribed in \cite{FS2}: 
\begin{equation*}
\Pa_{\beta,N}(\xi)
=
\dfrac{1}{Z_{\beta,N}}
e^{ -\beta \sum_{k\geq 1}  \xi(k)\Eb_k  -N^{-1} M}
\end{equation*}
where $\Eb_k\geq 0$ is the energy of a summand of size $k$, total size $M=\sum_{k\geq 1} k\xi(k)$, inverse temperature $\beta\geq 0$, and $Z_{\beta, N}$ is the normalizing factor.  When $\beta=0$, the canonical, or conditional measures, with size $M$, are of course $\Pa_M$.

Consider the scaled shape function $ \psi_{\beta,N}(x):= N\psi(Nx)/R_{\beta,N}(M)$,
where $R_{\beta,N}(M)=N^2e^{-\beta\Eb_N}$, as shown in \cite{FS2}, is of the order of the expected value of $M=\sum_{k\geq 1}k\xi(k)$ with respect to $\Pa_{\beta,N}$.  This scaling is such that the expected area of the rescaled Young diagrams, $\sum_{x\geq 1}\E_{\Pa_{\beta, N}}\big[\psi_{\beta, N}(x)\big]$ is of order $1$;
see Fig. \ref{before and after scaling}. 
As $N\to \infty$, $\psi_{\beta,N}(x)$ will converge with respect to $\Pa_{\beta,N}$ to different limits,
 depending on the choice of the energy $\Eb_k$.

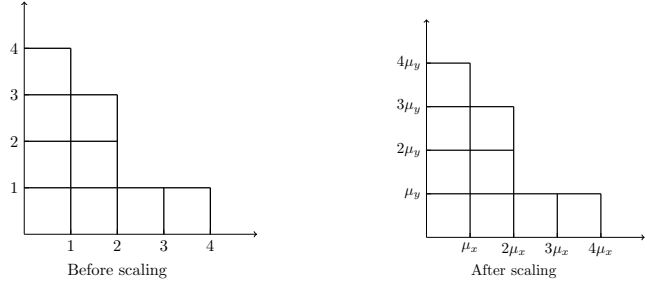
\begin{figure}
\resizebox{3.5cm}{!}
{
\begin{tikzpicture}
\draw [->] (0,0) -- (5,0);
\draw [->] (0,0) -- (0,5);
\draw [-,thick] (0,1) -- (4,1); 
\draw [-,thick] (0,2) -- (2,2); 
\draw [-,thick] (0,3) -- (2,3); 
\draw [-,thick] (0,4) -- (1,4); 
\draw [-,thick] (1,4) -- (1,0); 
\draw [-,thick] (2,3) -- (2,0); 
\draw [-,thick] (3,1) -- (3,0); 
\draw [-,thick] (4,1) -- (4,0); 
\coordinate [label=below: {Before scaling} ] () at (2,-0.5);
\draw  (1,0) node[below]{$1$} -- (1,0.1);
\draw  (2,0) node[below]{$2$} -- (2,0.1);
\draw  (3,0) node[below]{$3$} -- (3,0.1);
\draw  (4,0) node[below]{$4$} -- (4,0.1);
\draw  (0,1) node[left]{$1$} -- (0.1,1);
\draw  (0,2) node[left]{$2$} -- (0.1,2);
\draw  (0,3) node[left]{$3$} -- (0.1,3);
\draw  (0,4) node[left]{$4$} -- (0.1,4);
\end{tikzpicture}
}
\hspace{0.51in}
{
\resizebox{3.5cm}{!}
{
\begin{tikzpicture} 
\draw [->] (0,0) -- (5,0);
\draw [->] (0,0) -- (0,5);
\draw [-,thick] (0,1) -- (4,1); 
\draw [-,thick] (0,2) -- (2,2); 
\draw [-,thick] (0,3) -- (2,3); 
\draw [-,thick] (0,4) -- (1,4); 
\draw [-,thick] (1,4) -- (1,0); 
\draw [-,thick] (2,3) -- (2,0); 
\draw [-,thick] (3,1) -- (3,0); 
\draw [-,thick] (4,1) -- (4,0); 
\coordinate [label=below: {After scaling} ] () at (2,-0.5);
\draw  (1,0) node[below]{$\mu_x$} -- (1,0.1);
\draw  (2,0) node[below]{$2\mu_x$} -- (2,0.1);
\draw  (3,0) node[below]{$3\mu_x$} -- (3,0.1);
\draw  (4,0) node[below]{$4\mu_x$} -- (4,0.1);
\draw  (0,1) node[left]{$\mu_y$} -- (0.1,1);
\draw  (0,2) node[left]{$2\mu_y$} -- (0.1,2);
\draw  (0,3) node[left]{$3\mu_y$} -- (0.1,3);
\draw  (0,4) node[left]{$4\mu_y$} -- (0.1,4);
\end{tikzpicture}
}
}
\caption{Young diagrams before and after rescaling. 
$\mu_x = 1/N$, $\mu_y = N/R_{\beta,N}(M)$ in the rescaling  from $\psi$ to $ \psi_{\beta,N}$.}
\label{before and after scaling}
\end{figure}

Following \cite{FS2}, we assume that the energy function $\Eb_k$ is in form $\Eb_k = u(\ln k)$, where $u$ is a positive function diverging at infinity.
In particular, we consider two cases in this work:  (1) $u'(x) \to 1$, and (2) $u'(x)\to 0$.  We refer to these cases as
$\Eb_k \sim \ln k$, and $1\ll \Eb_k\ll \ln k$ respectively.  The precise specification later given in Condition \ref{condition} provides a large, varied class of energies, amenable to the scaling limits that we will take.

We remark, if $\Eb_k$ is not in this form, for instance 
the case $\Eb_k \gg \ln k $, there will be a finite number of particles, uniform over $N$, in the system (cf. Proposition 2.1 in \cite{FS2}), and so the associated scaling limits will be trivial.  Also, if $\Eb_k$ is constant, the situation is tantamount to taking $\beta=0$, and so we do not distinguish this case.  Furthermore, when $\Eb_k \sim \ln k$ and $\beta>1$, the variance of the scaled shape function $\psi_{\beta,N}$ diverges, and does not vanish for $\beta=1$ (cf. Proposition 2.4 in \cite{FS2}). 
There are also other interesting `boundary' energy scenarios discussed in \cite{FS2}, including condensation regimes, which we do not pursue here.

The following convergences follow from Propositions 2.1 and 2.2 of \cite{FS2}:  For $\epsilon>0$,
\begin{enumerate}
\item $\beta =0$: $ \Pa_{\beta,N}\big(\big|\dsp \psi_{\beta,N} - \ln(1-e^{-x})\big|>\epsilon\big)\to 0$;
\item $\Eb_k \sim \ln k$, $0<\beta <1$: $\Pa_{\beta,N}\big(\big| \dsp \psi_{\beta,N} -  \int_x^{\infty} u^{-\beta} e^{-u} du\big|>\epsilon\big)\to 0$;
\item $1\ll \Eb_k \ll \ln k$, $\beta>0$: $ \Pa_{\beta,N}\big(\big|\psi_{\beta,N} - e^{-x}\big|>\epsilon\big)\to 0$.
\end{enumerate}
We remark, the limit when $\beta=0$, is similar to Vershik's result, and in some sense, a reflection of the equivalence of ensembles between the canonical measures $\Pa_M$ and $\Pa_{0,N}$ as $M$ and $N$ diverge.

With this background, the purpose of the article is to consider a natural dynamics of these varied shapes and to understand their hydrodynamic limits.
Previously, in \cite{FS1}, Funaki and Sasada studied an evolutional model of the Young diagrams, with respect to the `uniform' grand ensembles $\Pa_{0,N}$, as well as certain `restricted' uniform ensembles when $\beta =0$, providing a dynamical interpretation with respect to the Vershik curve $\psi$ \eqref{vershik curve}.  However, the dynamics that we introduce is more general and different
than that in \cite{FS1}.

Consider the gradient particle system associated with the Young diagrams with generator 
\begin{equation*}
\begin{split}
Lf(\xi)
=&
\sum_{k=1}^{\infty}
\left\{
\lambda_k 
\left [  f\left(\xi^{k,k+1} \right)  - f(\xi)   \right] \chi_{\{\xi(k) > 0\}}
\right.\\
&\quad \quad\quad \quad+
\left. \left [  f\left(\xi^{k,k-1}  \right)   - f(\xi) \right ] \chi_{\{\xi(k) > 0,k>1\}}
\right\}
\end{split}
\end{equation*}
where $ \lambda_k = 
e^{-\beta\left (\Eb_{k+1} - \Eb_k\right) - 1/N }$ (cf. Fig. \ref{gradient dynamics}).  Here, $\xi^{k,k\pm 1}$ is the configuration obtained by moving a particle from $k$ to $k\pm 1$. 

 The interpretation of this dynamics, which preserves particle mass, in the `language of polymers' is as follows:  A monomer is added to a polymer of size $k$ with rate $\lambda_k$ and removed with rate $1$. In this dynamics, the gradients $\xi$ qualitatively tend to states of lower energy $\Eb_\cdot$.   This dynamics is spatially inhomogeneous when $\beta> 0$ in that $\lambda_k\neq \lambda_{k+1}$, and is not translation-invariant in general, being limited to $\Z^+$, rather than $\Z$.  An important feature is that the grand canonical measures $\Pa_{\beta,N}$ are invariant under $L$.
See \cite{Fl}, \cite{KSS} for discussions of related polymerization processes.

 Moreover, in terms of the associated Young diagrams, an `empty' lower left corner, adjacent to three squares, with vertex at $(k,\cdot)$ is filled with a square with rate $\lambda_k$, and a square, with an upper right corner not adjacent to any other square, is removed with rate $1$; for instance, in Fig. 1, turning the empty corner at $(1,3)$ into a square corresponds with the particle at $k=1$ moving to location $k=2$, and removing the square with corner $(2,3)$ means a particle at $k=2$ moves to $k=1$.

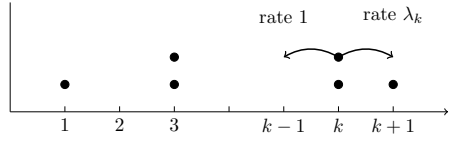
\begin{figure}
\resizebox{6cm}{!}
{
\begin{tikzpicture}
\draw [->] (0,0) -- (8,0);
\draw [-] (0,0) -- (0,2);
\draw [fill] (1,0.5) circle [radius=0.075];
\draw [fill] (3,1) circle [radius=0.075];
\draw [fill] (3,0.5) circle [radius=0.075];
\draw [fill] (6,0.5) circle [radius=0.075];
\draw [fill] (6,1) circle [radius=0.075];
\draw [fill] (7,0.5) circle [radius=0.075];

\draw [->, thick] (6,1) to [out=150,in=30] (5,1);
\draw [->, thick] (6,1) to [out=30,in=150] (7,1);

\draw  (1,0) node[below]{$1$} -- (1,0.1);
\draw  (2,0) node[below]{$2$} -- (2,0.1);
\draw  (3,0) node[below]{$3$} -- (3,0.1);
\draw  (4,0) -- (4,0.1);
\draw  (5,0) node[below]{$k-1$} -- (5,0.1);
\draw  (6,0) node[below]{$k$} -- (6,0.1);
\draw  (7,0) node[below]{$k+1$} -- (7,0.1);

\coordinate [label=above: {rate $1$} ] () at (5,1.5);
\coordinate [label=above: {rate $\lambda_k$} ] () at (7,1.5);
\end{tikzpicture}
}
\caption{Gradient particle system: Particles at sites $k\geq 2$ move to the left with rate $1$, to the right with rate $\lambda_k$; particles at $k=1$, move only to the right with rate $\lambda_1$.}
\label {gradient dynamics}
\end{figure}

Let $\xi_t$ denote the associated Markov process.  
We will be interested in the process $\eta_t = \xi_{N^2t}$ seen in diffusive scale, where time is speeded up by $N^2$ and space by $N$. 
Since $\eta_t$ is viewed as the negative gradient of its corresponding height function $\psi$,
the scaling from $\psi$ to $\psi_{\beta,N}$ (cf. Fig. \ref{before and after scaling}) motivates 
the following definition of the empirical measure
\begin{equation*}
\pi^{N}_t(dx)
=
\frac{N_\beta}{N}
\sum_{k=1}^{\infty}
\eta_t (k) \delta_{k/N}(dx).
\end{equation*}
 Here, $N_\beta = e^{\beta \Eb_N}$ is a choice so that the total mass of $\pi^N_0$ under $\Pa_{\beta,N}$ is of $O(1)$.

We will show (Theorems \ref{thm: beta 0}, \ref{thm: ln k}, and \ref{thm: lnln k}), under diffusive scalings, for a large class of initial conditions supported on configurations with $O(NN_\beta^{-1})$ expected number of particles at level $N$, that
the empirical measures $\pi_{t}^{N}$ converge weakly to a delta mass supported on the unique weak solution of a macroscopic equation, depending on the structure of the energy $\Eb_\cdot$, as $N\to \infty$:
\begin{enumerate}
\item $\beta=0$: 
$
\partial_t \rho
= \partial_x^2 \dfrac{\rho}{\rho+1} + \partial_x \dfrac{\rho}{\rho+1}$;
\item $0<\beta<1$, $\Eb_k \sim \ln k$:
$\partial_t \rho
=
\partial_x^2\rho + \partial_x \left(\dfrac{\beta+x}{x} \rho\right)$;
\item $\beta>0$, $1\ll \Eb_k \ll \ln k$: 
$
\partial_t \rho
=
\partial_x^2\rho + \partial_x \rho$.
\end{enumerate}

Since the particle density is related to the shape function by
$\psi(x) = \int_x^{\infty} \rho(u)du$, we obtain (Corollary \ref{cor: diagrams}) the macroscopic equations for $\psi$:
\begin{enumerate}
\item[(1')] $\beta=0$: 
$
\partial_t  \psi
= \partial_x \left(\dfrac{\partial_x \psi}{1- \partial_x \psi} \right)+ \dfrac{\partial_x \psi}{1-\partial_x \psi}
$;
\item[(2')] $0<\beta<1$, $\Eb_k \sim \ln k$:
$\partial_t \psi
=
\partial_x^2 \psi + \dfrac{\beta+x}{x} \partial_x \psi
$;
\item[(3')] $\beta>0$, $1\ll \Eb_k \ll \ln k$: 
$
\partial_t \psi
=
\partial_x^2 \psi + \partial_x  \psi
$.
\end{enumerate}

To shed light on these limits, the drift $N(\lambda_k -1)$ is quite informative.  When $\beta=0$, or when $1\ll \Eb_k \ll \ln k$, this drift tends to $-1$,
 but when $\Eb_k \sim \ln k$, it converges to a function of the scaled position.  The function $\rho/(1+\rho)$ is in a sense the macroscopic average value of $\chi_{\{\eta_t(k)>0\}}$ with respect to the grand canonical ensemble.   When $\beta =0$, the scaling limit recovers this form.  But, when $\beta>0$, as there is an additional scaling factor involved to obtain a nontrivial limit, what needs to be replaced is $N_\beta \chi_{\{\eta_t(k)>0\}}$, which is close to the linearization of $\rho/(1+\rho)$, namely $\rho$; see Step 1 of Section \ref{section beta not 0} for a more technical discussion.
From a physical perspective, the linear PDE limits reflect an effective transport of mass, which was not immediately apparent to us before deriving them.

The proof strategy is to consider the evolution of the empirical measure $\pi^N_t$ acting on test functions through It\^o's formula with respect to the zero-range process $\eta_\cdot$.  In calculating the generator action, nonlinear functions of $\eta_\cdot$ emerge.  However, because of non translation-invariance and inhomogeneity, standard methods such as `entropy' or `relative entropy' do not apply immediately to replace these terms with averaged expressions in terms of $\pi^N_t$.  We use nontrivial modifications, however, of certain `local' hydrodynamic $1$ and $2$-block replacement estimates, originally introduced in the study of `tagged' particles in \cite{JLS}.  This replacement, in particular, makes use of a spectral gap estimate that we provide and Feynman-Kac and Rayleigh formulas.  Interestingly, only when $\beta=0$, does one need both `local' $1$ and $2$-block replacements.  Otherwise, when $\beta> 0$, a `local' $1$-block replacement suffices.  In the proof of the $1$ and $2$-block estimates, we use that the process is `attractive', a feature which allows a certain coupling to be employed, facilitating truncation and other estimates.  Then, with tightness of the empirical measures, and uniqueness of weak solutions, that we provide, the limits follow.  See Sections \ref{section beta 0}, \ref{section beta not 0}, and \ref{section: diagrams limit} for more detailed proof outlines and remarks.

We note, although equations $(1), (1')$ when $\beta=0$ match that in \cite{FS1}, up to a constant in front of the first order derivative term, our results are different in several ways.
Here, the dynamics that we work with is weakly asymmetric zero-range process (WAZRP) on $\Z^+$, which is in general spatially inhomogeneous, and one
whose evolution preserves the total number of particles.  However, the model in \cite{FS1} is a different WAZRP on $\Z^+$, one which does not conserve particle mass, with a weakly asymmetric reservoir at site $0$.  Importantly, the proof in \cite{FS1} relies on the presence of this reservoir.
Also, \cite{FS1} considers initial profiles $\psi(0,x)$ where $\lim_{x\to 0}\psi(0,x) = \infty$
and obtain scaling limits $\psi(t,x)$ such that also $\lim_{x\to 0}\psi(t,x) = \infty$ and the hydrodynamic equation when $\beta =0$ holds. However, the initial conditions are different in our case:  We consider initial profiles, finite at time $0$ and for all later times $t$, that is $\psi(t,0) = \psi(0,0)<\infty$, by conservation of particles in the dynamics.  Moreover, it seems such profiles are not admissible with respect to the proof in \cite{FS1}, nor it seems are diverging profiles $\psi(0,x)$ amenable to our arguments, which make use that there are a finite number of particles at each level $N$.

From a broader point of view, random growth of Young diagrams also relates with the much studied corner growth model in which only the addition of squares to the diagram is allowed.
Formally, in the study of hydrodynamic limits of the corner growth model, the problem is often converted, by considering gradients, to a totally asymmetric simple exclusion process,
and the scaling is Euler, that is time and space are scaled at the same order.  See \cite{ES} which discusses such and other dynamics.
In contrast, our model of evolutional Young diagrams is studied via their gradient systems which is a WAZRP.
Our analysis is also directly on this WAZRP on $\Z^+$ and no further transformation to simple exclusion processes is employed.

\vskip .1cm
{\bf Organization of the article.}
The precise description of the model and results are given in Section \ref{model_results}.  Then, after preliminary definitions and estimates with respect to basic martingales in Section \ref{section: martingale}, we give the proof outlines of Theorems \ref{thm: beta 0}, \ref{thm: ln k}, and \ref{thm: lnln k}, and Corollary \ref{cor: diagrams} in Sections \ref{section beta 0}, \ref{section beta not 0}, and \ref{section: diagrams limit} respectively.  Main inputs into the proof are tightness and other estimates of the underlying measures given in Section \ref{section: tightness}.  In Section \ref{section: 1 and 2 blocks}, the important $1$ and $2$-block estimates are shown.  Useful properties of the initial measures are given in Section \ref{section: initial measures}.  Uniqueness of weak solution to the hydrodynamic equations is proved in Section \ref{section: uniqueness}.  Finally, in the appendix, some remarks about boundary phenomena of invariant measures are made.


\section{Model description and results}
\label{model_results}

We first specify certain Gibbs measures and their `static' limits, which inform and motivate next our dynamical model that we introduce.  Then, after prescribing the initial conditions considered, we give the hydrodynamic limit results.  

\subsection{Grand canonical ensembles and `static' limits}
\label{GC ensembles and static limit}
Let $\N = \{1,2,\ldots\}$ be the natural numbers, and $\Omega = \left\{0,1,2,\ldots \right\}^{\N}$ be the space of particle configurations.
A configuration $\xi = (\xi(k))_{k\in \N}\in \Omega$ specifies that there are $\xi(k)$ particles at sites $k\geq 1$.

Suppose that each particle at site $k$ carries energy $\Eb_k$, with respect to a function $\Eb_\cdot : \{0,1,2, \ldots\} \mapsto \R^+:=[0,\infty)$.
Following \cite{FS2}, we will assume that the energy function $\Eb_k$ has the following structure.  Let $\R_\circ^+:=(0,\infty)$.
\begin{Cond} \label {Condition on energy}
\label{condition}
$\Eb_k  = u(\ln k)$ where $u(\cdot): \R^+\mapsto \R_\circ^+$ is
differentiable and $u'(\cdot)$ is bounded, $\lim_{x\to \infty} u(x) = \infty$, and $\lim_{x\to \infty} u'(x) = 0$ or $1$.
We will say
\begin{itemize}
\item  `$\Eb_k \sim \ln k$' denotes the case $\lim_{x\to \infty} u'(x) =1$ and 
\item `$1\ll \Eb_k \ll \ln k$' stands for the case  $\lim_{x\to \infty} u'(x) =0$.
\end{itemize}
\end{Cond}

In passing, we note the constant $1$ in the limit when $\Eb_k\sim \ln k$ is chosen to be definite, although it could be specified as another positive constant.  Also, as the derivative $u'$ is bounded, that the infimum $\inf \Eb_k/k =0$ is achieved as $k\uparrow\infty$, a specification important in \cite{FS2}.  In addition, the condition allows a comparison, $\Eb_k - \Eb_l = u'(\ln y)\ln(k/l)$, where $y$ is between $k$ and $l$, afforded by the mean value theorem, which will be useful in some later estimates.

\medskip
For fixed $\beta\geq 0$, specify the grand canonical ensemble on $\Omega$, 
\begin{equation*}
\Pa_{\beta,N}(\xi)
=
\dfrac{1}{Z_{\beta,N}}
e^{-\beta \sum_{k\in \N} \xi(k) \Eb_k - N^{-1} \sum_{k\in \N} k \xi(k)}.
\end{equation*}
Observe that $\Pa_{\beta,N}$ has a product structure:  $\Pa_{\beta,N}(\xi) = \prod_{k=1}^{\infty} \Pa_{\beta,N,k} ( \xi(k) )$
where $\Pa_{\beta,N,k}$ is Geometric with parameter 
$$\theta_k = e^{-\beta \Eb_k - k/N},$$
  that is, for $n\geq 0$, 
	$$\Pa_{\beta,N,k}(n) = (1-\theta_k) \theta_k^{n}.$$

Let 
$$c_0 =\min_{k} e^{\beta \Eb_k}.$$
 Trivially $c_0 = 1$ when $\beta = 0$ and $c_0\geq 1$ otherwise.
For fixed $\beta$ and $0\leq c\leq c_0$, we introduce the product measures on $\Omega$,
\begin{equation*}
\Pam_{c,N} (\xi) = \prod_k \Pam_{\beta,c,N,k} (\xi(k)).
\end{equation*}
Here, the marginal $\Pam_{\beta,c,N,k}$ is the Geometric distribution with parameter 
$$\theta_{k,c} = c \theta_k = ce^{-\beta \Eb_k - k/N}$$ 
and mean 
\begin{equation}
\label{rho k c}
\rho_{k,c} = \frac{ \theta_{k,c}}{1-\theta_{k,c}} = \frac{ce^{-\beta \Eb_k - k/N}}{1-ce^{-\beta \Eb_k - k/N}},
\end{equation}
 well-defined when $c\leq c_0$.  

The strength of the parameter $c$ reflects the density of the sizes $\{\xi(k)\}$ in the system.   
Clearly, $\Pa_N = \Pa_{0,N}$ is the special case of $\Pam_{c,N}$ with $\beta =0$ and $c = 1$.  Also, we note the case $c=0$ is trivial, as $\Pam_{0,N}$ puts no particles anywhere.

The family $\{\Pam_{c,N}\}$ will be seen as invariant measures for the dynamics, specified in the next subsection.

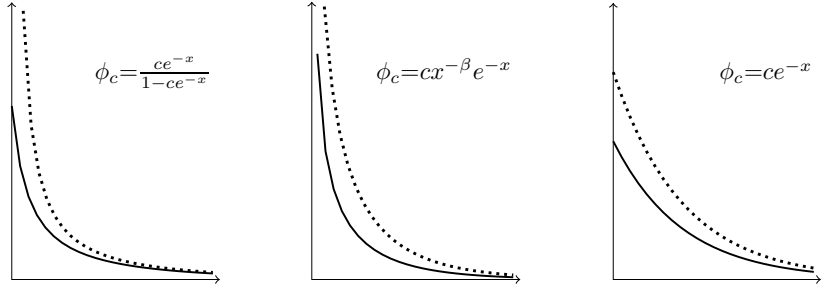
\begin{figure}
\resizebox{3cm}{!}
{
\begin{tikzpicture}
   \draw [<->] (0,4.) -- (0,0) -- (3.,0);
   \draw[dotted,very thick,domain=0.165:2.9] plot (\x, {(1/3)/(e^(\x/2) - 1)}  ); 
   \draw[thick,domain=0:2.9] plot (\x, {(1/3)/(e^((\x+0.25)/2) - 1)}  ); 
   \coordinate [label=left:$\phi_c  {=} \frac{c e^{-x}}{1- ce^{-x}}$] (U) at (3,3);
\end{tikzpicture}
}
\hspace{0.3in}
\resizebox{3cm}{!}
{
\begin{tikzpicture}
   \draw [<->] (0,4.) -- (0,0) -- (3.,0);
   \draw[dotted,very thick,domain=0.18:2.9] plot (\x, {2*\x^(-0.5)* e^(-\x)}  ); 
  \draw[ thick,domain=0.08:2.9] plot (\x, {\x^(-0.5)* e^(-\x)}  ); 
   \coordinate [label=left:$\phi_c {=} c x^{-\beta} e^{-x}$] (U) at (3,3); 
   \end{tikzpicture}
}
\hspace{0.3in}
\resizebox{3cm}{!}
{
\begin{tikzpicture}
   \draw [<->] (0,4.) -- (0,0) -- (3.,0);
   \draw[dotted,very thick,domain=0:2.9] plot (\x, {3* e^(-\x)}  ); 
  \draw[ thick,domain=0:2.9] plot (\x, {2* e^(-\x)}  ); 
   \coordinate [label=left:$\phi_c  {=} ce^{-x}$] (U) at (3,3); 
\end{tikzpicture}
}
\caption{Examples of $\phi_c$ in all the three regimes. The dotted curves represent $c=c_0$ and solid curves are for general $c$'s which are strictly less than $c_0$.}
\label {phi c pictures}
\end{figure}

Recall 
\begin{equation}
\label{Nbeta eqn}
N_{\beta} = e^{\beta \Eb_{N}}.
\end{equation}
We distinguish three regimes depending on the form of $\Eb_k$ and $\beta$:
\begin{itemize}
\item [(1)] $\beta =0$:  $N_{\beta}=1$, 
\item [(2)] $\Eb_k \sim \ln k$ and $0<\beta<1$: $N_\beta=o(N)$ and $\lim_{N\uparrow\infty}N_\beta = \infty$,
\item [(3)] $1\ll \Eb_k \ll \ln k$ and $\beta>0$:  $N_\beta = o(N)$ and $\lim_{N\uparrow\infty}N_\beta = \infty$.
\end{itemize}

When $c<c_0$, in Lemma \ref{Lem: mean variance mu N}, we show the following mean $E_{\Pam_{c,N}}$ and variance $\text{Var}_{\Pam_{c,N}}$ estimates, under $\Pam_{c,N}$, for the number of particles in the system:
 \begin{equation}
\label{mean_var_pam}
E_{\Pam_{c,N}} \sum_{k=1}^\infty \xi(k) =O(NN_\beta^{-1}), \ \ {\rm and \ \ }
\text{Var}_{\Pam_{c,N}} \sum_{k=1}^\infty \xi(k) = o(N^2N^{-2}_\beta).
\end{equation}
  However, when $c=c_0$, we show in Lemma \ref{app_1} in the Appendix that the orders of the expected value and variance are strictly greater.
In a sense, the case $c=c_0$ represents a boundary, avoided for the most part in the sequel, so that we may unify statements and techniques.

In the three cases above, we now associate certain profiles $\phi_c$:
\begin{itemize}
\item [(1)] $\phi_c=\dfrac{c e^{-x}}{1- ce^{-x}}$ when $\beta=0$, 
\item [(2)] $\phi_c=c x^{-\beta} e^{-x}$ when $\Eb_k \sim \ln k$ and $0<\beta<1$, 
\item [(3)] $\phi_c=c e^{-x}$ when $1\ll  \Eb_k\ll \ln k$ and $\beta>0$.
\end{itemize}
When $0\leq c<c_0$, we observe that $\phi_c\in L^1(\R^+)$.
These profiles are the `static' limits of the gradients under the measures $\Pam_{c,N}$.

\begin{Prop} \label {prop: static}
Suppose $\Eb$ and $\beta$ satisfy the conditions of regimes (1), (2) or (3) above.
Fix $0\leq c <c_0$. Then, for any test function $G\in C_c^{\infty}(\R^+_{\circ})$ and $\delta>0$
\begin{equation}
\label{static limit}
\lim_{N\to \infty} \Pam_{c,N}
\left[ \left| \frac{N_\beta}{N}\sum_{k=1}^\infty G(k/N)\xi(k)
-
\int_0^{\infty} G(x) \phi_c(x) dx
\right|
>\delta
 \right]
 =0
\end{equation}
where $\phi_c$ takes the appropriate form in each regime (1), (2) or (3).
\end{Prop}

In passing, we remark, when $c=c_0$, the above limit still holds.  See Lemma \ref{app_2} in the Appendix for an argument.

We will state later in Subsection \ref{Initial conditions} that this proposition is a corollary of \eqref{Initial convergence nu N}, which is proved in Proposition \ref{prop: initial convergence and mass for condition 1}.

\subsection{Dynamics}
We now define the gradient evolutions of the Young diagrams.  
Informally, particles at site $k$ jump to its right site $k+1$ with rate 
$\lambda_k := \dfrac{\theta_{k+1}}{\theta_k}$ and to its left site $k$ with rate $1$.
Particles at site $1$ jump only to site $2$.

For each $N\geq 1$, the evolution is a type of zero-range Markov process, $\xi_t = \left( \xi_t(k)\right)_{k\geq 1} \in \Omega$, on $\mathbb{Z}^+$ and generator 
\begin{equation*}
\begin{split}
Lf(\xi)
=&
\sum_{k=1}^{\infty}
\left\{
\lambda_k 
\left [  f\left(\xi^{k,k+1} \right)  - f(\xi)   \right] \chi_{\{\xi(k) > 0\}}
+
\left [  f\left(\xi^{k,k-1}  \right)   - f(\xi) \right ] \chi_{\{\xi(k) > 0,k>1\}}
\right\}
\end{split}
\end{equation*}
where 
\begin{equation} \label {eqn: lambda k}
\lambda_k
=
\dfrac{\theta_{k+1}}{\theta_k}
=
e^{-\beta\left (\Eb_{k+1} - \Eb_k\right) - 1/N }.
\end{equation}
Here, $\xi^{x,y} (k) =
\xi(k) -1$, $\xi(k)+1$, and $\xi(k)$ when respectively $k=x$, $k=y$, and $k\neq x,y$. We note when $\beta> 0$, the process has spatially inhomogeneous rates in that $\lambda_k$ is not constant in $k$.  See \cite{A} for more discussion about zero-range processes.

Under the initial measures we use, there will be a large, but finite number of particles, of order $O(NN_\beta^{-1})$, at all times in the system, and so in fact the process can be seen as a countable state space chain.

In Lemma \ref{lem: invariant measure},
 we verify that $E_{\Pam_{c,N}} (Lf(\xi)) = 0$ for all bounded, test functions $f$ depending only on a finite number of occupation variables $\{\xi(k)\}$.  Therefore, 
the family of measures $\left\{ \Pam_{c,N}\right\}$ is invariant under the dynamics generated by $L$.

We will observe the evolution speeded up by $N^2$, and consider in the sequel the process $\eta_t: = \xi_{N^2t}$, generated by $N^2L$, for times $0\leq t\leq T$, where $T>0$ refers to a fixed time horizon.

\medskip

We will access the space-time structure of the process through the scaled mass empirical measure,
\begin{equation*}
\pi_t^N(dx)
:=
\dfrac{N_\beta}{N} 
\sum_{k=1}^{\infty}
\eta_t (k) \delta_{k/N}(dx).
\end{equation*}
Clearly $\pi_t^{N}$ is a locally finite measure on $\R^+_{\circ}$.
Let $\Mb$ be the space of locally finite measures on $\R^+_{\circ}= (0,\infty)$, and observe that $\pi^N_t\in \Mb$.
Let also $C_c(\R^+_{\circ})$ be the space of compactly supported continuous on $\R^+_{\circ}$, endowed with the topology of uniform convergence on compact sets.
For $\left\{f_k\right\}_{k\in \N}$ a countable dense set in $C_c(\R^+_{\circ})$, we equip $\Mb$ with the distance 
$
d(\mu,\nu)
=
\sum_{k=1}^{\infty}
2^{-k}
\frac{\left|  \int f_k (d\mu - d\nu)  \right|}{1+  \left| \int f_k (d\mu - d\nu)  \right|}
$.

Then, $(\Mb, d)$ is a complete separable metric space
and, for  a sequence of measures in $\Mb$,
convergence in the metric $d$ is equivalent to convergence in the vague topology.
Here, the trajectories $\{\pi^N_t: 0\leq t\leq T\}$ are elements of the Skorokhod space $D([0,T],\Mb )$, endowed with the associated Skorokhod topology.

In the following, for $G\in C_c(\R^+_o)$ and $\pi\in \Mb$, denote $\langle G, \pi\rangle = \int_0^\infty G(u)d\pi(u)$.
Also, for a given measure $\mu$, we denote expectation and variance with respect to $\mu$ by $E_\mu$ and $\text{Var}_\mu$.  Also,
the process measure and associated expectation governing $\eta_\cdot$ starting from $\mu$ will be denoted by $\P_\mu$ and $\E_\mu$. 

\subsubsection{Attractiveness of the dynamics}
\label{subsec: attractiveness}
Since $\chi_{\{\xi(k)>0\}}$ is an increasing function in $\xi$, the dynamics generated by $L$ is `attractive', a fact that allows use of the `basic coupling' in our proofs (cf.\,\cite{A}, Chapter II in \cite{L}):
Let $\mu$, $\nu$ be two probability measures on $\Omega$.  We say that  $\mu\leq \nu$, that is $\mu$ is stochastically dominated by $\nu$, 
if for all $f: \Omega \to \R$ coordinately increasing, we have $E_{\mu}(f) \leq E_{\nu} (f)$. 
Attractiveness asserts that if $\mu\leq \nu$, then we have $\E_{\mu}(f(\xi_t)) \leq \E_{\nu}(f(\xi_t))$ for all $t\geq 0$.

\subsection{Initial conditions}
\label{Initial conditions}

We first specify a set of natural initial conditions, which will be a case of a more general class of initial conditions given later.  Consider an initial density profile $\rho_0:\R^+_{\circ}\to \R^+$ such that $\rho_0\in L^1(\R^+_{\circ})$.
For all $N,k\in \N$, let
$$\overline \rho_{N,k} = N\int_{(k-1)/N}^{k/N} \rho_0(x) dx.$$
Define
a sequence of `local equilibrium' measures $\left\{\mu^N\right\}_{N\in\N}$ corresponding to $\rho_0$:
\begin{enumerate}
\item For all $N\in \N$ and $\eta\in\Omega$, $\mu^N (\eta) = \prod_{k=1} \mu^N_k (\eta(k))$ with $\mu^N_k$ Geometric distributions with parameter $\theta_{N,k}$.
\item $ \lim_{N\to\infty} \dfrac 1N \sum_{k=1}^{\infty} |N_{\beta} \rho_{N,k} - \overline \rho_{N,k}| =0$
where $\rho_{N,k}=\dfrac{\theta_{N,k}}{1- \theta_{N,k}}$ is the mean of $\mu^N_k$.
\item $\mu^N$ is stochastically bounded by $\Pam_{c,N}$ for some $0\leq c<c_0$.
\end{enumerate}
We note that the last condition, given that the marginals of $\mu^N$ are Geometric, is equivalent to $\theta_{N,k}\leq \theta_{k,c}=c \theta_k=ce^{-\beta \Eb_k - k/N}$.

As might be suspected, given the family of profiles $\{\phi_c\}$ are the static limits when the process is started from $\{\Pam_{c,N}\}$ (Proposition \ref{prop: static}), we show in Lemma \ref{invariant measures are local equli}, that the invariant measures $\Pam_{c,N}$, for $0\leq c <c_0$, are local equilibrium measures with $\theta_{N,k} \equiv \theta_{k,c}$ and $\rho_0=\phi_c$.

We now specify a more general class of initial measures $\nu^N$, namely those which satisfy the following condition.  In Proposition \ref{prop: mu N in nu N}, we verify that the local equilibria $\mu^N$ are in fact explicit members of this class.

\begin{Cond} \label {condition 1}
For $N\in \N$, let $\nu^N$ be a sequence of probability measures on $\Omega$.
\begin{enumerate}
\item
Suppose $\rho_0\in L^1(\R^+)$, and for each $N\in \N$, $\nu^N $ is a product measure, $\nu^N(\eta) = \prod_{k=1} \nu^N_k (\eta(k))$ such that marginals
$\nu^N_k$ have mean $m_{N,k}$ 
 where
\begin{equation*}
\lim_{N\to\infty} \dfrac 1N \sum_{k=1}^{\infty} | N_{\beta} m_{N,k} - \overline \rho_{N,k}| =0.
\end{equation*}

\item
We have $\nu^N$ is stochastically bounded by $\Pam_{c,N}$ for a $0\leq c< c_0$.
\item
The relative entropy of $\nu^N$ with respect to $\Pam_{c,N}$ is of order $N N_\beta^{-1}$:  Let $f_0 = d\nu^N/d\Pam_{c,N}$.  Then, $H(\nu^N | \Pam_{c,N}):=  \int f_0 \ln f_0 d\Pam_{c,N} =O(NN_{\beta}^{-1})$.

\end{enumerate}
\end{Cond}

\medskip
 When the process starts from $\{\nu^N\}_{N\in \N}$, in the class satisfying Condition \ref{condition 1}, we will denote by $\P_N:= \P_{\nu^N}$ and $\E_N:= \E_{\nu^N}$, the associated process measure and expectation.
Members of this class have the following properties, useful in later arguments:

\begin{enumerate}
\item[$\bullet$] Total bound on the number of particles  (Lemma \ref{lem: total number}):  For $0\leq t\leq T$,
\begin{equation}
\label{total number}
\E_N \sum_{k=1}^\infty \eta_t (k) = O(N N^{-1}_\beta).
\end{equation}
\item[$\bullet$] Variance bound (Lemma \ref{lem: var bound}):  For $0\leq t\leq T$,
\begin{equation}\label{nu N var bound}
 \sum_{k=1}^{\infty} \text{Var}_{\P_N}(\eta_t(k)) = o\big(N^2N^{-2}_{\beta}\big).
\end{equation}
\item[$\bullet$] Site particle bound (Lemma \ref{lem: site particle bound}):  For $0<a<b$ and $0\leq t\leq T$,
\begin{equation}
\label{site_particle}
\sup_N \sup_{aN\leq k\leq bN} \sup_{0\leq t\leq T} N_\beta \E_N\big[ \eta_t(k)\big] <\infty.
\end{equation}

\item[$\bullet$] Initial convergence (Proposition \ref{prop: initial convergence and mass for condition 1}):  For any $G\in C_c^{\infty}(\R^+_{\circ})$, and $\delta>0$,
\begin{equation} \label {Initial convergence nu N}
\lim_{N\to \infty} \nu^N
\Big[ \Big|
\frac{N_\beta}{N}\sum_{k=1}^\infty G(k/N)\eta(k)
-
\int_0^{\infty} G(x) \rho_0(x) dx
\Big|
>\delta
 \Big]
 =0.
\end{equation}
\end{enumerate}

By the discussion of attractiveness in Subsection \ref{subsec: attractiveness}, and that $\nu^N\leq \Pam_{c,N}$ and $\Pam_{c,N}$ is an invariant measure, we have 
\begin{equation}
\label{attractiveness}
\E_N\left[f(\eta_t)\right] \leq \E_{\Pam_{c,N}}\left[f(\eta_t)\right] = E_{\Pam_{c,N}}\left[f(\eta)\right],
\end{equation}
for all functions $f$ increasing coordinatewise, and all $t\geq 0$.

In addition, we see that Proposition \ref{prop: static} is a corollary of \eqref{Initial convergence nu N}, since the invariant measures $\Pam_{c,N}$, for $c<c_0$, are local equilibrium measures, and in fact satisfy Condition \ref{condition 1}.

We note, as a consequence of the attractiveness and \eqref{Initial convergence nu N}, that $\int_0^\infty G(x)\rho_0(x)dx \leq \int_0^\infty G(x)\phi_c(x)dx$ for nonnegative $G$, and so necessarily $\rho_0\leq \phi_c$.

\subsection{Results}
Following on the discussion of `static' limits, we now arrive at our main results on the evolution of macroscopic density.  These separate into three limits depending on which of the three regimes are in force.  

Let $\C$ be the space of functions $\rho: [0,T]\times \R^+\mapsto \R^+$  such that the map $t\in [0,T]\mapsto \rho(t,x)dx\in \Mb$ is vaguely continuous; that is, for each $G\in C_c^\infty(\R^+_\circ)$, the map $t\in [0,T]\mapsto \int_0^\infty G(x)\rho(t,x)dx$ is continuous.

A standing assumption in the sequel is that the process $\eta_\cdot$ begins from initial measures $\{\nu^N\}_{N\in \N}$ satisfying Condition \ref{condition 1}.

\begin{Thm} \label {thm: beta 0}
Suppose $\beta=0$ and $\rho_0\in L^1(\R^+)$.
Then, for any $t\geq 0$, test function $G\in C_c^{\infty}(\R^+_{\circ})$, and $\delta>0$,
\begin{equation*}
\lim_{N\to \infty} \P_N
\Big[ \Big|
\langle G, \pi^N_t \rangle
-
\int_0^{\infty} G(x) \rho(t,x) dx
\Big|
>\delta
 \Big]
 =0,
\end{equation*}
where $\rho(t,x)$ is the unique weak solution in the class $\C$ of the equation
\begin{equation} \label {eqn: macro beta 0}
\begin{cases}
\partial_t \rho
= \partial_x^2 \dfrac{\rho}{\rho+1} + \partial_x \dfrac{\rho}{\rho+1}\\
\rho(0,\cdot)=\rho_0(\cdot), \quad
\dsp \int_0^{\infty} \rho(t,x) dx=\int_0^{\infty} \rho_0(x) dx\\
\rho(t,\cdot) \leq \phi_c(\cdot)\in L^1(\R^+) \text{ for all } t\in [0,T].
\end{cases}.
\end{equation}
\end{Thm} 

\begin{Thm} \label {thm: ln k} 
Suppose $\Eb_k \sim \ln k$, $0<\beta<1$ and $\rho_0\in L^1(\R^+)$.
Then, for any $t\geq 0$, test function $G\in C_c^{\infty}(\R^+_{\circ})$, and $\delta>0$,
\begin{equation*}
\lim_{N\to \infty} \P_N
\Big[ \Big| 
\langle G, \pi^N_t \rangle
-
\int_0^{\infty} G(x) \rho(t,x) dx
\Big|
>\delta
 \Big]
 =0,
\end{equation*}
where $\rho(t,x)$ is the unique weak solution  in the class $\C$ of the equation
\begin{equation} \label {eqn: macro ln k} 
\begin{cases}
\partial_t \rho
=
\partial_x^2\rho + \partial_x \Big(   \dfrac{\beta +x}{x} \rho\Big)\\
\rho(0,\cdot)=\rho_0(\cdot), \quad
\dsp \int_0^{\infty} \rho(t,x) dx=\int_0^{\infty} \rho_0(x) dx\\
\rho(t,\cdot) \leq \phi_c(\cdot)\in L^1(\R^+) \text{ for all } t\in [0,T].
\end{cases} .
\end{equation}
\end{Thm} 

\begin{Thm} \label {thm: lnln k} 
Suppose $1\ll \Eb_k \ll \ln k$, $\beta>0$ and $\rho_0\in L^1(\R^+)$.
Then, for any $t\geq 0$, test function $G\in C_c^{\infty}(\R^+_{\circ})$, and $\delta>0$,
\begin{equation*}
\lim_{N\to \infty} \P_N
\Big[ \Big| 
\langle G, \pi^N_t \rangle
-
\int_0^{\infty} G(x) \rho(t,x) dx
\Big|
>\delta
 \Big]
 =0,
\end{equation*}
where $\rho(t,x)$ is the unique weak solution  in the class $\C$ of the equation
\begin{equation} \label {eqn: macro lnln k}
\begin{cases}
\partial_t \rho
=
\partial_x^2\rho + \partial_x \rho \\
\rho(0,\cdot)=\rho_0(\cdot), \quad
\dsp \int_0^{\infty} \rho(t,x) dx=\int_0^{\infty} \rho_0(x) dx\\
\rho(t,\cdot) \leq \phi_c(\cdot)\in L^1(\R^+) \text{ for all }  t\in [0,T].
\end{cases} .
\end{equation}
\end{Thm} 

\medskip 
We now go back to the Young diagrams and explain the results in this context.
For each particle configuration $\eta_t$, the corresponding shape function of the diagram is 
\begin{equation} \label {def of psi N}
\psi_N(t,x) 
=
\dfrac {N_{\beta}}{N} \sum_{k \geq xN} \eta_t(k).
\end{equation}
 The hydrodynamic limits for the diagrams will follow from the hydrodynamic limits of the density profiles.

Let $\Wb$ be the class of continuous functions $\psi:  [0,T] \times \R^+\rightarrow \R^+$ such that, for each $t\in [0,T]$, 
$\psi(t,\cdot):\R^+\rightarrow\R^+$ is absolutely continuous.
\begin{Cor} \label {cor: diagrams}
With respect to the shape functions, the following limits hold.
\begin{enumerate}
\item 
Consider the assumptions of Theorem \ref{thm: beta 0}.
Then, for any $t\geq 0$, test function $G\in C_c^{\infty}(\R^+_{\circ})$, and $\delta>0$,
\begin{equation} \label {eqn: convergence for psi}
\lim_{N\to \infty} \P_N
\Big[ \Big|
\int_0^{\infty} G(x) \psi_N(t,x) dx
-
\int_0^{\infty} G(x) \psi(t,x) dx
\Big|
>\delta
 \Big]
 =0,
\end{equation}
where $\psi(t,x)$ is the unique weak solution in the class $\Wb$ 
of the equation
\begin{equation} \label {eqn: macro beta 0, psi}
\begin{cases}
\partial_t  \psi
= \partial_x \Big(\dfrac{\partial_x \psi}{1- \partial_x \psi} \Big)+ \dfrac{\partial_x \psi}{1-\partial_x \psi}\\
\psi(0,x)=\int_{x}^{\infty} \rho_0(u) du, \quad \lim_{x\to \infty} \psi(t,x) = 0\\
\psi(t,0) = \psi(0,0) , 
\quad
0\leq -\partial_x \psi(t,\cdot) \leq \phi_c(\cdot) \text{ for all } t\in [0,T].
\end{cases}.
\end{equation}
\item 
Consider the assumptions of Theorem \ref{thm: ln k}.
Then, for any $t\geq 0$, test function $G\in C_c^{\infty}(\R^+_{\circ})$, and $\delta>0$,
\begin{equation*}
\lim_{N\to \infty} \P_N
\Big[ \Big|
\int_0^{\infty} G(x) \psi_N(t,x) dx
-
\int_0^{\infty} G(x) \psi(t,x) dx
\Big|
>\delta
 \Big]
 =0,
\end{equation*}
where $\psi(t,x)$ is the unique weak solution in the class $\Wb$ of the equation
\begin{equation}  \label{eqn: macro ln k, psi}
\begin{cases}
\partial_t \psi
=
\partial_x^2 \psi + \dfrac{\beta+x}{x} \partial_x \psi \\
\psi(0,x)=\int_{x}^{\infty} \rho_0(u) du, \quad \lim_{x\to \infty} \psi(t,x) = 0\\
\psi(t,0) = \psi(0,0) ,
\quad
0\leq -\partial_x \psi(t,\cdot) \leq \phi_c(\cdot) \text{ for all } t\in [0,T].
\end{cases}.
\end{equation}
\item 
Consider the assumptions of Theorem \ref{thm: lnln k}.
Then, for any $t\geq 0$, test function $G\in C_c^{\infty}(\R^+_{\circ})$, and $\delta>0$,
\begin{equation*}
\lim_{N\to \infty} \P_N
\Big[ \Big|
\int_0^{\infty} G(x) \psi_N(t,x) dx
-
\int_0^{\infty} G(x) \psi(t,x) dx
\Big|
>\delta
 \Big]
 =0,
\end{equation*}
where $\psi(t,x)$ is the unique weak solution in the class $\Wb$ of the equation
\begin{equation}
\label{eqn: macro lnln k, psi}
\begin{cases}
\partial_t \psi
=
\partial_x^2 \psi + \partial_x  \psi \\
\psi(0,x)=\int_{x}^{\infty} \rho_0(u) du, \quad \lim_{x\to \infty} \psi(t,x) = 0\\
\psi(t,0) = \psi(0,0) ,
\quad
0\leq -\partial_x \psi(t,\cdot) \leq \phi_c(\cdot) \text{ for all } t\in [0,T].
\end{cases}.
\end{equation}
\end{enumerate}
\end{Cor}

\section{Martingale framework} \label{section: martingale}
The proofs of the main results make use of the stochastic differential of $\langle G, \pi^N_t\rangle$, written in terms of certain martingales.  Let $G$ be a compactly supported function on $\R^+ \times  \R^+_\circ$, and let us write $G_t(x) := G(t,x)$, for $t\geq 0$.  Consider the mean zero martingale,
\begin{equation*}
M^{N,G}_t
=
\left\langle G_t,   \pi^N_t \right\rangle
-
\left\langle G_0,  \pi^N_0 \right\rangle
-
\int_0^t \partial_s \left\langle G_s, \pi_s^N\right\rangle + N^2L \left\langle G_s,   \pi^N_s \right\rangle ds.
\end{equation*}
Define the discrete Laplacian $\Delta_{N}$ and discrete gradient $\nabla_N$ as
\begin{equation*}
\begin{split}
\Delta_N G\Big(\frac k N\Big)
:=&
N^2 \Big(G\Big(\frac {k+1}N\Big)+G\Big(\frac {k-1}N\Big)-2G\Big(\frac k N\Big)\Big),\\
\nabla_N G\Big(\frac k N\Big)
:=&
N \Big(G\Big(\frac {k+1}N\Big) - G\Big(\frac k N\Big)\Big).
\end{split}
\end{equation*}

Then, we may compute 
\begin{equation}
\label{gen_comp}
\begin{split}
&L \left\langle G_s,   \pi^N_s \right\rangle = \dfrac1N \sum_{k=2}^{\infty}
\left(
\Delta_{N} G_s\Big(\frac k N \Big) 
+
\dfrac {\lambda_k - 1}{1/N}
\nabla_{N} G_s\Big( \frac k N \Big)
\right)
N_{\beta}\chi_{\{ \eta_s (k)>0\}}\\
&\quad\quad\quad\quad\quad\quad\quad\quad\quad\quad\quad\quad\quad\quad\quad+
N \lambda_1 \nabla_N G_s\Big(\frac1N\Big) N_{\beta}  \chi_{\{ \eta_s(1)>0\}}.
\end{split}
\end{equation}
Since $G_s$ is compactly supported on $\R^+_{\circ}$, we note that the last term vanishes for all $N$ large.  

For later reference, we will call
\begin{equation}
\label{DG}
D^{G,s}_{N,k} := \Delta_{N} G_s\Big(\frac k N \Big) 
+
\dfrac {\lambda_k - 1}{1/N}
\nabla_{N} G_s\Big( \frac k N \Big).
\end{equation}
Define also 
$$\alpha(x,\beta):= \lim_{\substack{N\rightarrow\infty\\k/N\rightarrow x}} \frac{\lambda_k -1}{1/N}.$$
Observing 
$$\lambda_k = e^{-\beta(\Eb_{k+1} - \Eb_k) - 1/N} = e^{-\beta (u(\ln k+1) - u(\ln k)) - 1/N},$$
 we have for all $x>0$ that
\begin{equation}
\label{alpha_D}
\alpha(x,\beta) = \ \left\{\begin{array}{rl}
-1& \ \ {\rm when \ } \beta =0 \ {\rm or \ } 1\ll \Eb_k \ll \ln k \\
-\frac{\beta+x}{x}& \ \ {\rm when \ } \Eb_k \sim \ln k.
\end{array}\right.
\end{equation}
Moreover, for $0<a<b<\infty$, $N$ large, and $aN\leq k\leq bN$, we conclude
\begin{equation}
\label{DG_bound}
\left |D^{G,s}_{N,k}\right | \leq
 2\left(\|\Delta G\|_\infty
 +
 \frac{\beta + b}{a}\|\nabla G\|_\infty
 \right).
\end{equation}

The quadratic variation of $M^{N,G}_t$ is given by
\begin{equation*}
\langle M^{N,G} \rangle_t
=
\int_0^t 
\left\{
N^2L \left(\left\langle G_s,  \pi^N_s \right\rangle ^2 \right)
-
2 \left\langle G_s,  \pi^N_s \right\rangle
 L \left\langle G_s,  \pi^N_s\right\rangle 
\right\}
ds.
\end{equation*}
Straightforward calculation shows that
\begin{equation*} \label {quadratic variation ln k}
\begin{split}
\langle M^{N,G} \rangle_t
=&
\dfrac{N_{\beta}}{N}
\int_0^t 
\left\{
\dfrac1N \sum_{k=1}^{\infty}
\lambda_k 
\left(\nabla_N G_s\left( k/N \right) \right)^2
N_{\beta}\chi_{\{ \eta_s (k)>0\}} 
\right.\\
&\quad\quad\quad\quad\quad\quad\quad\quad\quad
+\left.
\dfrac1N \sum_{k=2}^{\infty}
\left(\nabla_N G_s\left( k/N \right) \right)^2
N_\beta\chi_{\{ \eta_s (k)>0\}}
\right\}
ds.
\end{split}
\end{equation*}

An useful bound on this variation is as follows.  Recall the estimates on $N_\beta$ (cf. \eqref{Nbeta eqn}).

\begin{Lem} \label{lem: martingale bounds} For smooth $G$ with compact support in
$\R^+\times \R^+_o$, there is a constant $C_G$ such that for large $N$,
$$\sup_{0\leq t\leq T}\E_N \langle M^{N,G}\rangle_t \leq C_G T N_{\beta}N^{-1}.$$
\end{Lem}
\begin{proof}
Suppose that $G_t$ is supported on $[a,b]$
with $0<a<b<\infty$ for all $t$.
For $N$ large, we have
\begin{equation*}
\begin{split}
\E_N\langle M^{N,G} \rangle_t
 =&
N_{\beta}N^{-1}
\E_N
 \Big [
 \dsp\int_0^t 
\dfrac1N \sum_{k=aN}^{bN}
\widehat D^{G,s}_{N,k}
 N_{\beta}\chi_{\{ \eta_s (k)>0\}}
 ds  \Big ]\\
 \leq &
 C^1_G
N_{\beta}N^{-1}
\E_N
 \Big [
 \dsp\int_0^t 
\dfrac1N \sum_{k=aN}^{bN}
 N_{\beta}\chi_{\{ \eta_s (k)>0\}}
 ds  \Big ],
\end{split}
\end{equation*}
where
$\widehat D^{G,s}_{N,k} = \lambda_k \left(\nabla_N G_s\left( k/N \right) \right)^2 + \left(\nabla_N G_s\left( k/N \right) \right)^2$ and
$|\widehat D^{G,s}_{N,k}| \leq C^1_G$.

For the case $\beta=0$, since $N_\beta=1$, we bound $\chi_{\{\eta(k)>0\}}$ by $1$.
Then, $\E_N\langle M^{N,G} \rangle_t \leq C^1_G N^{-1} (b-a) t$, from which the lemma follows.

For the other two cases of $\beta>0$, we bound $\chi_{\{\eta(k)>0\}}$ by $\eta(k)$. Then,
\begin{eqnarray*}
\E_N\langle M^{N,G} \rangle_t
 &\leq &
 C_G^1
N_{\beta}N^{-1}
\E_N
 \Big [
 \dsp\int_0^t 
\dfrac1N \sum_{k=1}^{\infty}
 N_{\beta}\eta_s (k)
 ds  \Big ]\\
 &=&
   C_G^1
N_{\beta}N^{-1}
 t\,
\E_N
 \Big [
\dfrac1N \sum_{k=1}^{\infty}
 N_{\beta}\eta_0(k)
  \Big ].
\end{eqnarray*}
We have used that total number of particles is conserved in the last equality. 
Then, by \eqref{total number},
we obtain
$\sup_N \E_N \Big [ \dfrac1N \sum_{k=1}^{\infty} N_{\beta}\eta_0(k)\Big ] < \infty$, thereby finishing the argument.
\end{proof}

\section{Proof outline: Hydrodynamic limits when $\beta=0$} \label{section beta 0}
We give the proof of Theorem \ref{thm: beta 0} in outline form, referring to estimates proved in later sections.
Since $N_\beta =1$ for $\beta=0$, we have
\begin{equation*}
\pi_t^N(dx)
=
\dfrac {1} N 
\sum_{k=1}^{\infty}
\eta_t(k) \delta_{k/N}(dx).
\end{equation*}

We denote by $Q^N$ the probability measure on the trajectory space $D([0,T],\Mb )$ 
governing $\pi_{\cdot}^N$ when the process starts from $\nu^N$.
By Lemma \ref{tightness beta 0} the family of measures $\left\{Q^N\right\}_{N\in \N}$ is tight with respect to the uniform topology, stronger than the Skorokhod topology, and all limit measures are supported on vaguely 
continuous trajectories $\pi_\cdot$, that is for each test function $G\in C^\infty_c(\R^+_{\circ})$, the map $t\mapsto \langle G, \pi_t\rangle$ is continuous.

Let now $Q$ be any limit measure. 
We show that $Q$ is supported on weak solutions to the nonlinear PDE \eqref{eqn: macro beta 0}. 
\vskip .1cm

{\it Step 1.}
Take any smooth $G$ with compact support in $[0,T]\times \R^+_{\circ}$.
To obtain the form of the limit equation, recall the martingale $M^{N,G}_t$ and its quadratic variation $\langle M^{N,G}\rangle_t$ introduced in the last section.

Since $G$ is smooth and with compact support, 
by Lemma \ref{lem: martingale bounds}, we have $\E_N \left( M^{N,G}_T\right)^2 = \E_N \left(\langle M^{N,G} \rangle_T\right)$
vanishes as $N\to \infty$.
Then, by Doob's inequality, for each $\delta>0$,
\begin{eqnarray*}
&&
\P_N
\big(
\sup_{0\leq t\leq T} 
\big| 
\big\langle G_t,  \pi^N_t \big\rangle
-
\big\langle G_0,  \pi^N_0 \big\rangle  
-
 \int_0^t \big( \big\langle \partial_s G_s,  \pi^N_s \big\rangle + N^2L  \big\langle G_s,  \pi^N_s \big\rangle   \big)   ds \big|
>\delta
\big)\\
&& \leq  \dfrac{4}{\delta^2}
 \E_N \big( \big\langle M^{N,G}\big\rangle_T\big) \ \rightarrow \ 0 \ \ {\rm as \ } N\to\infty.
\end{eqnarray*}
Recall the computation of $N^2L  \left\langle G_s,  \pi^N_s \right\rangle$ in \eqref{gen_comp}.
Then,
\begin{eqnarray} \label {beta 0 eqn before replace}
&&
\lim_{N\to \infty}
\P_N
\left(
\sup_{0\leq t\leq T} 
\big| 
\left\langle G_t,  \pi^N_t \right\rangle
-
\left\langle G_0,  \pi^N_0 \right\rangle
-
\int_0^t \left( \left\langle \partial_s G_s,  \pi^N_s \right\rangle 
\vphantom{\dfrac{aa}{NN}}
\right. \right. \\
&&
\left. \left.
+
\dfrac1N \sum_{k=aN}^{bN}
\left(
\Delta_{N} G_s\Big(\frac k N \Big) 
+
\dfrac {\lambda_k - 1}{1/N}
\nabla_{N} G_s\Big( \frac k N \Big)
\right)
\chi_{\{ \eta_s (k)>0\}}
   \right)   ds \big|
>\delta
\right)
=0.
\nonumber
\end{eqnarray}
\vskip .1cm

{\it Step 2.}
We would like to replace the nonlinear term $\chi_{\{\eta_s (k)>0\}}$
by a function of the empirical density of particles within a macroscopically small box.
To be precise,
let $\eta^l (x) = \dfrac{1}{2l+1} \sum_{|y-x|\leq l} \eta(y)$,
that is the average density of particles in the box centered at $x$ with length $2l+1$.  

Recall the coefficient $D^{G,s}_{N,k}$ in \eqref{DG}.
By the triangle inequality, the $1$ and $2$-block estimates (Lemmas \ref{one block lem} and \ref{lem: 2 blocks})
give immediately the following replacement lemma.

\begin{Lem} [Replacement Lemma]
For each $\delta>0$,
\begin{equation*}
\begin{split}
\limsup_{\e\to 0} \limsup_{N\to \infty} 
 \P_N
\Big[
\Big|
\dfrac1N \sum_{aN\leq k \leq bN }
\int_0^T 
D_{N,k}^{G,t}
\Big( \chi_{\{\eta_t (k) >0\}}
-\dfrac{\eta_{t}^{\e N}(k)}{1+ \eta_{t}^{\e N}(k)} \Big) dt
\Big|
 \geq \delta
\Big]
=0.
\end{split}
\end{equation*}
\end{Lem}

\vskip .1cm
{\it Step 3.}
For each $\e>0$, take $\iota_\e = (2\e)^{-1} \chi_{[-\e,\e]}$.
The average density $\eta_{t}^{\e N}(k)$ is written as a function of the empirical measure $\pi_{N^2 t}$ 
\begin{equation*}
\eta_{t}^{\e N}(k)
=
\dfrac{2\e N}{2\e N+1}
\langle \iota_{\e}(\cdot -  k/N),  \pi^N_t) \rangle.
\end{equation*}
Also, as $\lambda_k = e^{-1/N}$ when $\beta=0$, we have $N(\lambda_k -1)\sim -1$ (cf. \eqref{alpha_D}).  

Then, we get from \eqref{beta 0 eqn before replace}, noting the form of $D^{G,s}_{N,k}$, that
\begin{equation*}
\begin{split}
&
\limsup_{\e \to 0}
\limsup_{N\to \infty}
Q^N
\left(
\Big | 
\left\langle G_T,  \pi^N_T \right\rangle
-
\left\langle G_0,  \pi^N_0 \right\rangle
-
\int_0^T \left( \left\langle \partial_s G_s,  \pi^N_s \right\rangle 
\vphantom{\dfrac{aa}{NN}}
\right. \right. \\
&
\left. \left.
+
\dfrac1N \sum_{k=aN}^{bN}
\left(
\Delta G_s\Big(\frac k N \Big) 
-
\nabla G_s\Big( \frac k N \Big)
\right)
\dfrac{\langle \iota_{\e}(\cdot -  k/N),  \pi^N_s) \rangle}
{ \langle \iota_{\e}(\cdot -  k/N),  \pi^N_s) \rangle  + 1}
   \right)   ds \Big|
>\delta
\right)
=0.
\end{split}
\end{equation*}
Notice that we replaced $\nabla_N$ and $\Delta_N$
by $\nabla$ and $\Delta$,
respectively. 

The error in replacing the Riemann sum by an integral is $o(1)$.  We get
\begin{equation}
\label{step3 beta 0 eqn}
\begin{split}
&
\limsup_{\e \to 0}
\limsup_{N\to \infty}
Q^N
\left(
\Big | 
\left\langle G_T,  \pi^N_T \right\rangle
-
\left\langle G_0,  \pi^N_0 \right\rangle
-
\int_0^T \left( \left\langle \partial_s G_s,  \pi^N_s \right\rangle ds
\vphantom{\dfrac{aa}{NN}}
\right. \right. \\
&
\left. \left.
+\int_0^T
\int_0^{\infty}
\left(
\Delta G_s\left(x \right) 
-
\nabla G_s\left( x \right)
\right)
\dfrac{\langle \iota_{\e}(\cdot - x),  \pi^N_s) \rangle}
{ \langle \iota_{\e}(\cdot -  x),  \pi^N_s) \rangle  + 1}
   dx \right)  ds \Big |
>\delta
\right)
=0.
\end{split}
\end{equation}

Taking $N\to \infty$, along a subsequence, as the set of trajectories in \eqref{step3 beta 0 eqn} is open with respect to the uniform topology, we obtain
\begin{equation*}
\begin{split}
&
\limsup_{\e \to 0}
Q
\left(
\Big| 
\left\langle G_T,  \pi_T \right\rangle
-
\left\langle G_0,  \pi_0 \right\rangle
-
\int_0^T \left( \left\langle \partial_s G_s,  \pi_s \right\rangle 
\vphantom{\dfrac{aa}{NN}}
\right. \right. \\
&
\left. \left.
+
\int_0^{\infty}
\left(
\Delta G_s\left(x \right) 
-
\nabla G_s\left( x \right)
\right)
\dfrac{\langle \iota_{\e}(\cdot - x),  \pi_s) \rangle}
{ \langle \iota_{\e}(\cdot -  x),  \pi_s \rangle  + 1}
  dx \right)ds
  \Big |
>\delta
\right)
=0.
\end{split}
\end{equation*}

\vskip .1cm
{\it Step 4.} We show in
Lemma \ref{lem: rho < phi c} that $Q$ is supported on trajectories $\pi_s(dx) =\rho(s,x) dx $ where $\rho\in L^1([0,T]\times \R)$.
To replace $\langle \iota_{\e}(\cdot - x),  \pi_s \rangle$
by $\rho(s,x)$, it is enough to show, for all $\delta>0$, that
\begin{equation*}
\begin{split}
&
\limsup_{\e \to 0}
Q
\left(
\left| 
\int_0^T 
\int_0^{\infty}
D_{G,s}
\left(
\dfrac{\langle \iota_{\e}(\cdot - x),  \pi_s \rangle}
{ \langle \iota_{\e}(\cdot -  x),  \pi_s \rangle  + 1}
-
\dfrac{\rho(s,x)}{1+\rho(s,x)}
\right) dx 
ds \right|
>\delta
\right)
=0.
\end{split}
\end{equation*}
where $D_{G,s} = \Delta G_s\left(x \right) + \nabla G_s\left( x \right)$.
In fact, considering the Lebesgue points of $\rho$, almost surely with respect to $Q$,
\begin{equation*}
\lim_{\e\to 0}
\int_0^T 
\int_0^{\infty}
D_{G,s}
\dfrac{\langle \iota_{\e}(\cdot - x),  \pi_s \rangle}
{ \langle \iota_{\e}(\cdot -  x),  \pi_s \rangle  + 1}
 dx  ds 
 =
 \int_0^T 
\int_0^{\infty}
D_{G,s}
 \dfrac{\rho(s,x)}{1+\rho(s,x)}
 dx  ds .
\end{equation*}
Now, we have
\begin{equation*}
\begin{split}
&
Q
\left(
\Big | 
\left\langle G_T,  \rho(T,x) \right\rangle
-
\left\langle G_0,  \rho(0,x) \right\rangle
-
\int_0^T \left( \left\langle \partial_s G_s,  \rho(s,x) \right\rangle 
\vphantom{\dfrac{aa}{NN}}
\right. \right. \\
&
\left. \left.
+
\int_0^{\infty}
\left(
\Delta G_s\left(x \right) 
-
\nabla G_s\left( x \right)
\right)
\dfrac{ \rho(s,x) }
{\rho(s,x) + 1}
 dx  \right)  ds \Big|
   =0
\right)
=1.
\end{split}
\end{equation*}

\vskip .1cm
{\it Step 5.}
Hence, each $\rho(t,x)$ solves weakly the 
equation $ \partial_t \rho = \partial_x^2 \dfrac{\rho}{\rho+1} + \partial_x \dfrac{\rho}{\rho+1}$.  As we have already remarked that $Q$ is supported on vaguely continuous trajectories (Lemma \ref{tightness beta 0}), we have that $\rho$ belongs to $\C$.
 
We claim now that $\rho(t,x)$ satisfies the initial value problem  \eqref{eqn: macro beta 0}:  Indeed,
the initial condition $\rho(0,x) = \rho_0(x)$ holds by \eqref{Initial convergence nu N}.
By Lemma \ref{lem: rho < phi c}, we have $\rho(t,x) \leq \phi_c(x)$ for all $0\leq t\leq T$.
The conservation of mass $\int_0^{\infty} \rho(t,x)dx = \int_0^{\infty} \rho_0(x)dx$ is proved in Lemma \ref{lem mass conserve}.

We show in Subsection \ref{section uniqueness beta 0} that there is at most one weak solution $\rho$ to \eqref{eqn: macro beta 0}, subject to these constraints.
 We conclude then that the sequence of $Q^N$ converges weakly to the Dirac measure on $\rho(\cdot,x)dx$.
Finally, as $Q^N$ converges to $Q$ with respect to the uniform topology, we have for each $0\leq t\leq T$ that $\langle G,\pi_t^N\rangle$ weakly converges to the constant $\int G(x)\rho(t,x)dx$, and therefore convergence in probability as stated in Theorem \ref{thm: beta 0}.
\qed


\section{Proof outline: Hydrodynamic limits when $\beta>0$}  \label {section beta not 0}
In this section, we sketch a proof of both Theorems \ref{thm: ln k} and \ref{thm: lnln k}, following the the argument for the $\beta=0$ case.

\vskip .1cm
{\it Step 1.}
The replacement lemma we need here is simpler than for the case $\beta=0$, as it relies only on a $1$-block estimate.  Because of the form of the function $N_\beta \chi_{\{\eta_t(k)>0\}}$, from the $1$-block estimate, it is close to $N_\beta \eta^l_t(k)/(1+\eta^l_t(k))$.  However, as $N_\beta \eta^l_t(k)$ is of order $O(1)$, and therefore $\eta^l_t(k) = o(1)$, we may replace $N_\beta \eta^l_t(k)/\big(1+\eta^l_t(k)\big)$ by its linearization $N_\beta \eta^l_t(k)$.  Then, using smoothness of the test function, $\eta^l_t(k)$ may be replaced by $\eta_t(k)$, so that a $2$-blocks estimate is not needed.  Moreover, we see as a consequence that a linear PDE arises in the hydrodynamic limit.

Recall the expression $D^{G,t}_{N,k}$ in \eqref{DG}.

\begin{Lem}[Replacement Lemma]
For each smooth, compactly supported function $G$ on $[0,T]\times \R^+_\circ$, we have 
\begin{equation*}
\limsup_{N\to \infty}
\E_N \left |
\dfrac 1N \sum_{k= aN}^{ bN }
\int_0^T 
D_{N,k}^{G,t}
\left( N_{\beta}\chi_{\{\eta_t (k) >0\}}
-N_{\beta} \eta_t (k) \right) dt
\right |
=
0.
\end{equation*} 
\end{Lem}
\begin{proof}
By smoothness of the test function $G$, it suffices to show
$$\limsup_{l\to \infty} \limsup_{N\to \infty}
\E_N \left |
\dfrac 1N \sum_{k= aN}^{ bN }
\int_0^T 
D_{N,k}^{G,t}
\left( N_{\beta}\chi_{\{\eta_t (k) >0\}}
-N_{\beta} \eta^l_t(k) \right) dt
\right |
=
0,$$
and in turn enough to show that
\begin{equation}\label {eqn: replacement for ln k case}
\limsup_{l\to \infty} \limsup_{N\to \infty}
\sup_{aN\leq k \leq bN }
\E_N \left |
\int_0^T 
D_{N,k}^{G,t}
\left( N_{\beta}\chi_{\{\eta_t (k) >0\}}
-N_{\beta} \eta^l_t(k) \right) dt
\right |
=
0.
\end{equation}
By the $1$-block estimate (Lemma \ref{one block lem}),
\begin{equation*}
\limsup_{l\to \infty}
\limsup_{N\to \infty}
\sup_{aN\leq k \leq bN }
\E_N \left |
\int_0^T D_{N,k}^{G,t}
\left( N_{\beta}\chi_{\{\eta_t (k) >0\}}
-
\dfrac{N_{\beta} \eta^l_t(k)}{1+ \eta^l_t(k)}
 \right) dt
\right |
=
0.
\end{equation*}
Adding and subtracting $N_{\beta} \eta^l_t(k)$, noting the uniform bound on $D^{G,t}_{N,k}$ after \eqref{DG},
\eqref {eqn: replacement for ln k case} will follow if we have
$$\limsup_{l\to \infty}
\limsup_{N\to \infty}
\sup_{aN\leq k \leq bN }
\E_N
\int_0^T 
\left(
\dfrac
{N_{\beta} \left( \eta^l_t(k)\right)^2 }
{1+ \eta^l_t(k)}
 \right) dt
=
0.$$
In fact,
by attractiveness \eqref{attractiveness}, noting that $\Pam_{c,N}$ is an invariant measure,
it will be enough to verify that 
\begin{equation*}
\limsup_{l\to \infty}
\limsup_{N\to \infty}
\sup_{aN\leq k \leq bN }
E_{\Pam_{c,N}}
\left(
\dfrac
{N_{\beta} \left( \eta^l(k)\right)^2 }
{1+ \eta^l(k)}
 \right) 
=
0.
\end{equation*}

To this end, for any $l$, $N$, $aN\leq k \leq bN$, noting that $\Pam_{c,N}$ is a product measure, we have
\begin{eqnarray} 
 \label {eqn: error linear to Nlinear}
&&E_{\Pam_{c,N}}
\left(
\dfrac
{N_{\beta} \left( \eta^l(k)\right)^2 }
{1+ \eta^l(k)}
 \right)
\leq
E_{\Pam_{c,N}}
\left(
N_{\beta} \left( \eta^l(k)\right)^2
 \right) \\
&& =
 \dfrac{N_{\beta}} {(2l +1)^2}
 \sum_{|j-k|\leq l}
 E_{\Pam_{c,N}} \left( \eta(j)\right)^2.
 +
 \dfrac{N_{\beta}}{(2l +1)^2}
  \sum_{\substack{ j\neq m, |j-k|\leq l\\|m-k|\leq l }}
  E_{\Pam_{c,N}}\big(\eta(j)\big)
  E_{\Pam_{c,N}} \big(\eta(m)\big).\nonumber
\end{eqnarray}
Recall, under $\Pam_{c,N}$, that $\left\{\eta(j)\right\}$ is a sequence of Geometric variables 
with parameters $\theta_{j,c}=ce^{-\beta \Eb_j-j/N}$.
We may calculate that \eqref{eqn: error linear to Nlinear} equals
\begin{equation}
\label{step99}
\frac{N_\beta}{(2l+1)^2} \sum_{|j-k|\leq l} \big[\rho^2_{j,c} + \rho_{j,c}\big]
  + \frac{N_\beta}{(2l+1)^2} \sum_{\substack{ j\neq m, |j-k|\leq l\\|m-k|\leq l }}
\rho_{j,c}\rho_{m,c}.
\end{equation}

By the site particle bound
\eqref{site_particle}, we have
$$\sup_{N}\sup_{aN-l\leq j\leq bN+l} N_\beta \rho_{j,c}<\infty.$$
Also, as $\beta> 0$, we have $N_\beta=e^{\beta\Eb_N}\to \infty$.

Hence, we see that \eqref{step99} is of order $O(N_\beta^{-1}l^{-1}+ l^{-1} + N_\beta^{-1})$, which vanishes as $N\to\infty$ and then $l\to \infty$.
\end{proof}

\vskip .1cm
{\it Step 2.}
Now, with the help of this replacement lemma and following Steps 1 and 2 in the proof of Theorem \ref{thm: beta 0}, we readily have 
\begin{equation} \label {ln k eqn before replace}
\begin{split}
&
\lim_{N\to \infty}
Q^N
\left(
| 
\left\langle G_T,  \pi^N_T \right\rangle
-
\left\langle G_0,  \pi^N_0 \right\rangle
-
\int_0^T \left( \left\langle \partial_s G_s,  \pi^N_s \right\rangle 
\vphantom{\dfrac{aa}{NN}}
\right. \right. \\
&
\left. \left.
+
\dfrac1N \sum_{k=aN}^{bN}
\left(
\Delta_{N} G_s\left(\frac k N \right) 
+
\dfrac {\lambda_k - 1}{1/N}
\nabla_{N} G_s\left( \frac k N \right)
\right)
N_{\beta} \eta_s (k)
   \right)   ds |
>\delta
\right)
=0.
\end{split}
\end{equation}

Recall $\alpha(x,\beta)=\lim_{\substack{N\to \infty\\k/N\to x}} \dfrac {\lambda_k - 1}{1/N}$ equals $-(\beta+x)/x$ when $\Eb_k \sim \ln k$ and equals $-1$ when $1\ll \Eb_k \ll \ln k$ (cf. \eqref{alpha_D}).
Then, we may replace $\nabla_N$, $\Delta_N$, and $N(\lambda_k -1)$ by $\nabla$, $\Delta$, and $a(x,\beta)$
respectively, in \eqref{ln k eqn before replace}.  We obtain
\begin{equation*} 
\begin{split}
&
\lim_{N\to \infty}
Q^N
\Big [
\big| 
\left\langle G_T,  \pi^N_T \right\rangle
-
\left\langle G_0,  \pi^N_0 \right\rangle
-
\int_0^T \left( \left\langle \partial_s G_s,  \pi^N_s \right\rangle 
\vphantom{\dfrac{aa}{NN}}
\right.\\
&\ \ \ \ \ \ \ \ \ \ \ \ \ \ \ \ \ \ \ \ \ \ \ \ \ \ +
\left.
\vphantom{\dfrac{aa}{NN}}
\left \langle 
\Delta G_s
+
a(x,\beta)
\nabla G_s , \pi^N_s
\right\rangle
   \right)   ds \big |
>\delta
\Big ]
=0.
\end{split}
\end{equation*}

\vskip .1cm
{\it Step 3.}
Now, the sequence $\{Q^N\}$ is tight with respect to the uniform topology by Lemma \ref{tightness beta 0}.  Let $Q$ be a limit point.  Then,
\begin{equation*}
Q
\Big[
\left\langle G_T,  \pi_T \right\rangle
-
\left\langle G_0,  \pi_0 \right\rangle
-
\int_0^T \left(
 \left\langle \partial_s G_s,  \pi_s \right\rangle 
 \vphantom{\dfrac{aa}{NN}}
\right.
+
\left.
\vphantom{\dfrac{aa}{NN}}
\left \langle 
\Delta G_s
+
a(x,\beta)
\nabla G_s , \pi_s
\right\rangle
   \right)   ds
   =0
\Big ]
=1.
\end{equation*}
Since $Q$ is supported on absolutely continuous trajectories $\pi_t(dx) = \rho(t,x)dx$, where $\rho\in L^1([0,T]\times \R^+)$ 
by Lemma \ref{lem: rho < phi c},
we have that each $\rho(t,x)$ is a weak solution of \eqref {eqn: macro ln k} or \eqref {eqn: macro lnln k}, depending on the choice of energy $\Eb_k$.
Using the uniqueness results when $\beta> 0$ shown in Subsection \ref{section uniqueness beta neq 0}, we now follow exactly Step 5 of the proof given in $\beta=0$ case, to obtain the full statements of Theorems \ref{thm: ln k} and \ref{thm: lnln k}.
 \qed

\section{Proof outline:  Hydrodynamic limits for the diagrams} \label{section: diagrams limit}
In this section, we prove Corollary \ref{cor: diagrams}.
We will only prove the $\beta = 0$ case. The other two cases follow from similar arguments.

\vskip .1cm
{\it Step 1.}
We will assume the hydrodynamic limit result Theorem \ref {thm: beta 0} holds.
 First, we show that we may extend the limit
\begin{equation} \label {convergence of density general g}
\lim_{N\to \infty} \P_N
\Big[ \Big|
\dfrac {1} N 
\sum_{k=1}^{\infty}
g\Big(\dfrac k N\Big) \eta_t(k) 
-
\int_0^{\infty} g(x) \rho(t,x) dx
\Big|
>\delta
 \Big]
 =0
\end{equation}
to all $g\in C^{\infty}(\R^+_{\circ})$ supported on $[a,\infty)$ and satisfying $g(x)=g(b)$ for all $x\geq b$ for some $0<a < b <\infty$.
Indeed, fix such a $g$ and take $g_n\in C_c^{\infty}(\R^+_{\circ})$ such that $g_n = g$ on $(0,n)$.
Then,
\begin{equation*}
\begin{split}
&\P_N
\Big[ \Big|
\dfrac {1} N 
\sum_{k=1}^{\infty}
g\Big(\dfrac k N\Big) \eta_t(k) 
-
\int_0^{\infty} g(x) \rho(t,x) dx
\Big|
>\delta
 \Big]\\
 \leq&
 \P_N
\Big[ \Big|
\dfrac {1} N 
\sum_{k=1}^{\infty}
g_n\Big(\dfrac k N\Big) \eta_t(k) 
-
\int_0^{\infty} g_n(x) \rho(t,x) dx
\Big|
>\dfrac{\delta}2
 \Big]\\
 &+
 \P_N
\Big[ \Big|
\dfrac {1} N 
\sum_{k=1}^{\infty}
\Big(g\Big(\dfrac k N\Big) - g_n\Big(\dfrac k N\Big) \Big)\eta_t(k) 
-
\int_0^{\infty} (g(x) -g_n(x)) \rho(t,x) dx
\Big|
>\dfrac{\delta}2
 \Big]
\end{split}
\end{equation*}
Since $g_n$ is compacted supported, by Theorem \ref {thm: beta 0}, the first term vanishes as $N\to \infty$.

As $\rho\leq \phi_c$ and $\phi_c \in L^1(\R^+)$, for $n$ large enough,
the second term is bounded from above by
\begin{equation*}
\begin{split}
 \P_N
\Big[ \Big|
\dfrac {1} N 
\sum_{k=1}^{\infty}
\Big(g\Big(\dfrac k N\Big) - g_n\Big(\dfrac k N\Big) \Big)\eta_t(k) 
\Big|
>\dfrac{\delta}4
 \Big]
 \leq
  \P_N
\Big[ 
\dfrac {2\|g\|_\infty} N 
\sum_{k=nN}^{\infty}
\eta_t(k) 
>\dfrac{\delta}4
 \Big].
\end{split}
\end{equation*}

By attractiveness \eqref{attractiveness} and the Markov inequality, the right-hand side probability is bounded by
$(8\|g\|_\infty/\delta) N^{-1} \sum_{k\geq n} E_{\Pam_{c,N}}(\eta(k))$.
   By \eqref{mean_var_pam}, we observe $\sum_{k\geq 1} E_{\Pam_{c,N}}(\eta(k)) = O(N)$.  Hence, the above display vanishes as $n\to \infty$ uniformly for $N\geq 1$, and \eqref{convergence of density general g} is proved.

\vskip .1cm
{\it Step 2.}
Define $\psi(t,x)= \int_x^{\infty} \rho(t,x) dx$. 
 Then, $\psi(t,x)$ belongs to $\Wb$ and is the unique weak solution of  \eqref{eqn: macro beta 0, psi} as shown in Subsection \ref{section uniqueness beta 0}. 
Now, fix any $G\in C_c^{\infty}(\R^+_{\circ})$ and define $g(x) = \int_0^x G(u)du$ for all $x\in \R^+_{\circ}$.
By integration by parts, we have $\int_0^{\infty} G(x) \psi(t,x) dx = \int_0^{\infty} g(x) \rho(t,x) dx$.

Recall $\psi_N$ from \eqref{def of psi N}.  Using summation by parts, we have
\begin{equation*}
\begin{split}
 \int_0^{\infty} G(x) \psi_N(t,x) dx 
 =&
 \sum_{k=1}^{\infty}
 \Big[g\Big(\frac k N\Big) - g\Big(\frac {k-1} N\Big)\Big]  \psi_N(t,k/N)\\
 =&
  \dfrac1N \sum_{k=1}^{\infty} g\Big(\dfrac k N\Big)  \eta_t(k).
\end{split}
\end{equation*} 
Then, we obtain \eqref{eqn: convergence for psi} from \eqref{convergence of density general g}
and Corollary \ref{cor: diagrams} is proved.
\qed


\section{Tightness and properties of limit measures} \label {section: tightness}
In this section, we obtain tightness of the family of probability measures $\left\{Q^N\right\}_{N\in \N}$ on the trajectory space $D([0,T],\Mb )$.
Then, we show some properties of the limit measures $Q$.

\subsection{Tightness}
We show that $\{Q^N\}$ is tight with respect to the uniform topology, stronger than the Skorokhod topology on $D([0,T], \Mb)$.

\begin{Lem} \label {tightness beta 0}
$\left\{Q^N\right\}_{N\in \N}$ is relatively compact with respect to the uniform topology.  As a consequence, all limit points $Q$ are supported on vaguely continuous trajectories $\pi$, that is for $G\in C^\infty_c(\R^+_\circ)$ we have $t\in [0,T]\mapsto\langle G, \pi_t\rangle$ is continuous.
\end{Lem}
\begin{proof}
Recall the distance $d$ and space of measures $\Mb$ in the introduction.
To show that $\{Q^N\}$ is relatively compact with respect to uniform topology, 
we show the following items (cf.\,p.\,51 \cite{KL}).
\begin{enumerate} 
\item For each $t\in[0,T]$, $\e>0$, there exists compact set $K_{t,\e}\subset \Mb$ such that
\begin{equation}\label {eqn1 compactness}
\sup_{N} Q^{N} \left[\pi^N_\cdot: \pi^N_t\notin K_{t,\e}\right ] \leq \e.
\end{equation}
\item For every $\e>0$, 
\begin{equation}\label {eqn2 compactness}
\lim_{\gamma \to 0} 
\lim_{N\to \infty}
Q^N \Big [  \pi^N_\cdot : \sup_{|t-s|<\gamma} 
d(\pi^N_t,\pi^N_s ) >\e \Big]
=0.
\end{equation}
\end{enumerate}
We now argue the first condition \eqref{eqn1 compactness}. Indeed, since the dynamics is attractive (cf.\,\eqref{attractiveness}), we have
\begin{equation*}
Q^N\Big[\langle 1, \pi^N_t\rangle > A\Big]
\leq \Pam_{c,N}
\Big[ N^{-1}N_\beta \sum_{k\geq 1} \eta(k)> A\Big]
.
\end{equation*}
Applying Markov's inequality and using the mean particle estimate \eqref{mean_var_pam}, we obtain $$Q^N\left[\langle 1, \pi^N_t\rangle > A\right] \leq \frac{C}{A}$$ for some constant $C$ independent of $N$ and $A$. Notice that the set $\left\{ \mu\in \Mb: \langle 1,\mu\rangle \leq A  \right\}$ is compact in $\Mb$, then the first condition \eqref{eqn1 compactness} is checked by taking $A$ large. 

To show the second condition \eqref {eqn2 compactness}, it is enough to show a counterpart of the condition for the distributions of $\langle G, \pi^N_{\cdot}\rangle$ where $G$ is any smooth test function with compact support in $\R^+_{\circ}$
(cf.\,p.\,54, \cite{KL}).
In other words, we need to show,
 for every $\e>0$, 
\begin{equation}\label {eqn2 compactness with G}
\lim_{\gamma \to 0} 
\lim_{N\to \infty}
Q^N \Big [  \pi^N_\cdot : \sup_{|t-s|<\gamma} \Big| \langle G,  \pi^N_t \rangle  -\langle G, \pi^N_s \rangle\Big| >\e \Big]
=0.
\end{equation}

We now show the condition \eqref{eqn2 compactness with G}.
Since
\begin{equation*}
\left\langle G,  \pi^N_t \right\rangle
=
\left\langle G,  \pi^N_0 \right\rangle
+
\int_0^t N^2L \left\langle G,   \pi^N_s \right\rangle ds
+
M^{N,G}_t
\end{equation*}
we only need to consider the oscillations of $\dsp\int_0^t N^2L \left\langle G,   \pi^N_s \right\rangle ds$
and $M^{N,G}_t$ respectively.

Suppose that $G$ has support $[a,b]$ with $0<a<b<\infty$.  Recall the generator computation \eqref{gen_comp}.
For $N$ large, we have
\begin{equation*}
\begin{split}
&\sup_{|t-s| < \gamma}
 \Big|
 \dsp\int_s^t N^2L \left\langle G,  \pi^N_{\tau} \right\rangle d\tau
     \Big|\\
 =&
 \sup_{|t-s| < \gamma}
 \Big|
 \dsp\int_s^t 
\left\{
\dfrac{N_\beta}N \sum_{k=aN}^{bN}
\Big(
\Delta_{N} G\left(k/N \right) 
+
\dfrac {\lambda_k - 1}{1/N}
\nabla_{N} G\left( k/N \right)
\Big)
\chi_{\{ \eta_{\tau}(k)>0\}}
\right\}  d\tau
     \Big|\\
 \leq &
 C_G
 \sup_{|t-s| < \gamma}
 \dsp\int_s^t 
\left\{
\dfrac{N_\beta}N \sum_{k=aN}^{bN}
\chi_{\{ \eta_{\tau}(k)>0\}}
\right\}  d\tau.
\end{split}
\end{equation*}
When $\beta=0$, we have $N_\beta=1$.  Since $\chi_{\{ \eta(k)>0\}} \leq 1$, then
$
\sup_{|t-s| < \gamma}
 \Big|
 \dsp\int_s^t N^2L \left\langle G,  \pi^N_{\tau} \right\rangle d\tau
     \Big|
 \leq
  C_G (b-a)
\gamma
$
vanishes as $\gamma \to 0$.

For other case $\beta>0$, we bound $\chi_{\{ \eta(k)>0\}}\leq \eta(k)$.  Then, by conservation of mass,
\begin{equation*}
\begin{split}
\sup_{|t-s| < \gamma}
 \left|
 \dsp\int_s^t N^2L \left\langle G,  \pi^N_{\tau} \right\rangle d\tau
     \right|
 \leq
  C_G
 \sup_{|t-s| < \gamma}
 \dsp\int_s^t 
\left\{
\dfrac1N \sum_{k=1}^{\infty}
N_{\beta}\eta_{\tau}(k)
\right\}  d\tau\\
=
  C_G
\gamma 
\dfrac1N \sum_{k=1}^{\infty}
N_{\beta}\eta_0(k).
\end{split}
\end{equation*}
Recall the total expected number of particles is of order $NN_\beta^{-1}$ (cf. \eqref{total number}).
By Markov inequality, 
$
Q^N \Big [ 
\sup_{|t-s|<\gamma} 
 \left|
 \dsp\int_s^t N^2L \left\langle G,  \pi^N_{\tau} \right\rangle d\tau
     \right|
>\e \Big]
\leq
\dfrac{C_G \gamma}{\e}
\E_N \Big(  N^{-1} \sum_{k=1}^{\infty}
N_{\beta}\eta_0(k) \Big), 
$
vanishes as $N\uparrow\infty$ and $\gamma\downarrow 0$.

Next, we treat the martingale $M^{N,G}_t$.
Trivially, by $\big| M^{N,G}_t - M^{N,G}_s  \big| \leq  \big| M^{N,G}_t \big| + \big| M^{N,G}_s  \big|$, we have
$$
\P_N
\Big(
\sup_{|t-s|<\gamma} 
\big| 
M^{N,G}_t - M^{N,G}_s  
\big|
>\e
\Big)
\leq
2\P_N
\Big(
\sup_{0\leq t \leq T} 
\big| 
M^{N,G}_t  
\big|
>\e /2
\Big)
$$
which, by Chebychev and Doob's inequality, is bounded by
$$
\dfrac{8}{\e^2}
\E_N
\Big[
\Big(
\sup_{0\leq t \leq T} 
\big| 
M^{N,G}_t  
\big|
\Big)^2
\Big]
\leq
\dfrac{32}{\e^2}
\E_N
\Big[
\big(
M^{N,G}_T
\big)^2
\Big]
= 
\dfrac{32}{\e^2}\E_N
\langle
M^{N,G}
\rangle_T .
$$

Now, by Lemma \ref{lem: martingale bounds}, 
 $\langle M^{N,G}\rangle_T$ is of order $O(N_{\beta}N^{-1}) = o(1)$ (cf. \eqref{Nbeta eqn}). Then, we conclude
\begin{equation*}
\begin{split}
\lim_{\gamma \to 0}
\lim_{N\to \infty}
\P_N
\Big(
\sup_{|t-s|<\gamma} 
\left| 
M^{N,G}_t - M^{N,G}_s  
\right|
>\e
\Big)
=0.
\qedhere
\end{split}
\end{equation*}
\end{proof}


\subsection{Properties of limit measures.}
By Lemma \ref{tightness beta 0}, the sequence $\left\{ Q^N\right\}$ is relatively compact with respect to the uniform topology.  Consider any convergent subsequence of $Q^N$ and relabel so that $Q^N \Rightarrow Q$.  

We now show some properties of $Q$.

\begin{Lem}\label{lem: rho < phi c}
$Q$ is supported on absolutely continuous trajectories whose densities satisfy certain bounds:
\begin{equation*} 
Q \left[ \pi_\cdot :
\pi_t(dx) = \pi(t,x)dx \text{ with }
\pi(t,\cdot)\leq \phi_c(\cdot)
 \text{ for all } 0\leq t\leq T
\right]
=
1.
\end{equation*}
\end{Lem}
\begin{proof}
Let $C_c^+(\R^+_{\circ})$ be the space of nonnegative continuous functions with compact support on $\R^+_{\circ}$
and we equip it with the topology of uniform convergence on compact sets. 
Take $\left\{G_n\right\}_{n\in \N}$ be a dense sequence of $C_c^+(\R^+_{\circ})$.
The lemma is equivalent to
$$
Q \Big[
\langle G_n, \pi_t \rangle
\leq 
\int_{\R^+_{\circ}} G_n(x)\phi_c(x) dx
 \text{ for all } 0\leq t\leq T \text{ and } n\in \N
\Big]
=
1.
$$

Fix a dense set $\{t_k\}_{k\in \N}$ of $[0,T]$. 
Assume for this moment, for any $n,k\in \N$ and $\e>0$, that
\begin{equation} \label{eqn: rho < phi c for G}
Q \Big[
\langle G_n, \pi_{t_k} \rangle
\leq 
\int_{\R^+_{\circ}} G_n(x)\phi_c(x) dx + \e
\Big]
=
1.
\end{equation}
Since $Q$ is supported on vaguely continuous trajectories by Lemma \ref{tightness beta 0}, we obtain for all $\e>0$,
$$
Q \Big[
\langle G_n, \pi_{t} \rangle
\leq 
\int_{\R^+_{\circ}} G_n(x)\phi_c(x) dx + \e
 \text{ for all } 0\leq t\leq T, n\in \N 
\Big]
=
1.
$$
Then, we conclude the lemma by taking $\e\to 0$.
 
It remains to prove \eqref{eqn: rho < phi c for G}. Fix $k,n$, and $\e$ and observe
\begin{equation*}
\begin{split}
Q^N \Big[
\langle G_n, \pi_{t_k}^N \rangle
\leq 
\int_{\R^+_{\circ}} G_n\phi_c dx + \e
\Big]
=&
\P_N \Big[
 \frac{N_\beta}{N}\sum_{j=1}^\infty G_n(j/N)\eta_{t_k}(j)
\leq 
\int_{\R^+_{\circ}} G_n\phi_c dx + \e
\Big]
\end{split}
\end{equation*}
By attractiveness (cf. Subsection \ref{subsec: attractiveness}) and the assumption $\nu^N \leq \Pam_{c,N}$, the above display is bounded from below by
$$
\Pam_{c,N} \Big[
 \frac{N_\beta}{N}\sum_{j=1}^\infty G_n(j/N)\eta(j)
\leq 
\int_{\R^+_{\circ}} G_n\phi_c dx + \e
\Big],
$$
which approaches $1$ as $N\to \infty $ by Proposition \ref {prop: static}.
Then, we have
$$
\limsup_{N\to \infty} Q^N \left[
\langle G_n, \pi_t^N \rangle
\leq 
\int_{\R^+_{\circ}} G_n\phi_c dx + \e
\right]
=
1.
$$
As compactness of $\{Q^N\}$ was shown in the uniform topology in Lemma \ref{tightness beta 0}, the distribution of $\langle G_n,\pi^N_t\rangle$ under $Q^N$ converges weakly to $\langle G_n,\pi_t\rangle$ under $Q$. 
Hence, \eqref{eqn: rho < phi c for G} follows.
\end{proof}

\begin{Lem}  \label {lem mass conserve}
 $Q$ is supported on trajectories with constant total mass:
\begin{equation*}
Q \Big[
\pi_\cdot:
 \langle 1,\pi_t\rangle 
 = \dsp \int_0^\infty \rho_0 dx 
 \text{ for all } 0\leq t\leq T
\Big]
=
1.
\end{equation*}
\end{Lem}
\begin{proof} 
Fix a dense set $\{t_k\}_{k\in \N}$ of $[0,T]$.
By compactness in the uniform topology, we have that as $N\to \infty$, the distribution of $\pi^N_t$ under $Q^N$ converges weakly to $\pi_t$ under $Q$. 
We will show that there exist an increasing sequence of $\left\{G_n\right\}_{n\geq 1} \subset C_c(\R^+_{\circ})$ such that
  $\dsp \lim_{n\to \infty} G_n (x)= 1$ and for all $n,k$, 
\begin{equation} \label {case beta=0 deviation of mass}
\liminf_{N\to \infty} Q^N
\Big[
\Big| \langle G_n,\pi^N_{t_k}\rangle - \dsp \int_0^\infty  \rho_0 dx  \Big|
>\dfrac 1n
\Big]
=0.
\end{equation}

Since $Q^N$ converges to $Q$ with respect to the uniform topolgy (cf. Lemma \ref{tightness beta 0}), we have $\pi^N_{t_k}$ converges weakly to $\pi_{t_k}$.  Then, assuming \eqref{case beta=0 deviation of mass}, we conclude for all $n,k$ that
\begin{equation*}
Q\Big[
\Big| \langle G_n,\pi_{t_k}\rangle - \dsp \int_0^\infty \rho_0 dx  \Big|
>\dfrac 1n
\Big]
=0,
\end{equation*}
and therefore
\begin{equation*}
Q\Big[
\Big| \langle G_n,\pi_{t_k} \rangle - \dsp \int_0^\infty \rho_0 dx  \Big|
\leq \dfrac 1n, \text{ for all } n,k
\Big]
=1.
\end{equation*}
Since also $Q$ is supported on vaguely continuous $\pi_\cdot$, we have
\begin{equation*}
Q\Big[
\Big| \langle G_n,\pi_t \rangle - \dsp \int_0^\infty \rho_0 dx  \Big|
\leq \dfrac 1n, \text{ for all } n,0\leq t\leq T
\Big]
=1,
\end{equation*}
which clearly implies the lemma.

Now, we focus on proof of \eqref{case beta=0 deviation of mass}.
For $G\geq 0$
\begin{equation} \label {eqn: mass bound beta not 0}
\begin{split}
&Q^N
\Big[
\Big| \langle G,\pi_{t_k}\rangle - \dsp \int_0^\infty \rho_0 dx  \Big|
>\dfrac 1n
\Big]\\
\leq&
Q^N
\Big[
\langle 1- G,\pi_{t_k}\rangle
>\dfrac 1{2n}
\Big]
+
Q^N
\Big[
\Big| \langle 1,\pi_{0}\rangle - \dsp \int_0^\infty \rho_0 dx  \Big|
>\dfrac 1{2n}
\Big].
\end{split}
\end{equation}
By \eqref{nu N var bound}, the variance 
$\lim_{N\to \infty} \dfrac{N_\beta^2}{N^2} \sum_{k=1}^{\infty} \text{Var}_{\nu^N}(\eta(k)) =0$.  Also, by part (1) of Condition \ref{condition 1}, $\lim_{N\to \infty}\frac{1}{N}\sum_{k\geq 1}\big|N_\beta m_{N,k} - \overline{\rho}_{N,k}\big| = 0$. Therefore, by adding and substracting the mean $m_{N,k}$ inside the absolute value, the second term on the right-hand side of \eqref{eqn: mass bound beta not 0} vanishes.

We now specify $G_n\in C_c(\R^+_{\circ})$ as follows:
\begin{equation*}
0\leq G_n \leq 1,
\quad
G_n =1 \text{ on } [a,b]
\text{ where }
\int_{(0,a) \cup (b,\infty)} \phi_c dx < \dfrac 1 {3n^2}. 
\end{equation*} 
Since $\nu^N\leq \Pam_{c,N}$, by attractiveness (cf. Subsection \ref{subsec: attractiveness}), for each $t\in [0,T]$, we have
\begin{equation} 
\label{9.1_eq}
Q^N
\Big[
\langle 1- G_n,\pi_{t_k}\rangle
>\dfrac 1{2n}
\Big]
\leq
\Pam_{c,N}\Big[\frac{1}{N} \sum_{\frac kN <a \text{ or } >b}
N_{\beta}\eta(k)
>\dfrac 1{2n}
\Big].
\end{equation}
Recall that $\rho_{k,c} = E_{\Pam_{c,N}}\eta(k)$ (cf. \eqref{rho k c}).
In Lemma \ref{invariant measures are local equli}, it is shown that $\frac{1}{N}\sum_{k\geq 1} \big| N_\beta\rho_{k,c} - \overline{\rho}_{N,k}\big|$, where $\overline{\rho}_{N,k} \sim \phi_c(k/N)$, vanishes as $N\rightarrow\infty$.   
 Note also that
$\int_{(0,a) \cup (b,\infty)} \phi_c dx < 1/(3n^{2})<2/n$, for all $n\geq 1$.  Then, by subtracting and adding the mean $N_\beta \rho_{k,c}$, we conclude by Markov inequality and straightforward manipulation that \eqref{9.1_eq} vanishes as $N\rightarrow \infty$.
\end{proof}


\section{$1$- and $2$-blocks estimates} \label {section: 1 and 2 blocks}

In this section, we prove the $1$- and $2$-block estimate.
The statement and proof for the $1$-block estimate is written for all three cases of $\beta$ and $\Eb_k$,
while the $2$-block estimate assumes $\beta =0$. 
In passing, although it is not consequential in this work, we remark that the $2$-block estimate may not hold for the other cases.


The plan is now to show in the succeeding subsections, a spectral gap bound, and then the $1$ and $2$-block estimates.

\subsection{Spectral gap bound for $1$-block estimate}
We obtain now a spectral gap bound to prepare for the $1$-block estimate.
Define, for $k,l\geq 1$ such that $k-l\geq 1$, the set $\L_{k,l}=\left\{k-l,k-l+1,\ldots,k+l\right\}\subset \N$.   Recall that $\theta_k = e^{-\beta \Eb_k -k/N}$ and $\lambda_k =\dfrac{\theta_{k+1}}{\theta_k}$ (cf. \eqref{eqn: lambda k}).  Consider the process restricted to $\L_{k,l}$ generated by
$L_{k,l}$ where
\begin{equation*}  \label {eqn: L_k l}
\begin{split}
L_{k,l}f(\eta)
=&
\sum_{x,x+1\in \L_{k,l}}
\left\{
\lambda_k
\left [  f\left(\eta^{x,x+1} \right)  - f(\eta)   \right] \chi_{\{\eta(x) > 0\}}\right.\\
&
\quad \quad \quad \quad \quad \quad \left. +
\left [  f\left(\eta^{x+1,x}  \right)   - f(\eta) \right ] \chi_{\{\eta(x+1) > 0\}}
\right\}.
\end{split}
\end{equation*}

We will obtain the spectral gap estimate by showing a Poincar\'e inequality.  To state this bound, we need a few more definitions.
With respect to product measure $\mu:=\Pam_{c,N}$, let $\mu_{k,l}$ be its restriction to
  $\Omega_{l,k} = \left\{0,1,2,\ldots \right\}^{\L_{k,l}}$, 
that is 
\begin{equation}
\label{doublestar}
\mu_{k,l} (\eta)= \prod_{x\in \L_{k,l}} (1-\theta_{x,c})\theta_{x,c}^{\eta(x)}, \ \ {\rm where \ }
\theta_{x,c} = ce^{-\beta\Eb_x - x/N}.
\end{equation}
Let $\mu_{k,l,j}$ be the associated canonical measure on $\Omega_{k,l,j} = \{\eta \in \Omega_{l,k}: \sum_{x\in \L_{k,l} }\eta(x) = j\}$, that is $\mu_{k,l}$ is conditioned so that there are exactly $j$ particles counted in $\Omega_{k,l}$.
   
The corresponding Dirichlet form is written
\begin{equation} \label {d form asymm}
\begin{split}
E_{\mu_{k,l,j}}
\left[ f(-L_{k,l} f) \right]
=
\sum_{x,x+1\in \L_{k,l}}
 E_{\mu_{k,l,j}}
 \left[ 
 \chi_{\{\eta(x+1) > 0\}}
\left [  f\left(\eta^{x+1,x}  \right)   - f(\eta) \right ]^2
 \right].
\end{split}
\end{equation}

The primary method will be to compare with the spectral gap for the standard translation-invariant localized process.  
Consider the generator $L_l$ on $\Omega_{l,k}$ given by
\begin{equation}
\label{onestar}
\begin{split}
L_lf(\eta)
=&
\dfrac12
\sum_{x,x+1\in \L_{k,l}}
\left\{
\left [  f\left(\eta^{x,x+1} \right)  - f(\eta)   \right] \chi_{\{\eta(x) > 0\}} \right. \\
&\quad \quad \quad \quad \quad \quad +
\left.
\left [  f\left(\eta^{x+1,x}  \right)   - f(\eta) \right ] \chi_{\{\eta(x+1) > 0\}}
\right\}.
\end{split}
\end{equation}
Let $\nu_\rho$ be the product measure on $\Omega$ with common Geometric marginal on each site $k\in \N$ with mean $\rho$,
and let $\nu^\rho_{l}$ be its restriction to $\Omega_{k,l}$. 

Consider $\nu_{l,j}$, the associated canonical measure on $\Omega_{k,l,j}$, with respect to $j$ particles in $\L_{k,l}$, which does not depend on $\rho$.  It is well-known that $\nu^\rho_l$ and $\nu_{l,j}$ are both invariant measures with respect to the localized $L_l$ (cf. \cite{A}).
The corresponding Dirichlet form is given by
\begin{equation} \label {d form symm}
\begin{split}
E_{\nu_{l,j}}
\left[ f(-L_l f) \right]
=
\dfrac12
\sum_{x,x+1\in \L_{k,l}}
E_{\nu_{l,j}}
 \left[ 
 \chi_{\{\eta(x+1) > 0\}}
\left [  f\left(\eta^{x+1,x}  \right)   - f(\eta) \right ]^2
 \right].
\end{split}
\end{equation}

Finally, let $x_1=\argmax_{x\in \L_{k,l}} \Eb_x$
and $x_2= \argmin_{x\in \L_{k,l}}\Eb_x$.  Also, for convenience, let $\e = e^{-1/N}$.

\begin{Lem} \label {spectral gap ln k}
We have the following estimates:

\begin{enumerate}
\item Uniform bound:  For all $\eta \in \Omega_{k,l,j}$, we have
\begin{equation} \label {eqn: RN for klj ln k case}
r_{k,l,\e}^{-1}
\leq
\dfrac{\mu_{k,l,j} (\eta) }{\nu_{l,j}(\eta)}
\leq
r_{k,l,\e}
\end{equation}
where 
$
r_{k,l,\e} := 
\left(
 \dfrac
 {1- c e^{-\beta\Eb_{x_1}}  \e^{k+l}}
{1- c e^{-\beta\Eb_{x_2}}  \e^{k-l}}
\right) ^{2l+1} 
\left( e^{-\beta (\Eb_{x_2} -\Eb_{x_1} )}  \e^{-2l}\right)^{j}  .
$
\item Poincar\'e inequality:  We have
\begin{equation*}
{\rm Var}_{\mu_{k,l,j}} (f)
\leq
 C_{k,l,j} 
 E_{\mu_{k,l,j}}
 \left[ f(-L_{k,l}f) \right]
\end{equation*}
where 
$
 C_{k,l,j}
:=
\dfrac C2
(2l+1)^2 \Big(1+\dfrac{j}{2l+1}\Big)^2 r_{k,l,\e}^2
$ bounds the inverse of the spectral gap of $-L_{k,l}$ on $\Omega_{k,l,j}$
and $C$ is an universal constant.

\item For each $0<a<b<\infty$, $l$ and $j$, we have
$$\lim_{N\uparrow\infty} \sup_{aN\leq k\leq bN} r_{k,l,\e} = 1,$$
and hence $\sup_{N\geq 1}\sup_{aN\leq k\leq bN} C_{k,l,j}<\infty$.
\end{enumerate}
\end{Lem}
\begin{proof}
First, the spectral gap for one dimensional localized symmetric zero range process with rate function $\chi_{\{\cdot >0\}}$ is well known (cf. \cite{LSV}): 
For all $j$, with respect to an universal constant $C$,
\begin{equation}
\label {eqn: spetral gap for nu}
{\rm Var}_{\nu_{l,j}} (f)
\leq
C (2l+1)^2 \Big(1+\dfrac{j}{2l+1}\Big)^2
E_{\nu_{l,j}}
 [ f(-L_{l}f) ].
\end{equation}
Therefore, the inverse of the spectral gap is bounded below by $\Big[C (2l+1)^2 \Big(1+\dfrac{j}{2l+1}\Big)^2\Big]^{-1}$.

To get an estimate with respect to $-L_{k,l}$, we will compare $\mu_{k,l,j}$ with $\nu_{l,j}$.  The canonical measure $\nu_{l,j}$ is the measure $\nu_\rho$ conditioned on $j$ particles in $\L_{k,l}$ for any $\rho$.  It will be convenient now to choose $\rho$ such that $\dfrac{\rho}{1+\rho} = \e$, that is, $\e$ is the common parameter of the Geometric marginals  of $\nu_\rho$.

For $\eta\in \Omega_{k,l,j}$, we have
\begin{equation*}
\dfrac{\mu_{k,l,j} (\eta) }{\nu_{l,j}(\eta)}
=
\dfrac{\mu_{k,l} (\eta) }{\nu^\rho_{l}(\eta)}
\dfrac{\nu^\rho_{l}(\Omega_{k,l,j})}{\mu_{k,l} (\Omega_{k,l,j})}.
\end{equation*}
Since $\mu_{k,l}$ (cf. \eqref{doublestar}) and $\nu^\rho_{l}$ are product measures,
\begin{equation}
\label{mu nu prod}
\dfrac{\mu_{k,l} (\eta) }{\nu^\rho_{l}(\eta)}
=
\dfrac  {\prod_{x\in \L_{k,l}} (1-\theta_{x,c}) \theta_{x,c}^{\eta(x)}}  
 {\prod_{x\in \L_{k,l}} (1-\e) \e ^{\eta(x)}}.
\end{equation}
Now, for $\eta\in \Omega_{k,l,j}$, recalling the definitions of $x_1$ and $x_2$ given above, we have
\begin{equation*}
\begin{split}
&\left( 1-c e^{-\beta\Eb_{x_2}} \e^{k-l} \right)^{2l+1} \left( c e^{-\beta\Eb_{x_1}} \e^{k+l}  \right)^j\\
&\quad \quad\quad \quad\quad \quad \leq 
\mu_{k,l} (\eta)
\leq
\left( 1- c e^{-\beta\Eb_{x_1}} \e^{k+l} \right)^{2l+1} 
\left( c e^{-\beta\Eb_{x_2}} \e^{k-l}  \right)^j.
\end{split}
\end{equation*}

Inputting into \eqref{mu nu prod}, we obtain 
\begin{equation*}
\dfrac  {( 1 - c e^{-\beta\Eb_{x_2}} \e^{k-l} )^{2l+1} c^j e^{-\beta\Eb_{x_1} j } \e^{(k+l)j}  }   { (1-\e)^{2l+1}\e^j}
\leq
\dfrac{\mu_{k,l} (\eta) }{\nu^\rho_{l}(\eta)}
\leq
\dfrac  {(1- c e^{-\beta\Eb_{x_1}} \e^{k+l})^{2l+1} c^j e^{-\beta\Eb_{x_2} j } \e^{(k-l)j} }   { (1-\e)^{2l+1}\e^j}. 
\end{equation*}
Noting $\mu_{k,l}(\Omega_{k,l,j}) = \sum_{\eta\in \Omega_{k,l,j}}\big[\mu_{k,l}(\eta)/\nu^\rho_{l}(\eta)\big]\nu_{l,\e}(\eta)$, we have
\begin{align*}
&\dfrac  {( 1 - c e^{-\beta\Eb_{x_2}} \e^{k-l} )^{2l+1} c^j e^{-\beta\Eb_{x_1} j } \e^{(k+l)j}  }   { (1-\e)^{2l+1}\e^j}\\
&\ \ \ \ \ \ \ \ \ \ \ \ \leq
\dfrac{\mu_{k,l} (\Omega_{k,l,j}) }{\nu^\rho_{l}(\Omega_{k,l,j})}
\leq
\dfrac  {(1- c e^{-\beta\Eb_{x_1}} \e^{k+l})^{2l+1} c^j e^{-\beta\Eb_{x_2} j } \e^{(k-l)j} }   { (1-\e)^{2l+1}\e^j}. 
\end{align*}
Then, rearranging the formulas establishes \eqref{eqn: RN for klj ln k case}:
$
r_{k,l,\e}^{-1}
\leq
\big[\mu_{k,l,j} (\eta)/\nu_{l,j}(\eta)\big]
\leq
r_{k,l,\e}$.

Turning now to the Poincar\'e inequality, from \eqref{d form symm} and \eqref{d form asymm}, using \eqref{eqn: RN for klj ln k case}, we have
\begin{equation} \label {eqn: d from comparison}
E_{\nu_{l,j}}
\left[ f(-L_{l}f) \right]
\leq
\dfrac {r_{k,l,\e}} 2
E_{\mu_{k,l,j}}
\left[ f(-L_{k,l}f) \right].
\end{equation}
Now, since
\begin{equation*}
\begin{split}
{\rm Var}_{\mu_{k,l,j}}(f)
=
\inf_a E_{\mu_{k,l,j}}\big[ (f-a)^2 \big]
\leq&\ 
r_{k,l,\e}
\inf_a E_{\nu_{l,j}}\big[ (f-a)^2 \big]\\
=&\ 
r_{k,l,\e}
{\rm Var}_{\nu_{l,j}}(f),
\end{split}
\end{equation*}
the desired Poincar\'e inequality follows from \eqref{eqn: spetral gap for nu} and \eqref{eqn: d from comparison}.

For the last item, we observe that $\e \rightarrow 1$ as $N\uparrow\infty$.  Also, $\Eb_{x_i} = u(\ln(x_i))\rightarrow \infty$ as $N\uparrow\infty$ given that $aN-l\leq x_i\leq bN+l$ for $i=1,2$.  Finally, $\Eb_{x_2}-\Eb_{x_1} = u(\ln(x_2))-u(\ln(x_1)) = u'(y)\ln(x_2/x_1)$ where $y$ is between $\ln(x_2)$ and $\ln(x_1)$ and so $u'(y)\rightarrow 0$ or $1$, by assumption, as $N\uparrow\infty$.  Hence, as $k-l\leq x_1,x_2\leq k+l$ and $aN\leq k\leq bN$, we have that $\ln(x_2/x_1)\rightarrow 0$ as $N\uparrow\infty$.  All these comments immediately lead to the claim that $r_{k,l,\e}\rightarrow 1$, uniformly over $aN\leq k\leq bN$, as $N\uparrow\infty$.  Moreover, as a consequence, we see that $C_{k,l,j}$ is uniformly bounded for $aN\leq k\leq bN$ and $N\geq 1$, by the form of $C_{k,l,j}$.  \end{proof}

\subsection{$1$-block estimate}
Recall $D^{G,s}_{N,k}$ from \eqref{DG}.  Define
\begin{equation*} \label {eqn: V_kl}
V_{k,l}(s,\eta)
:=
D_{N,k}^{G,s} 
\left(h(\eta(k))
- 
E_{\nu_{\eta^l(k)}}[h]
\right)
\end{equation*}
where $h(x):=\chi_{\{x >0\}}$ and 
$E_{\nu_\rho}[ h]:=E_{\nu_\rho}[ h(\eta(k))]=\dfrac{\rho}{1+\rho}$.

The $1$-block estimate is the following limit.
\begin{Lem}[$1$-block estimate] \label {one block lem}
We have
\begin{equation*}
\limsup_{l\to \infty} \limsup_{N\to \infty}
\sup_{aN\leq k \leq bN }
\E_N
\Big[
\Big |
\int_0^T
N_{\beta} V_{k,l}(s,\eta_s) ds
\Big |
\Big]
=
0.
\end{equation*}
\end{Lem}

\begin{proof}
We separate the argument into steps.

\vskip .1cm
{\it Step 1.}
We first introduce a cutoff of large densities:
For any $l$ and $\epsilon>0$, we may find an $A$ such that for all $t\geq 0$, large $N$,  and $aN\leq k \leq bN$, we have
$\E_N (\chi_{\{\eta_s^l (k)>A\}}) <\epsilon N_{\beta}^{-1}$.  Indeed, as $\nu^N \leq \Pam_{c,N}$, by attractiveness \eqref{attractiveness}, $\E_N(\chi_{\{\eta_t^l(k)>A\}}) \leq E_{\Pam_{c,N}}(\chi_{\{\eta^l (k)>A\}})$.
By Markov's inequality,
\begin{equation*}
E_{\Pam_{c,N}}(\chi_{\{\eta^l (k)>A\}})
\leq \dfrac{1}{A(2l+1)} \sum_{j=k-l}^{k+l} E_{\Pam_{c,N}}( \eta(k) ).
\end{equation*}
Since $N_\beta E_{\Pam_{c,N}}( \eta(k) )$
 is uniformly bounded for all $aN \leq k \leq bN$ and $N\in \N$ by \eqref{site_particle}, it suffices to take $A$ large enough.

Note that $|D^{G,s}_{N,k}|\leq C(a,b,G)$ (cf. \eqref{DG_bound}).  Then,
\begin{equation*}
\begin{split}
&\E_N
\Big[
\Big |
\int_0^T
V_{k,l}(s,\eta_s) ds
\Big |
\Big]\\
\leq&
\E_N
\Big[
\Big |
\int_0^T
V_{k,l}(s,\eta_s) \chi_{\{\eta_s^l(k)\leq A\}} ds
\Big |
\Big]
+
\E_N
\Big[
\Big |
\int_0^T
V_{k,l}(s,\eta_s) \chi_{\{\eta_s^l(k)>A\}} ds
\Big |
\Big]\\
\leq&
\E_N
\Big[
\Big |
\int_0^T
V_{k,l}(s,\eta_s) \chi_{\{\eta_s^l(k)\leq A\}} ds
\Big |
\Big]
+
C(a,b,G)
\E_N
\Big[
\int_0^T
\chi_{\{\eta_s^l(k)>A\}} ds
\Big]\\
=&
\E_N
\Big[
\Big |
\int_0^T
V_{k,l}(s,\eta_s) \chi_{\{\eta_s^l(k)\leq A\}} ds
\Big |
\Big]
+
C(a,b,G)
\int_0^T
\E_N
\big[
\chi_{\{\eta_s^l(k)>A\}} 
\big]
ds\\
\leq&
\E_N
\Big[
\Big |
\int_0^T
V_{k,l}(s,\eta_s) \chi_{\{\eta_s^l(k)\leq A\}} ds
\Big | 
\Big]
+
C(a,b,G) T \epsilon N_{\beta}^{-1}.
\end{split}
\end{equation*}
Hence,
\begin{equation*}
\begin{split}
&
\limsup_{N\to \infty}
\sup_{aN\leq k \leq bN }
\E_N
\Big[
\Big |
\int_0^T
N_{\beta} V_{k,l}(s,\eta_s) ds
\Big |
\Big]\\
\leq&
\limsup_{N\to \infty}
\sup_{aN\leq k \leq bN }
\E_N
\Big[
\Big |
\int_0^T
N_{\beta} V_{k,l}(s,\eta_s) \chi_{\{\eta_s^l(k)\leq A\}} ds
\Big |
\Big]
+
C(a,b,G) T \epsilon.
\end{split}
\end{equation*}

For convenience, we write 
\begin{equation*}
\tilde V_{k,l,A} (s,\eta)
:=
V_{k,l}(s,\eta) \chi_{\{\eta^l_s (k)\leq A\}}.
\end{equation*}
Then, to prove the lemma, it will be enough to show
\begin{equation*}
\limsup_{l\to \infty} \limsup_{N\to \infty}
\sup_{aN\leq k \leq bN }
\E_N
\Big[
\Big |
\int_0^T
N_{\beta} \tilde V_{k,l,A} (s,\eta_s) ds
\Big |
\Big]
=
0.
\end{equation*}

\vskip .1cm
{\it Step 2.}
Define $\L_{k,l} (\eta)$ be the number of particles in $\L_{k,l}$, that is $\L_{k,l} (\eta): = (2l+1) \eta^l(k)$. 
We would like to replace $\tilde V_{k,l,A} (s,\eta)$ by its `centering':
\begin{equation*}
V_{k,l,A}(s,\eta)
:=
D_{N,k}^{G,s} 
\big(h(\eta(k))
- 
E_{\mu_{k,l, \L_{k,l}(\eta)}}[h(\eta(k))]
\big)\chi_{\{\eta^l (k)\leq A\}}.
\end{equation*}
The advantage of working with $ V_{k,l,A}$ is that $E_{\mu_{k,l,j}}  V_{k,l,A}= 0$ for all $k,l,j$.
The difference in making such a replacement is less than
\begin{equation}
\label{step2eqn}
\E_N
\Big[
\int_0^T
N_{\beta} 
\chi_{\{0<\eta^l_s (k)\leq A\}}
\Big |
E_{\mu_{k,l,\L_{k,l}(\eta_s)}}[h(\eta(k))]
-
E_{\nu_{\eta^l_s(k)}}[h]
\Big |
ds
\Big].
\end{equation}
In the above, we replaced $\chi_{\{\eta^l (k)\leq A\}}$ by $\chi_{\{0<\eta^l (k)\leq A\}}$, since $h$ vanishes when $\eta^l (k) =0$.

By adding and subtracting, \eqref{step2eqn} is bounded by
\begin{eqnarray*}
&&
\E_N
\Big[
\int_0^T
N_{\beta} 
\chi_{\{0<\eta^l_s (k)\leq A\}}
\Big |
E_{\mu_{k,l, \L_{k,l}(\eta_s)}}[ h(\eta(k))]
-
E_{\nu_{k,l,\L_{k,l}(\eta_s)}}[ h(\eta(k))]
\Big |
ds
\Big]\\
&&+
\E_N
\Big[
\int_0^T
N_{\beta} 
\chi_{\{0<\eta_s^l (k)\leq A\}}
\Big|
E_{\nu_{k,l,\L_{k,l}(\eta_s)}} [h(\eta(k))]
-
E_{\nu_{\eta^l_s (k)}}[ h]
\Big |
ds
\Big]
=:A_1+A_2.
\end{eqnarray*}

\vskip .1cm
{\it Step 3.}
Now, by \eqref{eqn: RN for klj ln k case} and $0\leq h \leq 1 $, we have
\begin{equation*}
\begin{split}
&\left |
E_{\mu_{k,l,\L_{k,l}(\eta)}}[ h(\eta(k))]
-
E_{\nu_{k,l,\L_{k,l}(\eta)}} [h(\eta(k))]
\right |\\
&\quad\quad\quad\quad\quad\quad\quad\quad\leq
E_{\nu_{k,l,\L_{k,l}(\eta)}}[ h(\eta(k))] (r_{k,l,\e} - 1)
\leq
r_{k,l,\e} - 1.
\end{split}
\end{equation*}
Then, by $\nu^N\leq \Pam_{c,N}$ and attractiveness (cf. Subsection \ref{subsec: attractiveness}), and $\chi_{\{\eta^l(k)>0\}}\leq \eta^l(k)$, the term $A_1$ is bounded by
\begin{equation*}
\begin{split}
 &(r_{k,l,\e} - 1)
\E_N
\Big[
\int_0^T
N_{\beta} 
\chi_{\{0<\eta^l_s (k)\leq A\}}
ds
\Big]
\leq
 (r_{k,l,\e} - 1)
\E_{\Pam_{c,N}}
\Big[
\int_0^T
N_{\beta} 
\chi_{\{\eta^l_s (k)>0\}}
ds
\Big]\\
&\quad\quad\quad\quad\quad\quad\quad\quad\quad\quad\leq
 (r_{k,l,\e} - 1)TN_\beta 
 E_{\Pam_{c,N}} \big[\eta^l (k)\big].
\end{split}
\end{equation*}
By \eqref{site_particle},
$N_\beta 
 E_{\Pam_{c,N}} \big[\eta^l (k)\big] \leq N_\beta \sup_{k-l\leq j\leq k+l}\rho_{j,c}$ is uniformly bounded for each $l\geq 1$, and $aN\leq k\leq bN$ for all $N$ large.
Hence, for each $l$, $\sup_{aN\leq k\leq bN} A_1$ vanishes as $N\uparrow\infty$, as $r_{k,l,\e}\rightarrow 1$ by item (3) in Lemma \ref{spectral gap ln k}.

On the other hand, by equivalence of ensembles (cf.\,p.355, \cite{KL}),  
the absolute value in $A_2$ vanishes as $l\to \infty$, uniformly in $k$ as $\nu_{k,l,j}$ and $\nu_{j/{2l+1}}$ are translation-invariant and do not depend on $k$. 
Therefore, the term $A_2$ vanishes as well as we take $N\to\infty$, $l\to \infty$ in order.

\vskip .1cm
{\it Step 4.}
Now, the proof of the lemma is reduced to prove
\begin{equation*}
\limsup_{l\to \infty} \limsup_{N\to \infty}
N_{\beta}
\sup_{aN\leq k \leq bN }
\E_N
\Big[
\Big |
\int_0^T
V_{k,l,A} (s,\eta_s) ds
\Big |
\Big]
=
0.
\end{equation*}

By the entropy inequality $E_\mu[f] \leq H(\mu|\nu) + \log E_\nu[e^f]$ (cf.\,p.338 \cite{KL}) and the assumption $H(\nu^N| \Pam_{c,N})\leq CNN^{-1}_\beta$, we have
\begin{equation*}
\begin{split}
\E_N
\Big[
\Big |
\int_0^T
V_{k,l,A} (s,\eta_s) ds
\Big |
\Big]
\leq 
\dfrac{C_0}{\gamma  N_\beta}
+
\dfrac{1}{\gamma N} \ln \E_{\Pam_{c,N}}
\Big[
\exp
 \left\{
  \gamma N  
  \Big |
\int_0^T
V_{k,l,A} (s,\eta_s) ds
\Big |
  \right\}
\Big].
\end{split}
\end{equation*}
The absolute value in the right hand side of last inequality can be dropped by using $e^{|x|} \leq e^x + e^{-x}$.
By Feynman-Kac formula (cf.\,p.336, \cite{KL}),
\begin{equation*}
\begin{split}
\dfrac{1}{\gamma N} \ln \E_{\Pam_{c,N}}
\Big[
\exp
 \left\{
  \gamma N  
\int_0^T
V_{k,l,A} (s,\eta_s)
 ds
  \right\}
\Big]
\leq
\dfrac 1 {\gamma N} 
\int_0^T \lambda_{N,l}(s) ds
\end{split}
\end{equation*}
where $\lambda_{N,l}(s)$ is the largest eigenvalue of 
$N^2 L + \gamma N V_{k,l,A}(s,\eta)$.

\vskip .1cm
{\it Step 5.}
Fix $s\in [0,T]$; we will omit the argument $s$ to simply notation.
Note the variational formula for $\lambda_{N,l}$: 
\begin{equation*}
(\gamma N)^{-1} \lambda_{N,l}
=
\sup_f
\left\{
E_{\Pam_{c,N}}\left[  V_{k,l,A} f  \right]
-\gamma^{-1} N 
E_{\Pam_{c,N}}\left [
\sqrt f (-L \sqrt f)
\right]
\right\},
\end{equation*}
where the supremum is over all $f$ which are densities with respect to $\Pam_{c,N}$.

Let $f_{k,l}= E_{\Pam_{c,N}}\big[f|\Omega_{k,l}\big]$, be the conditional expectation of $f$ given the variables on $\L_{k,l}$. Recall that $\mu_{k,l}$ is the restriction of $\Pam_{c,N}$ to $\L_{k,l}$, and that $L_{k,l}$ is the localized generator.  Since the Dirichlet form $E_{\Pam_{c,N}}\left [
\sqrt f (-L_N \sqrt f)
\right]$ is convex, we have
\begin{equation*}
(\gamma N)^{-1} \lambda_{N,l}
\leq
\sup_{f_{k,l}}
\left\{
E_{\mu_{k,l}}\left [ V_{k,l,A} f_{k,l}  \right ]
-\gamma^{-1} N 
E_{\mu_{k,l}}\left [
\sqrt {f_{k,l}} (-L_{k,l} \sqrt {f_{k,l}})
\right]
\right\}.
\end{equation*}\

\vskip .1cm
{\it Step 6.}
We now decompose $f_{k,l}  d\mu_{k,l}$ with respect to sets $\Omega_{k,l,j}$ of configurations with total particle number $j$ on $\L_{k,l}$:
\begin{equation} \label {beta 0, c klj}
E_{\mu_{k,l}}\left[  V_{k,l,A} f_{k,l}  \right]
=
\sum_{j\geq 0}
c_{k,l,j}(f) \int V_{k,l,A} f_{k,l,j} d\mu_{k,l,j},
\end{equation}
where
$ c_{k,l,j}(f) = \int_{\Omega_{k,l,j}} f_{k,l} d\mu_{k,l}$, and 
$ f_{k,l,j} = c_{k,l,j}(f)^{-1} \mu_{k,l} \left( \Omega_{k,l,j} \right) f_{k,l}$.  Here, $\sum_{j\geq 0} c_{k,l,j} = 1 $ and $f_{k,l,j}$ is a density with respect to $\mu_{k,l,j}$.

Straightforwardly, on $\Omega_{k,l,j}$, we have
\begin{equation*}
\begin{split}
\dfrac{L_{k,l} \sqrt {f_{k,l}}}{\sqrt {f_{k,l}}}
=
\dfrac{L_{k,l} \sqrt {f_{k,l,j}}}{\sqrt {f_{k,l,j}}}.
\end{split}
\end{equation*}
Using \eqref{beta 0, c klj}, we write
\begin{equation*}
E_{\mu_{k,l}}\left[
\sqrt {f_{k,l}} (-L_{k,l} \sqrt {f_{k,l}})
\right]
=
\sum_{j\geq 0}
c_{k,l,j}(f) 
E_{\mu_{k,l,j}}\left[
\sqrt {f_{k,l,j}} (-L_{k,l} \sqrt {f_{k,l,j}})
\right ].
\end{equation*}
Then, we get
\begin{equation*}
(\gamma N)^{-1} \lambda_{N,l}
\leq
\sup_{0\leq j\leq A(2l+1)}
\sup_f
\left\{
E_{\mu_{k,l,j}}\left[ V_{k,l,A} f  \right]
-\gamma^{-1} N 
E_{\mu_{k,l,j}}\left [
\sqrt f (-L_{k,l} \sqrt f)
\right]
\right\},
\end{equation*}
where the second supremum is on densities $f$ with respect to $\mu_{k,l,j}$.

\vskip .1cm
{\it Step 7.}
We now use the Rayleigh expansion (see p.375, \cite{KL}), 
where $C_{k,l,j}$ is the uniformly bounded inverse spectral gap estimate of $L_{k,l}$ (cf. Lemma \ref {spectral gap ln k}) and $ \|V_{k,l,A}\|_{\infty} \leq  |D^{G,s}_{N,k}| \leq C(a,b,G)$.  We have
\begin{equation} \label {eqn: rayleigh expansion}
\begin{split}
&E_{\mu_{k,l,j}}\left[   V_{k,l,A} f  \right]
-\gamma^{-1} N 
E_{\mu_{k,l,j}}\left [
\sqrt f (-L_{k,l} \sqrt f)
\right ]\\
\leq&
\dfrac{\gamma N^{-1}}{1-2C(a,b,G) C_{k,l,j}\,\gamma N^{-1}}
E_{\mu_{k,l,j}}\left[
  V_{k,l,A}  (-L_{k,l} )^{-1}
  V_{k,l,A}
\right ].
\end{split}
\end{equation}

The spectral gap estimate of $L_{k,l}$ in Lemma \ref{spectral gap ln k} also implies that $\|L_{k,l}^{-1}\|_2$, the $L^2(\mu_{k,l,j})$ norm of the operator $L_{k,l}^{-1}$ on 
mean zero functions, is less than or equal to $C_{k,l,j}$.

Now, by Cauchy-Schwartz and the estimate of $\|L_{k,l}^{-1}\|_2$, we have 
\begin{equation*} 
\begin{split}
E_{\mu_{k,l,j}}\left[
  V_{k,l,A}  (-L_{k,l} )^{-1}
  V_{k,l,A}
\right ]
\leq
C_{k,l,j}
E_{\mu_{k,l,j}}\left[
  V_{k,l,A} ^2
\right ].
\end{split}
\end{equation*}

Accordingly, retracing our steps, noting \eqref{eqn: rayleigh expansion}, we have
\begin{equation*}
\begin{split}
\E_N
\Big[
\Big |
\int_0^T
N_{\beta} V_{k,l,A} (\eta_s) ds
\Big |
\Big]
\leq 
\dfrac{C_0}{\gamma}
+
\sup_{0\leq j\leq A(2l+1)}
\dfrac{\gamma N_\beta N^{-1} C_{k,l,j}}{1-2C(a,b,G) C_{k,l,j}\,\gamma N^{-1}}
E_{\mu_{k,l,j}}\left [
  V_{k,l,A} ^2
\right ].
\end{split}
\end{equation*}
The last expression vanishes uniformly as $N\to \infty$ 
for $aN\leq k \leq bN$ and $j\leq A(2l+1)$.
The lemma now is proved by letting $\gamma \to \infty$ and $\epsilon \to 0$.
\end{proof}


\subsection{ $2$-block estimate}
In this subsection, we will restrict to the case $\beta=0$, since a $2$-block estimate is not needed for the other cases, and as remarked earlier may not hold when $\beta> 0$.

Recall the notation $\L_{k,l}$ from the $1$-block estimate.  For $l\geq 1$ and $l< k<k'$, let  $\L_{k,k',l} = \L_{k,l} \cup \L_{k',l}$ for $|k-k'|>2l$.
We introduce the following localized generator $L_{k,k',l}$ governing the coordinates $\Omega_{k,k',l} = \left\{0,1,2,\ldots \right\}^{\L_{k,k',l}}$.
 Inside each block, the process moves as before, but we add an extra bond interaction between sites $k+l$ and $k'-l$:
\begin{equation*}
\begin{split}
&L_{k,k',l}f(\eta)\\
=&
\sum_{x,x+1\in \L_{k,k',l}}
\left\{
\lambda_x
\left [  f\left(\eta^{x,x+1} \right)  - f(\eta)   \right] \chi_{\{\eta(x) > 0\}}
+
\left [  f\left(\eta^{x+1,x}  \right)   - f(\eta) \right ] \chi_{\{\eta(x+1) > 0\}}
\right\}\\
&+
 \dfrac{\theta_{k'-l}}{\theta_{k+l}}
\left [  f\Big(\eta^{k+l,k'-l} \Big)  - f(\eta)   \right] \chi_{\{\eta(k+l) > 0\}}
+
\left [  f\Big(\eta^{k'-l,k+l}  \Big)   - f(\eta) \right ] \chi_{\{\eta(k'-l) > 0\}}.
\end{split}
\end{equation*}
Here, as $\beta=0$, we have $\theta_{k} = e^{-k/N}$ and $\lambda_k = e^{-1/N}$. As before, the localized measure $\mu_{k,k',l}$ defined by $\mu=\Pam_{c,N}$ limited to sites in $\L_{k,k',l}$, 
as well as the canonical measure $\mu_{k,k',l,j}$ on $\Omega_{k,k',l,j}: = \{\eta\in \Omega_{k,k',l}: \sum_{x\in \L_{k,k',l}} \eta(x) = j\}$, that is $\mu_{k,k',l}$ is conditioned so that there are exactly $j$ particles counted in $\Omega_{k,k',l}$, are both invariant for the dynamics.

The corresponding Dirichlet form, with measure $\kappa$ given by $\mu_{k,k',l}$ or $\mu_{k,k',l,j}$, is given by
\begin{equation*}
\begin{split}
E_{\kappa} \left[ f(-L_{k,k',l} f) \right ]
=&
\sum_{x,x+1\in \L_{k,k',l}}
 E_\kappa \left[ 
 \chi_{\{\eta(x+1) > 0\}}
\left [  f\left(\eta^{x+1,x}  \right)   - f(\eta) \right ]^2
 \right ]\\
 &+
  E_\kappa \Big [
 \chi_{\{\eta(k'-l) > 0\}}
\left [  f\left(\eta^{k'-l,k+l}  \right)   - f(\eta) \right ]^2
 \Big ].
\end{split}
\end{equation*}

Recall that $\e = e^{-1/N}$.
Corresponding to the set-up of the gap bound Lemma \ref{spectral gap ln k},
let $\nu^\rho_{l,l}$ be the product of $4l+2$ Geometric distributions with common parameter $\e$ and mean $\rho$ such that $\e = \dfrac{\rho}{1+\rho}$, 
and $\nu_{l,l,j}$ be $\nu^\rho_{l,l}$ conditioned on that the total number of particles in the $4l+2$ sites is $j$.
Note that $\nu_{l,l,j}$ is independent of $\rho$.

\begin{Lem} \label {spectral gap of 2 blocks}
We have the following estimates:
\begin{enumerate}
\item Uniform bound:  For all $\eta\in \Omega_{k,k', l, j}$, we have
\begin{equation} \label {eqn: RN for kk'lj}
r_{k,k',l,\e}^{-1}
\leq
\dfrac{\mu_{k,k',l,j} (\eta) }{\nu_{l,l,j}(\eta)}
\leq
r_{k,k',l,\e}
\end{equation}
 where 
$r_{k,k',l,\e} :=  \left( \dfrac{1- c \e^{k'+l}}{1- c \e^{k-l}}\right) ^{4l+2} \e^{-2lj}\e^{(k-k')j}$.

\item Poincar\'e inequality:
For $1\leq l<k<k'$ and fixed $j\geq 0$, we have
\begin{equation} 
\label{spectral gap, 2 blocks}
{\rm Var}_{\mu_{k,k',l,j}}(f)
\leq
C_{k,k',l,j} E_{\mu_{k,k',l,j}} \big [ f(-L_{k,k',l}) \big ]
\end{equation}
where
$
C_{k,k',l,j}
=
\dfrac C2
(4l+2)^2 \Big(1+\dfrac{j}{4l+2}\Big)^2 r^2_{k,k',l,\e}
$
for an universal constant $C$.

\item 
For fixed $j$, $l$, and $0<a<b<\infty$, we have
\begin{equation}
\label{4stars}
\limsup_{\tau\downarrow 0}\limsup_{N\uparrow\infty}\sup_{\stackrel{aN<k<k'\leq bN}{2l+1\leq |k'-k|\leq \tau N}}r_{k,k',l,\e} \leq 1,
\end{equation}
and so 
$\limsup_{\tau\downarrow 0}\limsup_{N\uparrow\infty} 
 \sup_{\substack{aN\leq k<k' \leq bN  \\ 2l+1\leq |k'-k| \leq \tau N }}
C_{k,k',l,j}<\infty$. 
\end{enumerate}
\end{Lem}

\begin{proof}
We will compare $\mu_{k,k',l,j}$ with $\nu_{l,l,j}$ and make use of the known Poincar\'e bound, as in the proof of Lemma \ref{spectral gap ln k}:
\begin{equation} \label {eqn: spetral gap for nu k k'}
{\rm Var}_{\nu_{l,l,j}}( f) 
\leq
C
(4l+2)^2 \Big(1+\dfrac{j}{4l+2}\Big)^2
 E_{\nu_{l,l,j}} [ f(-L_{l,l}f) ],
\end{equation}
where $C$ is some universal constant.

For $\eta \in \Omega_{k,k',l,j}$, we have
\begin{equation*}
\dfrac{\mu_{k,k',l,j} (\eta) }{\nu_{l,l,j}(\eta)}
=
\dfrac{\mu_{k,k',l} (\eta) }{\nu^\rho_{l,l}(\eta)}
\dfrac{\nu^\rho_{l,l}(\Omega_{k,k',l,j})}{\mu_{k,k',l} (\Omega_{k,k',l,j})}.
\end{equation*}
Since $\mu_{k,k',l}$ and $\nu^\rho_{l,l}$ are product measures, and $\beta =0$, that is, 
\begin{equation}
\label{triplestar}
\mu_{k,k',l} (\eta) = \prod_{x\in \L_{k,k',l}} (1-c\e^x) c ^{\eta(x)}\e ^{x\eta(x)}, 
\ {\rm and \ }
\nu^\rho_{l,l}(\eta)= \prod_{x\in \L_{k,k',l}} (1-\e) \e ^{\eta(x)},
\end{equation}
we have
\begin{equation*}
\dfrac  {(1-c\e^{k-l})^{4l+2} c^j \e ^{(k'+l)j} }   { (1-\e)^{4l+2}\e^j}
\leq
\dfrac{\mu_{k,k',l} (\eta) }{\nu^\rho_{l,l}(\eta)}
\leq
\dfrac  {(1-c \e^{k'+l})^{4l+2}c^j \e ^{(k-l)j} }   { (1-\e)^{4l+2}\e^j}. 
\end{equation*}
 Consequently,
\begin{equation*}
\dfrac  {(1-c\e^{k-l})^{4l+2} c^j \e ^{(k'+l)j} }   { (1-\e)^{4l+2}\e^j}
\leq
\dfrac{\mu_{k,k',l} (\Omega_{k,k',l,j})} {\nu^\rho_{l,l}(\Omega_{k,k',l,j})}
\leq
\dfrac  {(1-c \e^{k'+l})^{4l+2} c^j \e ^{(k-l)j} }   { (1-\e)^{4l+2}\e^j} .
\end{equation*}
Therefore,
$r_{k,k',l,\e}^{-1}
\leq
\dfrac{\mu_{k,k',l,j} (\eta) }{\nu_{l,j}(\eta)}
\leq
r_{k,k',l,\e}$ and \eqref{eqn: RN for kk'lj} holds.

Recall the generator of symmetric zero-range $L_l$ with respect to $\Lambda_{k,l}$ (cf. \eqref{onestar}).  Let $L'_l$ be the generator with respect to $\Lambda_{k',l}$.    Define, noting $1\leq l<k<k'$, the generator $L_{l,l}$ with respect to $\L_{l,l}$ given by
\begin{align*}
L_{l,l}f(\eta) &= L_l + L'_l \\
&\ \ \ \ + \frac{1}{2}\left[f\big(\eta^{k+l, k'-l}\big) - f(\eta)\right]\chi_{\{\eta(k+l)>0\}} +
\frac{1}{2}\left[f\big(\eta^{k'-l,k+l}\big) - f(\eta)\right]\chi_{\{\eta(k'-l)>0\}}.
\end{align*}
When $|k-k'|$ is large, the process governed by $L_{l,l}$ in effect treats the blocks as adjacent.  The canonical measure $\nu_{l,l,j}$ is invariant to the dynamics.  The corresponding Dirichlet form is given by
$$E_{\nu_{l,l,j}}\left[f(-L_{l,l} f)\right] = \frac{1}{2} \sum_{x,x+1\in \Lambda_{k,k',l}} E_{\nu_{l,l,j}}\left\{ \left[f\big(\eta^{x,x+1}\big)-f(\eta)\right\}^2\chi_{\{\eta(x)>0\}}\right],$$
where we interpret that $k+l$ and $k'-l$ are neighbors in the above formula.

From \eqref{eqn: RN for kk'lj}, we have
\begin{equation} \label {eqn: d from comparison, kk'}
E_{\nu_{l,l,j}}\left[ f(-L_{l,l}f) \right ]
\leq
\dfrac12
r_{k,k',l,\e}
E_{\mu_{k,k',l,j}}\left[ f(-L_{k,k',l}f) \right].
\end{equation}

Also, in turn,
\begin{equation*}
\begin{split}
{\rm Var}_{\mu_{k,k',l,j}}( f) 
=
\inf_a E_{\mu_{k,k',l,j}} \left [(f-a)^2 \right ]
\leq&
r_{k,k',l,\e}
\inf_a E_{\nu_{l,l,j}} \left [(f-a)^2 \right ]\\
=&
r_{k,k',l,\e}
{\rm Var}_{\nu_{l,l,j}}( f).
\end{split}
\end{equation*}
The spectral gap estimate \eqref{spectral gap, 2 blocks} now follows from \eqref {eqn: spetral gap for nu k k'} and \eqref{eqn: d from comparison, kk'}.

To complete the proof of the lemma, noting that $\e=e^{-1/N}$, for any fixed $l$, $j$, we see straightforwardly that
$$\limsup_{N\uparrow\infty}
 \sup_{\substack{aN\leq k<k' \leq bN  \\ 2l+1\leq |k'-k| \leq \tau N }}
r_{k,k',l,\e}
 \leq \sup_{a\leq x\leq b}\left( (1-c e^{-\tau}e^{-x})/(1-c e^{-x})\right)^{4l+2} e^{\tau j},$$
which converges to $1$ as $\tau\downarrow 0$. Hence, the limit \eqref{4stars} and the desired uniform boundedness of $C_{k,k',l,j}$ both follow.
\end{proof}

\medskip
We now state and show a $2$-blocks estimate.  The scheme is similar to that of the $1$-block estimate.  Recall $D^{G,s}_{N,k}$ and its bound for $aN\leq k\leq bN$ that $|D^{G,s}_{N,k}| \leq C(a,b,G)$ (cf. \eqref{DG_bound}).

\begin{Lem}[$2$-block estimate] \label{lem: 2 blocks} We have
\begin{equation*}
\limsup_{l\to \infty}
\limsup_{\tau \to 0}
 \limsup_{N\to \infty}
 \sup_{aN\leq k \leq bN }
\E_N \Big |
\int_0^T 
D_{N,k}^{G,s}
\Big( 
\dfrac{\eta_s^l(k)}{1+ \eta_s^l(k)}
-
\dfrac{\eta_s^{\tau N}(k)}{1+ \eta_s^{\tau N}(k)} \Big) 
ds \Big |
=
0.
\end{equation*}
\end{Lem}

\begin{proof}
We separate the argument into steps.
\vskip .1cm

{\it Step 1.}
Since $\dfrac{x}{1+x}$ is Lipschitz on $\R^+$ and $D_{N,k}^{G,s} $ is bounded,
it is enough to show
\begin{equation*}
\limsup_{l\to \infty}
\limsup_{\tau \to 0}
 \limsup_{N\to \infty}
 \sup_{aN\leq k \leq bN }
\E_N 
\int_0^T 
\left |
\eta_s^{\tau N}(k)
-
\eta_s^l(k)
 \right |
ds
=
0.
\end{equation*}

By the triangle inequality, it will be enough to show that as $N\to \infty$, $\tau\to 0$, and $l\to \infty$,
\begin{equation} \label {eqn: two block,12}
\begin{split}
 \sup_{aN\leq k \leq bN }
\E_N 
\int_0^T 
\Big |
\eta_s^{\tau N}(k)
-
\dfrac{1}{2\tau N+1} \sum_{|x-k|\leq \tau N}
\eta_s^{l}(x)
 \Big |
ds
\to
0 \ \ {\rm and }
\\
 \sup_{aN\leq k \leq bN }
\E_N
\int_0^T 
\Big |
\dfrac{1}{2\tau N+1} \sum_{|x-k|\leq \tau N}
\eta_s^{l}(x)
-
\eta_s^l(k)
 \Big |
ds
\to
0.
\end{split}
\end{equation}

\vskip .1cm

{\it Step 2.}
We now show that the first limit in \eqref{eqn: two block,12}.
Note that
\begin{equation*}
\begin{split}
&\Big |
\eta^{\tau N}(k)
-
\dfrac{1}{2\tau N+1} \sum_{|x-k|\leq \tau N}
\eta^{l}(x)
\Big|
\leq
\dfrac{1}{2\tau N+1} 
\sum_{\substack{|x-k-\tau N|\leq l \\ \text{ or } |x-k+\tau N| \leq l}}
\eta(x)\\
&\quad\quad\quad\quad\quad\quad\quad\quad\quad\quad\quad\quad=
\dfrac{2l+1}{2\tau N+1} \left(\eta^l(k-\tau N) + \eta^l(k+\tau N) \right).
\end{split}
\end{equation*}
Then, the expectation in the first limit in \eqref{eqn: two block,12}, given that $\nu^N\leq \Pam_{c,N}$ and that the process is attractive (cf. Subsection \ref{subsec: attractiveness}), is bounded from above by
\begin{equation*}
\dfrac{2l+1}{2\tau N+1} 
\int_0^T 
\E_{\Pam_{c,N}}
\left(
\eta_s^l(k-\tau N) + \eta_s^l(k+\tau N) 
\right)
ds
\end{equation*}
For fixed $l$ and $\tau<a$, since $k\geq aN$ and $\beta=0$, we have $E_{\Pam_{c,N}}[\eta(k)] =
 \rho_{k,c} = ce^{-k/N}/(1-ce^{-k/N})\leq 1$ (cf. \eqref{rho k c}).  Hence, the above display vanishes uniformly in $k$ as $N\to \infty$.

\vskip .1cm
{\it Step 3.}
By the same argument as in Step 2, we can restrain the $x$ in the summation of the second limit in \eqref {eqn: two block,12}
to be $k'$ such that $2l+1\leq |k'-k|\leq \tau N$.
Then, the second limit will follow if we show that
\begin{equation*}
\limsup_{l\to \infty}
\limsup_{\tau \to 0}
 \limsup_{N\to \infty}
 \sup_{\substack{aN\leq k<k' \leq bN  \\ 2l+1\leq |k'-k| \leq \tau N }}
\E_N 
\int_0^T 
\left|
\eta_s^l(k)
-
\eta_s^l(k') 
\right |
ds
=
0.
\end{equation*}

\vskip .1cm
{\it Step 4.}
We will apply a cutoff of large densities first. 
Let 
\begin{equation*}
\eta_s^{l} (k,k') =  \eta_s^l(k) + \eta_s^l(k') .
\end{equation*}
For any $A$,
\begin{equation*}
\begin{split}
&\E_N 
\int_0^T 
\left|
\eta_s^l(k)
-
\eta_s^l(k') 
\right |
ds
=
\E_N \int_0^T 
\left|
\eta_s^l(k) - \eta_s^l(k') 
\right |
\chi_{\{\eta_s^{l} (k,k') \leq A\} }
ds\\
&\quad\quad\quad\quad\quad\quad \quad\quad\quad+
\E_N  \int_0^T 
\left|
\eta_s^l(k) - \eta_s^l(k') 
\right |
\chi_{\{\eta_s^{l} (k,k') > A \}}
ds = I_1 + I_2.
\end{split}
\end{equation*}

As $\nu^N\leq \Pam_{c,N}$ and the process is attractive (cf. Subsection \ref{subsec: attractiveness}), we may bound the second expectation $I_2$ by
\begin{equation}
\label{step 3 eq}
\E_N  \int_0^T 
\eta_s^{l} (k,k')
\chi_{\{\eta_s^{l} (k,k') > A \}}
ds
\leq
\dfrac T A
E_{\Pam_{c,N}} \left( \eta^{l} (k,k') \right)^2.
\end{equation}
Recall $\rho_{k,c} = ce^{-k/N}/(1-ce^{- k/N})$ when $\beta=0$ (cf. \eqref{rho k c}). Trivially, $\rho_{k,c}\leq c/(1-c)$ for all $k$. Note that $\Pam_{c,N}$ has Geometric marginals, therefore, $E_{\Pam_{c,N}}[\eta(k)^2] = 2 \rho^2_{k,c} + \rho_{k,c}$ is uniformly bounded.  Then, as
$$\left( \eta^{l} (k,k') \right)^2 \leq 2\left( \eta^l(k)^2 + \eta^l(k')^2\right) \leq 2(2l+1)^{-1}\sum_{x\in \L_{k,l}\cup \L_{k',l}}\eta(x)^2,$$
we have that \eqref{step 3 eq} is of order $O(A^{-1})$ and that the second expectation $I_2$ is negligible.

Hence, it remains to show that 
\begin{equation*}
\begin{split}
 \sup_{\substack{aN\leq k<k' \leq bN  \\ 2l+1\leq |k'-k| \leq \tau N }}
\E_N \int_0^T 
\left|
\eta_s^l(k) - \eta_s^l(k') 
\right |
\chi_{\{\eta_s^{l} (k,k') \leq A \}}
ds
\end{split}
\end{equation*}
vanishes as we take $N\to \infty$, $\tau\to 0$, and then $l\to \infty$.

\vskip .1cm
{\it Step 4.}
Let
\begin{equation*} 
V_{k,k',l,A}(\eta)
:=
\left|
\eta^l (k)
- 
\eta^l(k')
\right|
\chi_{\{\eta^l(k,k')\leq A\}}.
\end{equation*}
Following the proof of Lemma \ref{one block lem}, for fixed $l,\tau,N,k,k'$, in order to estimate
\begin{equation*}
\E_N
\int_0^T 
V_{k,k',l,A} (\eta_s)
ds
\end{equation*} 
it suffices to bound
\begin{equation}
\label{eigen_2blocks}
(\gamma N)^{-1} \lambda_{N,l}
=
\sup_f
\left\{
E_{\Pam_{c,N}} \left [  V_{k,k',l,A} f  \right ]
-\gamma^{-1} N 
E_{\Pam_{c,N}} \left [
\sqrt f (-L \sqrt f)
\right ]
\right\}.
\end{equation}
where the supremum is over all $f$ which are densities with respect to $\Pam_{c,N}$.

\vskip .1cm
{\it Step 5.}
Recall the generator $L_{k,k',l}$ and its Dirichlet form defined in the beginning of this Subsection.  Recall also $\mu_{k,k',l}$ is the restriction of $\Pam_{c,N}$ to $\L_{k,k',l}$.  The Dirichlet form with respect to the full generator $L$ under $\Pam_{c,N}$ is given by
$$E_{\Pam_{c,N}} \left[ f(-L f)\right ]
= \sum_{x\geq 1} E_{\Pam_{c,N}} \left [ \chi_{\{\eta(x+1)>0\}} \left ( f(\eta^{x+1,x}) - f(\eta(x) \right)^2\right].$$

We now argue the following Dirichlet form inequality:
\begin{equation}\label {eqn: d form kk'l LN}
E_{\mu_{k,k',l}} \left [
\sqrt f (-L_{k,k',l}\sqrt f)
\right ]
\leq
(1+\tau N)
E_{\mu_{k,k',l}} \left [
\sqrt f (-L \sqrt f)
\right ].
\end{equation}
First, we observe that
\begin{equation*}
\begin{split}
E_{\mu_{k,k',l}} \left [ f(-L_{k,k',l} f) \right ]
=&
\sum_{x,x+1\in \L_{k,k',l}}
 E_{\Pam_{c,N}} \left [
 \chi_{\{\eta(x+1) > 0\}}
\left [  f\left(\eta^{x+1,x}  \right)   - f(\eta) \right ]^2
 \right ]\\
 &+
  E_{\Pam_{c,N}} \left [ 
 \chi_{\{\eta(k'-l) > 0\}}
\left [  f\left(\eta^{k'-l,k+l}  \right)   - f(\eta) \right ]^2
 \right ].
\end{split}
\end{equation*}
Here, we wrote the terms on the right-hand side as expectations over the product measure $\mu= \Pam_{c,N}$.

Next, by adding and subtracting at most $\tau N$ terms, we have
\begin{equation*}
\begin{split}
&\left [  f\left(\eta^{k'-l,k+l}  \right)   - f(\eta) \right ]^2\\
\leq&
(k'-k-2l) \sum_{q=0}^{k'-k-2l-1}
\left [  f\left(\eta^{k'-l,k+l+q}  \right)   -f\left(\eta^{k'-l,k+l+q+1}  \right)  \right ]^2
.
\end{split}
\end{equation*}
Also, when $\eta(k'-l)>0$,
 by applying the change of variables $\xi = \eta^{k'-l,k+l+q+1}$ which takes away a particle at $k'-l$ and adds one at $k+l+q+1$, we have (cf. \eqref{triplestar})
\begin{equation*}
\mu(\eta)
=
\e^{k'-k-2l-q-1}\Pam_{c,N}(\xi)
\leq
\mu(\xi).
\end{equation*}
Then, as
$\chi_{\{\eta(k'-l)>0\}} = \chi_{\{\xi(k+l+q+1)> 0\}}$, we have
\begin{equation*}
\begin{split}
 &E_{\Pam_{c,N}} \left [
 \chi_{\{\eta(k'-l) > 0\}}
\left [  f\left(\eta^{k'-l,k+l+q}  \right)   -f\left(\eta^{k'-l,k+l+q+1}  \right)  \right ]^2
 \right]\\
 =&
\sum_{\xi} \mu(\eta) 
 \chi_{\{\xi(k+l+q+1) > 0\}}
 \left [  f\left(\xi^{k+l+q+1,k+l+q}  \right)   -f\left(\xi  \right)  \right ]^2\\
 \leq&
 E_{\Pam_{c,N}} \left [ 
 \chi_{\{\eta(k+l+q+1) > 0\}}
 \left [  f\left(\eta^{k+l+q+1,k+l+q}  \right)   -f\left(\eta  \right)  \right ]^2
 \right ].
\end{split}
\end{equation*}
From these observations, \eqref {eqn: d form kk'l LN} follows.

\vskip .1cm
{\it Step 6.} 
Inputting \eqref{eqn: d form kk'l LN} into \eqref{eigen_2blocks}, and considering the conditional expectation of $f$ with respect to $\Omega_{k,k',l}$ as in the $1$-block estimate proof, for $N$ large, we have
\begin{equation*}
(\gamma N)^{-1} \lambda_{N,l}
\leq
\sup_{f_{k.k',l}}
\left\{
E_{\mu_{k,k',l}} \left [  V_{k,k',l,A} f_{k,k',l}  \right ]
-\dfrac{1}{2\tau \gamma} 
E_{\mu_{k,k',l}} \left [
\sqrt {f_{k,k',l}} (-L_{k,k',l} \sqrt {f_{k,k',l}})
\right ]
\right\},
\end{equation*}
where the supremum is over densities with respect to $\mu_{k,k',l}$.

Again, as in the proof of the $1$-block estimate, decomposing $f_{k,k',l}  d\mu_{k,k'.l}$ along configurations with common total number $j$, we need only to bound
\begin{equation*}
\sup_{0\leq j\leq A(2l+1)}
\sup_f
\left\{
E_{\mu_{k,k',l,j}} \left [ V_{k,k',l,A} f  \right ]
-\dfrac{1}{2\tau \gamma} 
E_{\mu_{k,k',l,j}} \left [
\sqrt f (-L_{k,k',l} \sqrt f)\right]
\right\},
\end{equation*}
where the supremum is over densities with respect to $\mu_{k,k',l,j}$.
\vskip .1cm

{\it Step 7.}
Let 
\begin{equation*}
\widehat V_{k,k',l,A}
=
V_{k,k',l,A}
-
E_{\mu_{k,k',l,j}}
\left[V_{k,k',l,A}   \right].
\end{equation*}
Using the Rayleigh expansion (cf. p.375, \cite{KL}) 
where the inverse spectral gap $C_{k,k',l,j}$ of $L_{k,k',l}$ is bounded (Lemma \ref{spectral gap of 2 blocks}), and $ \| \widehat V_{k,k',l,A}\|_{\infty} \leq  A$, we have 
\begin{equation*}
\begin{split}
&E_{\mu_{k,k',l,j}} \left [  \widehat V_{k,k',l,A} f  \right ]
-\dfrac{1}{2\tau \gamma} 
E_{\mu_{k,k',l,j}} \left [
\sqrt f (-L_{k,k',l} \sqrt f)
\right ]\\
\leq&
\dfrac{2\tau \gamma}{1-4 A C_{k,k',l,j}\,\tau \gamma}
E_{\mu_{k,k;,l,j}}\left [
 \widehat V_{k,k',l,A}  (-L_{k,k',l} )^{-1}
 \widehat V_{k,k',l,A}
\right ]\\
\leq&
\dfrac{2\tau \gamma C_{k,k',l,j}}{1-4 A C_{k,k',l,j}\,\tau \gamma}
E_{\mu_{k,k',l,j}}\left [
 \widehat V_{k,k',l,A} ^2
\right ]
\to 0 \text{ as } \tau \to 0.
\end{split}
\end{equation*}

\vskip .1cm
{\it Step 8.}
To finish, we still need to show that $E_{\mu_{k,k',l,j}}\left[  V_{k,k',l,A}   \right ]$ vanishes.
 In fact, by Lemma \ref{spectral gap of 2 blocks},
$E_{\mu_{k,k',l,j}}\left [ V_{k,k',l,A}   \right ]
\leq r_{k,k',l,\e} E_{\nu_{l,l,j}}\left [  V_{k,k',l,A}   \right ]$ and, for $l$ and $j$ fixed,
$$
\limsup_{\tau\downarrow 0}
\limsup_{N\uparrow\infty}
 \sup_{\substack{aN\leq k<k' \leq bN  \\ 2l+1\leq |k'-k| \leq \tau N }}
r_{k,k',l,\e} \leq 1.
$$
The term $E_{\nu_{l,l,j}}\left [  V_{k,k',l,A}   \right ]$ does not depend on $N$ or $\tau$.   By adding and subtracting $j/(2(2l+1))$, we need only bound $E_{\nu_{l,l,j}}\left[\big |\eta^l(k)- j/(2(2l+1))\big|\right]$.  By an equivalence of ensemble estimate (cf. p. 355 \cite{KL}), $E_{\nu_{l,l,j}}\left [  \big|\eta^l(k) - j/(2(2l+1))\big| ^2  \right ] \leq C(A){\rm Var}_{\nu^{j/(2(2l+1))}_{l,l}}\left( \eta^l(k)\right)$. 
This variance is of order $O(l^{-1})$, since the single site variance ${\rm Var}_{\nu^{j/(2(2l+1))}_{l,l}}\left(\eta(k)\right)$ is uniformly bounded for $j/(2(2l+1))\leq A$.  Hence, $E_{\nu_{l,l,j}}\left [  V_{k,k',l,A}   \right ]$ is of order $O(l^{-1/2})$, finishing the proof. 
\end{proof}


\section{Properties of the initial measures} \label {section: initial measures}

In this section, we show key properties of the invariant measures $\Pam_{c,N}$ in Subsection \ref{muNsect}, the local equilibria $\mu^N$ in Subsection \ref{subsection: mu N}, and also of $\nu^N$ in Subsection \ref{subsec: nu N}.

Recall the three regimes in Subsection \ref{GC ensembles and static limit}:  (1) $\beta =0$, (2) $\Eb_k \sim \ln k$ and $0<\beta<1$, and (3) $1\ll \Eb_k \ll \ln  k$ and $\beta>0$.

\subsection{Properties of the invariant measures}
\label{muNsect}
We first show that $\Pam_{c,N}$ is indeed an invariant measure.

\begin{Lem} \label{lem: invariant measure} For $0\leq c\leq c_0$, we have
$E_{\Pam_{c,N}} (Lf(\xi)) = 0$ for all bounded functions $f$ depending only on a finite number of occupation variables $\{\xi(k)\}$.  Hence, $\Pam_{c,N}$ is an invariant measure.
\end{Lem}
\begin{proof} When $c=0$, there are no particles in the system and the statement is trivial.  For $0<c\leq c_0$,
recall that $\lambda_k = \theta_{k+1}/\theta_k = \theta_{k+1,c}/\theta_{k,c}$, and the definition of the generator $L$ (cf. \eqref{eqn: lambda k}).
We need only show that
\begin{equation*}
\begin{split}
E_{\Pam_{c,N}}
\sum_{k=1}^{\infty}
\lambda_k 
\left (  f\left(\xi^{k,k+1} \right)  - f(\xi) 
\right) \chi_{\{\xi(k) > 0\}}
=
-
E_{\Pam_{c,N}}
\sum_{k=2}^{\infty}
\left (  f\left(\xi^{k,k-1}  \right)   - f(\xi)
 \right ) \chi_{\{\xi(k) > 0\}}.
\end{split}
\end{equation*}
For any fixed $k\geq 1$, make a change of variable $\eta=\xi^{k,k+1}$ when $\xi(k)>0$.
Then, $\xi = \eta^{k+1,k}$ and $\eta(k+1)>0$.
 Using that $\Pam_{c,N}$ is a product of 
Geometric marginals with parameters $\{\theta_{k,c}\}$, we have $\chi_{\{\xi(k)>0\}}\dfrac{\Pam_{c,N}(d\xi)}{\Pam_{c,N}(d\eta)} = \dfrac{\chi_{\{\eta(k+1)>0\}}}{\lambda_k}$.
Therefore,
\begin{equation*}
\begin{split}
E_{\Pam_{c,N}}
\left[
\lambda_k 
\left (  f\left(\xi^{k,k+1} \right)  - f(\xi) 
\right) \chi_{\{\xi(k) > 0\}}
\right]
=
E_{\Pam_{c,N}}
\left[
\left (  f\left(\eta \right)  - f(\eta^{k+1,k}) 
\right) \chi_{\{\eta(k+1) > 0\}}
\right].
\end{split}
\end{equation*}
The lemma now follows from
a change of notation from $\eta$ back to $\xi$. 
\end{proof}


To prepare to show that $\Pam_{c,N}$ is a local equilibrium measure, we will need the following.  Recall $\theta_k = e^{-\beta \Eb_k - k/N}$ and $N_\beta= e^{\beta \Eb_N}$ (cf. \eqref{Nbeta eqn}).

\begin{Lem} \label {lem: theta k limit to rho}
For any fixed $0\leq a < b \leq \infty$, we have
\begin{equation*}
\lim_{N\to \infty}
\sum_{k=aN} ^{bN} 
N^{-1} N_\beta \theta_k
=
\lim_{N\to \infty}
\dfrac 1N 
\sum_{k=aN} ^{bN} 
e^{-\beta (u(\ln k) - u(\ln N))  -k/N}
=
\int_a^b \phi(x) dx
\end{equation*}
where
$\phi = e^{-x}$ in regime (1), 
$\phi=x^{-\beta} e^{-x}$ in regime (2)
and 
$\phi = e^{-x}$ in regime (3).
\end{Lem}
\begin{proof}
We will show the lemma in regime (2), that is when $\Eb_k \sim \ln k$ and $0<\beta<1$
as the other regime (3), when $1\ll \Eb_k\ll \ln k$ and $\beta>0$, can be proved in a similar way, and regime (1), when $\beta=0$ is more trivial.
We will also suppose $a=0$, $b=\infty$, as the argument is the same for any other pair $a,b$.
Define 
$$\Phi_N(x) = \sum_{k=1}^{\infty} e^{-\beta (u(\ln k) - u(\ln N))-k/N} \chi_{(\frac{k-1}{N},\frac{k}{N}]}(x).$$
We need only show that  $\lim_{N\to \infty}\int_0^{\infty} \Phi_N(x) dx = \int_0^\infty \phi(x) dx$ to finish.

 By the mean value theorem, 
$u(\ln k) - u(\ln N) = u'(x_{k,N}) \ln \dfrac kN$, where $x_{k,N}$ is between $\ln k$ and $\ln N$.
Fix $ \beta_1$ such that $\beta<\beta_1<1$.
Since $u'(x)\to 1$ as $x\to \infty$, we may find $m_\beta$ such that $0<u'(x) < \dfrac{\beta_1}{\beta}$, for all $x\geq \ln m_\beta$.
Therefore,  $\Phi_N(x) \leq x^{-\beta_1}e^{-x}$ for $\frac{m_\beta}N<x\leq 1$ and $\Phi_N(x) \leq e^{-x}$ for $x>1$.

By dominated convergence we obtain $\int_{m_\beta /N}^{\infty} \Phi_N(x)dx \to \int_0 ^{\infty} x^{-\beta}e^{-x}dx$.
Also, the remaining term $\int_0^{m_\beta/N} \Phi_N(x)dx \leq \frac{m_{\beta}N_\beta}{N}$ vanishes as $N_\beta = o(N)$ for $0<\beta<1$.  This completes the argument.
\end{proof}

\begin{Lem} \label{invariant measures are local equli}
For all $c$ such that $0\leq c < c_0$, we have 
\begin{equation}
\label{total var difference}
\lim_{N\to \infty}\dfrac 1N \sum_{k=1}^{\infty} 
 \left|
 N_{\beta}\rho_{k,c} - \overline \rho_{N,k}
 \right | =0,
 \end{equation}
  where $\overline \rho_{N,k}= N\int_{(k-1)/N}^{k/N} \phi_c(x) dx$.  As an immediate consequence,
the product invariant measures $\left\{\Pam_{c,N}\right\}_{N\in \N}$, with Geometric marginals, are local equilibrium measures corresponding to $\rho_0=\phi_c$.
\end{Lem}
\begin{proof}
Recall that $\theta_{k,c} =c e^{-\beta \Eb_k - k/N}$ and $E_{\Pam_{c,N}}\eta(k) = \rho_{k,c} = \theta_{k,c}/(1-\theta_{k,c})$ (cf. \eqref{rho k c}).
We now verify the limit \eqref{total var difference}. 
When $\beta=0$, we have $N_\beta =1$ and $\phi_c = \dfrac{ce^{-x}}{1- c e^{-x}}$.
Since $\phi_c$ is decreasing, we have $ \rho_{k,c}=\dfrac{ce^{-k/N}}{1- c e^{-k/N}} < \overline \rho_{N} $.
Then, the left-hand side of \eqref {total var difference} equals to
$$
\int_0^\infty \dfrac{ce^{-x}}{1- ce^{-x}}dx
-
\lim_{N\to \infty}\dfrac 1N \sum_{k=1}^{\infty}  \dfrac{ce^{-k/N}}{1-ce^{-k/N}}, 
$$
which clearly vanishes as $N\to \infty$ by dominated convergence.

For the remaining two regimes when $\beta>0$, 
we will split the summation in \eqref{total var difference} into two parts: $aN \leq k \leq bN$ and the rest, for an $0<a<b$ that we will specify. 
In fact, it will be enough to show, for any $\e>0$, that we can find $a>0$ small enough and $b>0$ big enough such that 
\begin{equation} \label {sum: leq aN}
\lim_{N\to\infty}\dfrac 1N \sum_{\substack{1\leq k \leq aN\\ k> bN} }
 \left|
 \dfrac{N_{\beta}\theta_{k,c}}{1- \theta_{k,c}} - \overline \rho_{N,k}
 \right |
 \leq \e 
\end{equation} 
and, for all $b>a>0$, that
\begin{equation}  \label {sum: geq aN}
\lim_{N\to\infty}\dfrac 1N \sum_{k=aN}^{bN}
 \left|
 \dfrac{N_{\beta}\theta_{k,c}}{1- \theta_{k,c}} - \overline \rho_{N,k}
 \right |
=0.
\end{equation}

To verify \eqref{sum: leq aN},
$$
\dfrac 1N \sum_{\substack{1\leq k \leq aN\\ k> bN} } 
 \left|
 \dfrac{N_{\beta}\theta_{k,c}}{1- \theta_{k,c}} - \overline \rho_{N,k}
 \right | 
 \leq
 \dfrac 1N \sum_{\substack{1\leq k \leq aN\\ k\geq bN} }
 \dfrac{N_{\beta}\theta_{k,c}}{1- \theta_{k,c}}
 + 
 \int_{(0,a)\cup(b,\infty)} \phi_c dx .
$$
Recall $c_0 = \min_k e^{\beta \Eb_k}$.  Since $\theta_{k,c} = c e^{-\beta \Eb_k - k/N} \leq \dfrac {c}{c_0}$, by Lemma \ref{lem: theta k limit to rho}, we have 
$$ \dfrac 1N \sum_{\substack{1\leq k \leq aN\\ k> bN} }
  \dfrac{N_{\beta}\theta_{k,c}}{1- \theta_{k,c}}
\leq \dfrac{c_0}{c_0 - c}  \dfrac 1N
 \sum_{\substack{1\leq k \leq aN\\ k> bN} }  N_{\beta}\theta_{k,c}
  \to \dfrac{c_0}{c_0 - c}  \int_{(0,a)\cup(b,\infty)} \phi_c dx.$$
Then, \eqref{sum: leq aN} follows as $\phi_c \in L^1(\R^+)$.

It remains to show \eqref {sum: geq aN}. By adding and subtracting, for each $N$ the left side of  \eqref{sum: geq aN} is bounded by
\begin{equation*}
\begin{split}
\dfrac 1N \sum_{k=aN}^{bN} 
 \left|
 \dfrac{N_{\beta}\theta_{k,c}}{1- \theta_{k,c}} - N_{\beta}\theta_{k,c}
 \right |
 +
 \dfrac 1N \sum_{k=aN}^{bN} 
  \left|
 N_{\beta}\theta_{k,c} - c (k/N)^{-\beta u'(\infty)} e^{-k/N}
 \right |
 \\
 +\dfrac 1N \sum_{k=aN}^{bN} 
  \left|
 c (k/N)^{-\beta u'(\infty)} e^{-k/N} - \overline \rho_{N,k}
 \right |
 =: I_1 + I_2 + I_3
\end{split}
\end{equation*}
where $u'(\infty) = \lim_{x\to \infty} u'(x)$ takes value either $0$ or $1$.

The term $I_1$ is trivially bounded by
$$
\max_{aN\leq k\leq bN}\dfrac {\theta_{k,c}}{1-\theta_{k,c}} N^{-1 }\sum_{k=aN}^{bN} N_{\beta}\theta_{k,c}.
$$
Recall that $\theta_{k,c} =c e^{-\beta \Eb_k - k/N}$ and $\Eb_k\to \infty$ as $k\to\infty$. Then, $\max_{aN\leq k\leq bN}\dfrac {\theta_{k,c}}{1-\theta_{k,c}}$ vanishes as $N\to \infty$.
Since also $N^{-1 }\sum_{k=aN}^{bN} N_{\beta}\theta_{k,c}\to \int_a^b \phi_c dx <\infty$ (Lemma \ref{lem: theta k limit to rho}) is bounded, the term $I_1$ vanishes.

For term $I_2$, we spell out $N_\beta \theta_{k,c}$ as 
$$
ce^{-\beta (\Eb(k) - \Eb(N)) -k/N}.
$$ 
By the mean value theorem, we have $\Eb(k) - \Eb(N) = \ln \dfrac kN u'(y_{k,N})$, where $y_{k,N}$ is in between $\ln k$ and $\ln N$.
Then, $I_2$ is less than or equal to
$$
 \max_{aN\leq k\leq bN}\left\{ \left| ( k/N)^{\beta (u'(y_{k,N}) - u'(\infty))}  -1\right| \right\}
 \dfrac 1N  \sum_{k=aN}^{bN}  N_{\beta}\theta_{k,c}.
$$
We observed in estimating $I_1$ above that $N^{-1}\sum_{k=aN}^{bN}N_\beta \theta_{k,c}$ is bounded.  Hence, $I_2$ vanishes as $N\to \infty$.
 
We now address the last term $I_3$.  Observe, as $\phi_c$ is decreasing, that 
$$
I_3 \ = \ \int_{a-\frac1N}^b \phi_c(x)dx
-
\dfrac 1N \sum_{k=aN}^{bN} 
 c (k/N)^{-\beta u'(\infty)} e^{-k/N},
$$
which vanishes as $N\to \infty$ by the dominated convergence theorem.
\end{proof}

We now give an useful mean and variance estimate.
\begin{Lem}
\label{Lem: mean variance mu N}
For all $c$ such that $0\leq c < c_0$
we have that
\begin{eqnarray} 
&&E_{\Pam_{c,N}} \sum_{k=1}^\infty \eta (k) = \sum_{k=1}^\infty \rho_{k,c} = O\big(N N^{-1}_\beta\big)
\ \ {\rm and \ \ }\nonumber \\
&&\ \ \ \ \ \ \ \ \ \ \ \ \ \ \ \ \ \ \ \ \ \ \ \ \frac{N_\beta^2}{N^2}\sum_{k=1}^\infty \text{Var}_{\Pam_{c,N}} (\eta(k)) = \frac{N_\beta^2}{N^2} \sum_{k=1}^\infty \left[ \rho^2_{k,c} + \rho_{k,c}\right]\rightarrow 0.
\label{eqn: particle number estimate}
\end{eqnarray}
\end{Lem}

\begin{proof}
We first consider the means:
$$
\frac{N_\beta}{N}E_{\Pam_{c,N}}\sum_{k=1}^\infty \eta(k)
=
\dfrac {N_\beta}N \sum_{k=1}^{\infty} \rho_{k,c}.
$$ 
By \eqref{total var difference}, $ \lim_{N\to\infty} \dfrac 1N \sum_{k=1}^{\infty} |N_{\beta} \rho_{k,c} - \overline \rho_{N,k}| =0$, where $\overline \rho_{N,k} = N\int_{(k-1)/N}^{k/N}\rho_c(x)dx$.
As $\dfrac1N \sum_{k=1}^\infty   \overline \rho_{N,k} = \int_0^{\infty} \phi_cdx <\infty$, then the estimate on sum of means in \eqref{eqn: particle number estimate} follows.

Next, we consider the variances.
Since $N_\beta = o(N)$ and $N^{-1}\sum_{k=1}^\infty N_\beta \rho_{k,c}<\infty$ by the first estimate in \eqref{eqn: particle number estimate}, 
we have that $N_\beta^2 N^{-2} \sum_{k=1}^{\infty} \rho_{k,c}$ vanishes as $N\to \infty$.
For the term $N_\beta^2 N^{-2} \sum_{k=1}^{\infty} \rho_{k,c}^2$, we 
use $\sum (N_\beta\rho_{k,c})^2 \leq 2( \sum | N_\beta \rho_{k,c} -\overline \rho_{N,k}| )^2 + 2\sum\overline \rho_{N,k}^2 $.
Since $N^{-1}\sum |N_\beta \rho_{N,k} - \overline\rho_{N,k}| \to 0$, it suffices to show that
$\lim_{N\to\infty} \dfrac1{N^2} \sum_{k=1}^{\infty} \overline\rho_{N,k}^2 =0$.

To this end, let $\hat \rho_N = \max_{k\geq1} \overline\rho_{N,k}$. 
Then, $ \dfrac1{N^2} \sum_{k=1}^{\infty} \overline\rho_{N,k}^2 \leq \dfrac{\hat \rho_N}{N^2}\sum_{k=1}^{\infty} \overline\rho_{N,k}$.
Now, $N^{-1}\sum_{k=1}^\infty \overline\rho_{N,k}=\int_0^{\infty}\phi_cdx < \infty$.  The desired limit holds since, by absolute continuity of the Lebesgue integral, $N^{-1} \hat \rho_N\to 0$ as $N\to \infty$.
\end{proof}


\subsection{Properties of local equilibria $\mu^N$} \label {subsection: mu N}

We now observe that the local equilibria $\mu^N$ satisfy Condition \ref{condition 1}.
\begin{Prop} \label {prop: mu N in nu N}
Local equilibrium measures $\mu^N$ satisfy Condition \ref{condition 1}.
\end{Prop}
\begin{proof}
First, by the definition of $\mu^N$,
parts (1) and (2) of Condition \ref{condition 1} are met.
In Lemma \ref{entropy: ln k} below, we show that the relative entropy estimate, part (3), holds.
\end{proof}

\begin{Lem}  \label {entropy: ln k}
Fix $c$ such that $0\leq c\leq c_0$ and assume that $\mu^N\leq \Pam_{c,N}$ .
Then there exists a constant $C$ such that  $H(\mu^N | \Pam_{c,N}) \leq CN N_{\beta}^{-1}$ holds for all $N$.
\end{Lem}
\begin{proof}
Let $\zeta$ and $\chi$ be two Geometric distributions with rate $p$ and $q$ respectively.
Assuming $p \leq q$ we have
\begin{equation*} \label {relative entropy two geometrics}
H(\zeta \vert \chi)
= \sum_{n\geq 0} (1-p)^n \ln \frac{1-p}{1-q} \frac{p^n}{q^n} = 
\ln \dfrac{1-p}{1-q}
+
\dfrac{p}{1-p} \ln \dfrac{p}{q}
\leq
\ln \dfrac{1}{1-q}.
\end{equation*}

Suppose now, for $k\geq 1$, that $\zeta=\mu^N_{k}$ and $p= \theta_{N,k}$ and $\chi = \Pam_{\beta,c,N,k}$ and $q = c\theta_k$.  Note, by assumption, that $\theta_{N,k} \leq c\theta_k =ce^{-\beta u(\ln k) -k/N}$.  Then, as $\mu^N$ and $\Pam_{c,N}$ is the product over $\{\mu^N_{k}\}_{k\geq 1}$ and $\{\Pam_{\beta, c,N,k}\}_{k\geq 1}$ respectively, we have
\begin{equation*}
H(\mu^N | \Pam_{c,N}) \leq \sum_{k=1} ^{\infty} \ln \dfrac{1}{1- c_0e^{-\beta u(\ln k)-k/N}}.
\end{equation*}
When $\beta = 0$, we have 
$$ H(\mu^N | \Pam_{c,N}) \leq \sum_{k=1} ^{\infty} \ln \dfrac{1}{1-e^{-k/N}}
\leq -N\int_0^{\infty} \ln (1-e^{-x}) dx =: CN.$$ 
For the cases $\beta >0$, we recall that $c_0 = \min_k e^{\beta \Eb_k}$.  Let $K_0=\{k_{0,j}\}_{1\leq j \leq n}$ be the indices
where $c_0$ is attained.  
The contribution from each $k_{0,j}$ to the relative entropy $H(\mu^N | \Pam_{c,N})$ is bounded above by
\begin{equation*}
\ln \dfrac1{1-e^{-k_{0,j}/N}} = O(\ln N).
\end{equation*}
This order is negligible compared with $NN_{\beta}^{-1}= Ne^{-\beta \Eb_N}=Ne^{-\beta u(\ln N)}$ in the two cases when $u'(\ln N) \to 1$ and $0<\beta<1$ or when $u'(\ln N) \to 0$ and $\beta>0$.  
We will be able to disregard later these $k_{0,j}$'s.

Now, as $u(\ln k) \to \infty$ as $k\to \infty$, find $0<\alpha<1$ such that $0< c_0e^{-\beta u(\ln k)-k/N} \leq \alpha$ for all $N$ and $k\notin K_0$.
Using convexity of $-\ln(1-x)$, there exists $c_1>0$ such that $-\ln(1-x) \leq c_1 x$ on $[0,\alpha]$.
Then, we have
\begin{equation*}
\begin{split}
\sum_{k=1} ^{\infty} \ln \dfrac{1}{1- c_0e^{-\beta u(\ln k)-k/N}}
\leq
c_1c_0 \sum_{k=1} ^{\infty}  e^{-\beta u(\ln k)-k/N}
+
O(\ln N).
\end{split}
\end{equation*}
Multiplying and dividing by the term $NN_{\beta}^{-1}$, we get
\begin{equation} \label {eqn: estimate for general relative entropy}
\begin{split}
H(\mu^N | \Pam_{c,N}) \leq
c_1c_0 NN_{\beta}^{-1} 
\left[
\sum_{k=1} ^{\infty}\dfrac 1N e^{-\beta (u(\ln k) - u(\ln N))  -k/N}
+
\ O(N^{-1} N_\beta \ln N)
\right].
\end{split}
\end{equation}
Now, $N^{-1} N_\beta \ln N$ vanishes as $N\to \infty$
and by Lemma \ref{lem: theta k limit to rho} the summation in \eqref{eqn: estimate for general relative entropy} approaches a finite limit.
The proof is now complete.
\end{proof}

\subsection{Properties of $\nu^N$ satisfying Condition \ref{condition 1}} \label {subsec: nu N}
We will establish the items \eqref{total number}, \eqref{nu N var bound}, \eqref{site_particle}, and \eqref{Initial convergence nu N}. 
We start with an estimate on the number of particles in the system.
\begin{Lem}
\label{lem: total number}
We have that `the total expected particle bound' \eqref{total number} holds.
\end{Lem}
\begin{proof} 
Since the total number of particles is conserved we have 
$$
\frac{N_\beta}{N}\E_N\sum_{k=1}^\infty \eta_t(k)
=
\frac{N_\beta}{N}\E_N\sum_{k=1}^\infty \eta_0(k) 
= \frac{1}{N}\sum_{k=1}^\infty N_\beta m_{N,k} = o(1) + \frac{1}{N}\sum_{k=1}^\infty \overline\rho_{N,k}.$$
by Condition \ref{condition 1}.
However,
$N^{-1}\sum_{k=1}^\infty \overline\rho_{N,k} = \int_0^\infty \rho_0(x)dx$, which is finite.
\end{proof}

\begin{Lem} \label{lem: var bound}
We have that the `variance bound' \eqref{nu N var bound} holds.
\end{Lem}
\begin{proof}
By attractiveness \eqref{attractiveness},
 $$\text{Var}_{\P_N}(\eta_t(k)) = \E_N[\eta_t^2(k)] -\left(\E_N \eta_t(k)\right)^2 \leq E_{\Pam_{c,N}}[\eta^2(k)]  \leq  \text{Var}_{\Pam_{c,N}}(\eta(k)) +\rho^2_{k,c}.$$

 Then, by Lemma \ref{Lem: mean variance mu N}, we conclude that
 $N_{\beta}^2N^{-2} \sum_{k=1}^{\infty} \text{Var}_{\P_N}(\eta_t(k)) \to 0$ as $N\to \infty$. 
\end{proof}

\begin{Lem}\label{lem: site particle bound}
We have that the `site particle bound' \eqref{site_particle} holds.
\end{Lem}
\begin{proof}
First, by attractiveness \eqref{attractiveness}, we have that
$\E_N \big[\eta_t(k)\big] \leq E_{\Pam_{c,N}} \big[\eta(k)\big] = \rho_{k,c}$ (cf. \eqref{rho k c}).  
To bound $N_\beta \rho_{k,c}$, recall that $c_0=\min_k e^{\beta \Eb_k}$ and $c<c_0$.

When $\beta=0$, we have $c_0=1$ and $N_\beta =1$. In this case, we have the desired bound, $N_\beta \rho_{k,c} \leq \dfrac{e^{-a}}{1-e^{-a}}$
for all $k\geq aN$. 

When $\beta> 0$, 
using definition of $c_0$, and that $c<c_0$, we have the denominator $1- ce^{-\beta \Eb_k -k/N} \geq 1-e^{-a}$ as $k\geq aN$. Write $N_\beta e^{-\beta \Eb_k - k/N} \leq e^{-\beta (\Eb_k - \Eb_N) -a}$.  By the mean value theorem, $\Eb_k - \Eb_N = u'(r) \ln (k/N)$ where $r$ is between $aN\leq k$ and $N$.  By assumption, $u'(r)$ tends to $0$ or $1$, and $\ln (k/N) \leq \ln b$ for $k\leq bN$.  We conclude then that $N_\beta e^{-\beta \Eb_k - k/N}$ is uniformly bounded in $N$, and the lemma follows.
\end{proof}

We now address initial convergence.
\begin{Prop} \label{prop: initial convergence and mass for condition 1}
We have `initial convergence' \eqref{Initial convergence nu N} holds. 
\end{Prop}
 
\begin{proof}
By assumption, $ \lim_{N\to\infty} \dfrac 1N \sum_{k=1}^{\infty} |N_{\beta} m_{N,k} - \overline \rho_{N,k}| =0$.  For a test function $G$, 
since $N^{-1}\sum_{k=1}^{\infty} G(k/N)\overline \rho_{N,k}$ approximates $\int_{\R^+} G(x)\rho_0(x) dx$,
it is enough to check that
\begin{equation*} \label {eqn: static estimate a b spelt out }
\nu^N
\left[ \left|
N^{-1}
\sum_{k=1}^{\infty}
 N_{\beta} (\eta(k) - m_{N,k})
\right|
>\delta
 \right].
\end{equation*}
By Chebychev's inequality, we have the upperbound of 
 $ \delta^{-2}  N_{\beta}^2N^{-2} \sum_{k=1}^{\infty} \text{Var}_{\nu^N}(\eta(k))$,
 which vanishes by the variance bound in Lemma \ref{lem: var bound}.
 \end{proof}

\section{Uniqueness of weak solutions} \label {section: uniqueness}
In this section, we present some uniqueness results for 
the macroscopic equations in Theorems \ref{thm: beta 0}, \ref{thm: ln k} and \ref{thm: lnln k}, governing the particle density $\rho(t,x)$ or the height function $\psi(t,x) := \int_x^{\infty} \rho(t,u)du$.
The methods are based on maximum principles for linear parabolic equations.

We first need a lemma to relate properties of $\psi$ with those of $\rho$. 
Recall that
 $\C$ is space of functions $\rho: [0,T]\times \R^+\mapsto \R^+$  such that $t\in [0,T]\mapsto \rho(t,x)dx\in \Mb$ is vaguely continuous: Namely, for each $G\in C_c^\infty(\R^+_\circ)$, the map $t\in [0,T]\mapsto \int_0^\infty G(x)\rho(t,x)dx$ is continuous.

Also, recall
$$\Wb = \left\{
\psi \in C\left( [0,T] \times \R^+ \right): \text{ for } t\in [0,T], \ \psi(t,\cdot) \text{ is absolutely continuous on } \R^+_{\circ}
\right\}.$$
\begin{Lem}
\label{lem_uniqueness}
Let $\rho(t,x) \in \C$. Suppose, for all $t\in[0,T]$, that
\begin{equation} \label {eqn rho macro prop}
\rho(t,\cdot) \leq \phi_c(\cdot)\in L^1(\R^+),
\quad
 \int_0^{\infty} \rho(t,x) dx=\int_0^{\infty} \rho_0(x) dx<\infty.
\end{equation}
Let $\psi(t,x) = \int_x^{\infty} \rho(t,u)du$.
Then, $\psi(t,x)$ belongs to the class $\Wb$ with
\begin{equation} \label{eqn psi macro prop}
\begin{split}
\lim_{x\to\infty} \psi(t,x) = 0,
\quad
0\leq -\partial_x\psi(t,\cdot) \leq \phi_c(\cdot),
\quad
\psi(t,0) =\psi(0,0).
\end{split}
\end{equation}
\end{Lem}
\begin{proof}
The absolute continuity of $\psi(t,\cdot)$ follows from definition of $\psi$ and it is trivial to verify \eqref{eqn psi macro prop} from  \eqref{eqn rho macro prop}.
To finish, we need only to check that $\psi(t,x)$ is a continuous function on $[0,T] \times \R^+$. 

We claim
 that such continuity will follow if $\psi$ is continuous in $x$ and $t$ separately.
Indeed, fix any $(t_0,x_0) \in (0,T) \times \R^+_{\circ}$ and denote $\psi(t_0,x_0) = a_0$. 
If $x\mapsto \psi(t,x)$, for each $t$, is continuous at $x_0$, then for any $\epsilon>0$ there exists $\delta$ such that 
\begin{equation*}
a_0-\epsilon \leq \psi(t_0,x_0\pm \delta) \leq a_0+\epsilon.
\end{equation*}
Suppose $t\mapsto \psi(t, x)$, for each $x$, is continuous in $t$, then we may find $\delta'$, such that for all $t$ where $|t-t_0| \leq \delta'$, we have
\begin{equation*}
\psi(t_0,x_0\pm \delta) - \epsilon \leq \psi(t,x_0\pm \delta) \leq \psi(t_0,x_0\pm \delta) +\epsilon.
\end{equation*}
Since $x\mapsto \psi(t,x)$, for each $t$,  is monotone in $x$, we have, for all $(t,x)$ such that $|t-t_0|\leq \delta'$ and $|x-x_0|\leq \delta$, that 
\begin{equation*}
- 2\epsilon \leq \psi(t,x) -\psi(t_0,x_0\pm \delta)   \leq 2 \epsilon.
\end{equation*}
Hence, we deduce continuity of $\psi$ at $(t_0,x_0)$.
Continuity for boundary points $(t,x)$ on the boundary is verified in the same way.

Now, we focus on showing that $t\mapsto \psi(t,x)$ and $x\mapsto\psi(t,x)$ are both continuous. 
For any fixed $t\in [0,T]$, $x\mapsto \psi(t,x)$ is continuous on $\R^+$ since $\psi$ is in form $\psi(t,x) = \int_x^{\infty} \rho(t,u) du$ 
and $\int_0^{\infty}\rho(t,u)du < \infty$.

To show continuity in $t$, we first note that $\psi(t,0) = \psi(0,0)$ for all $t\in [0,T]$, and therefore $t\mapsto\psi(t,0)$ is continuous.
Fix now any $x_0>0$ and $t_0\in [0,T]$. For any $\epsilon > 0$, using $\rho(t,x)\leq \phi_c(x)$ and that $\psi_c\in L^1(\R^+)$, we may find $G$ continuous and with compact support in $\R^+_{\circ}$ 
such that for all $t\in [0,T]$,
\begin{equation*}
\left |   
\int_x^\infty \rho(t,u) du 
-
\int_0^{\infty} G(u) \rho(t,u) du 
  \right |
  \leq \dfrac{\epsilon}{4}.
\end{equation*} 
Then, by the triangle inequality using two applications of the above inequality, we have $|\psi(t,x_0) - \psi(t_0,x_0) |$ is bounded from above by
\begin{equation*}
\left |   
\int_0^\infty G(u) \rho(t,u) du 
-
\int_0^{\infty} G(u) \rho(t_0,u) du 
  \right |
  +
 \dfrac{\epsilon}{2}.
\end{equation*}
Finally,  continuity of $t\mapsto\psi(t,x_0)$ at $t_0$ follows as $\rho\in \C$, namely from the vague continuity of $\rho(t,x)dx$.
\end{proof}

\subsection{Case: $\beta=0$} \label{section uniqueness beta 0}
Let $\rho(t,x) \in \C$ with $\rho(0,\cdot) = \rho_0(\cdot)$ be a weak solution of the equation
\begin{equation*}
\partial_t \rho
= \partial_x^2 \dfrac{\rho}{\rho+1} + \partial_x \dfrac{\rho}{\rho+1},
\end{equation*}
that is, 
for all $G \in C_c^{\infty} ([0,T) \times \R^+_{\circ})$,
\begin{equation}
\label {eqn: beta=0, rho, weak form}
\int_0^{\infty}
G(0,x)\rho_0 dx
+
\int_0^T \int_0^{\infty}
\left\{
\partial_t G  \rho
+
\partial_x^2 G
 \dfrac{\rho}{\rho+1}
-
\partial_x G 
 \dfrac{\rho}{\rho+1}
 \right\}
  dxdt
 =
 0.
\end{equation}
Assume also that $\rho(t,x)$ satisfies \eqref{eqn rho macro prop}.

\begin{Prop} We have
$\psi(t,x) = \int_x^\infty \rho(t,u)du$ belongs to $\Wb$ and \eqref {eqn psi macro prop} holds by Lemma \ref{lem_uniqueness}.  In particular, $\psi$ solves weakly the equation 
\begin{equation*} 
\label{eqn: beta=0, psi}
\partial_t \psi
= 
\partial_x \left( \dfrac{\partial_x \psi}{1 - \partial_x \psi} \right)+   \dfrac{\partial_x \psi}{1 - \partial_x \psi},
\end{equation*}
that is, for all $G \in C_c^{\infty} ([0,T) \times \R^+_{\circ})$
\begin{equation}  \label{eqn: beta=0, psi, weak form}
\int_0^{\infty}
G(0,u)\psi_0 dx
+
\int_0^T \int_0^{\infty}
\left\{
\partial_t G  \psi
-
\partial_x G
 \dfrac{\partial_x\psi}{1-\partial_x\psi}
+
 G 
 \dfrac{\partial_x\psi}{1-\partial_x\psi}
 \right\}
  dxdt
 =
 0,
\end{equation}
where $\psi_0(x) = \int_x^{\infty} \rho_0(u)du$.

Moreover, $\psi(t,x)$ is the unique weak solution in the class $\Wb$ of the initial-boundary value problem  \eqref {eqn: macro beta 0, psi}. 
Consequently, $\rho(t,x)$ is the unique weak solution in $\C$ of the equation \eqref{eqn: macro beta 0}.
\end{Prop}

\begin{proof} We first show \eqref{eqn: beta=0, psi, weak form}.  Since $\rho(t,x)\leq \phi_c(x)\in L^1(\R^+)$ (cf. \eqref{eqn rho macro prop}), by straightforward approximations, the test functions admissible for \eqref{eqn: beta=0, rho, weak form} may be extended to include all functions of the form $\widehat G(t,x) = \int_0^x G(t,u)du$ where $G \in C_c^{\infty} ([0,T) \times \R^+_{\circ})$. Then, by integration by parts, \eqref{eqn: beta=0, psi, weak form} follows.

We now show $\psi(t,x)$ is the unique weak solution to \eqref {eqn: macro beta 0, psi} in the space $\Wb$.
Suppose there exist two such weak solutions $\psi_1$, $\psi_2$.
 Let $\psi = \psi_1 - \psi_2$ and $H(p) = \dfrac{p}{1-p}$.
As \eqref{eqn: beta=0, psi, weak form} holds for $\psi_1$, $\psi_2$, in the new notation, we have
\begin{equation*}
\int_0^T \int_0^{\infty}
\left\{
\partial_t G  \psi
-
\partial_x G
\left(
H \left(\partial_x \psi_1  \right)
-
H \left(\partial_x \psi_2  \right)
\right)
+
  G 
 \left(
H \left(\partial_x \psi_1  \right)
-
H \left(\partial_x \psi_2  \right)
\right)
 \right\}
  dxdt
 =
 0
\end{equation*}
and
\begin{equation*}
\begin{split}
H (\partial_x \psi_1 )
-
H (\partial_x \psi_2 )
&=
(\partial_x \psi_1 -\partial_x \psi_2)
\int_0^1 H'(\tau \partial_x \psi_1 + (1-\tau)\partial_x \psi_2) 
d\tau
\\
&=:
(\partial_x \psi_1 -\partial_x \psi_2)
\widehat{H}(t,x).
\end{split}
\end{equation*}
Then, $\psi$ satisfies
\begin{equation*}
\int_0^T \int_0^{\infty}
\left\{
\partial_t G  \psi
-
\partial_x G
\left(
\widehat{H}(t,x)
\partial_x \psi
\right)
+
G 
\left(
\widehat{H}(t,x)
\partial_x \psi
\right)
 \right\}
  dxdt
 =
 0,
\end{equation*}
that is, $\psi$ is a weak solution in $\Wb$ of the linear problem
\begin{equation}
\label{eqn: hat H}
\begin{cases}
\partial_t \psi
= 
\partial_x \left( \widehat{H} \partial_x \psi  \right)
+
 \widehat{H} \partial_x \psi\\
 \psi(0,x)=0,
 \quad
 \psi(t,0)=0,\\
 \lim_{x\to\infty} \psi(t,x) = 0,
 \quad
 -\phi_c(\cdot)\leq \partial_x \psi(t,\cdot) \leq \phi_c(\cdot) \text{ for all } t\in [0,T].
 \end{cases}
\end{equation}
To show that $\psi \equiv 0$, and therefore uniqueness of weak solution. it suffices to show, for all $\e>0$ and all compact set $D\subset (0,T)\times\R_\circ^+$, that $|\psi| < \e$ on $D$.

For such a $D$, we may find 
$0<a<b<\infty$ where $D\subset Q^T_{a,b}:=(0,T) \times (a,b)$.
Since $\big|\partial_x \psi(t,\cdot)\big| \leq \phi_c(\cdot) \in L^1(\R^+)$  for all $t\in [0,T]$, and  $\psi$ vanishes for both $x=0$ and $x\rightarrow\infty$, we can adjust $a$, $b$
so that $|\psi(t,a)| <\e$ and $|\psi(t,b)| <\e$ for all $t\in[0,T]$. Then, we have $|\psi| <\e$ on the parabolic boundary of $Q^T_{a,b}$.

Notice that, on $Q^T_{a,b}$, the PDE in \eqref{eqn: hat H} is uniformly parabolic and has bounded coefficients:
Since
$ H'(p) = \dfrac{1}{(1-p)^2} $
and  $-\phi_c(a)\leq \partial_x \psi_1, \partial_x \psi_2 \leq 0$ on $Q^T_{a,b}$,
we have
\begin{equation*}
\dfrac{1}{(1+\phi_c(a))^2}
\leq
\widehat{H}
\leq 1
\ \ \text{ on } Q^T_{a,b}.
\end{equation*}
Then, by a maximum principle (cf.\,p.\,188, \cite{LSU}), we have $|\psi| < \e$ on $Q^T_{a,b}$, and therefore on $D$.

Finally, if $\rho(t,x)$ were not unique with respect to \eqref{eqn: beta=0, rho, weak form}, one could construct two different weak solutions $\psi(t,x)$, which is a contradiction.
\end{proof}

\subsection{Case $\beta> 0$} \label{section uniqueness beta neq 0}
Let $\rho(t,x)\in \C$ with $\rho(0,\cdot)=\rho_0(\cdot)$ be a weak solution of 
\begin{equation} 
\label{eqn: beta neq 0, pde}
\partial_t \rho
=
\partial_x^2\rho 
-
\partial_x \Big(\alpha(x,\beta)\rho\Big).
\end{equation}
where $\alpha(x,\beta) = -(\beta +x)/x$ when $\Eb_k \sim \ln k$ and equals $-1$ when $1\ll \Eb_k \ll \ln\ln k$ (cf. \eqref{alpha_D}), and $\rho$ satisfies \eqref{eqn rho macro prop}.

\begin{Prop}
We have $\psi(t,x)=\int_x^\infty \rho(t,u)du$ belongs to $\Wb$ and \eqref{eqn psi macro prop} holds by Lemma \ref{lem_uniqueness}, and solves weakly the equation 
\begin{equation} \label {eqn: beta neq 0, psi pde}
\partial_t \psi
=
\partial_x^2\psi 
-
\alpha(x,\beta)\partial_x  \psi,
\end{equation}
where $\psi(0,x) = \int_x^{\infty} \rho_0(u) du$.

Then, $\psi(t,x)$ is the unique weak solution in $\Wb$ of the initial-boundary value problem \eqref {eqn: macro ln k, psi} when $\Eb_k \sim \ln k$, and of \eqref{eqn: macro lnln k, psi} when $1\ll \Eb_k \ll  \ln k$. 
Consequently, $\rho(t,x)$ is the unique weak solution in $\C$ of the equation \eqref{eqn: macro ln k} when $\Eb_k \sim \ln k$ and of \eqref{eqn: macro lnln k} when $1\ll \Eb_k \ll  \ln k$.
\end{Prop} 

\begin{proof}
That $\psi$ solves weakly \eqref{eqn: beta neq 0, psi pde} follows, as in the proof of Lemma \ref{lem_uniqueness}, from the assumptions $\rho$ is a weak solution of \eqref{eqn: beta neq 0, pde} and $\rho\leq \phi_c$. 

Notice that, in equation \eqref{eqn: beta neq 0, psi pde}, the coefficient $-\alpha(x,\beta)$ before $\partial_x \psi$ equals $\dfrac{\beta + x}{x}$ when $\Eb_k \sim \ln k$ and equals $1$ when $1\ll \Eb_k \ll \ln k$.  In both situations, it is bounded on any $[a,b]$ with $0<a<b<\infty$, even if it blows up at $x=0$ when $\Eb_k \sim \ln k$.
Then, the same proof of uniqueness given for Lemma \ref{lem_uniqueness} applies to show uniqueness of weak solutions for the equations \eqref{eqn: macro ln k, psi}  and \eqref{eqn: macro lnln k, psi}.
\end{proof}


\appendix
\section{Remarks on limits when $c=c_0$}
\label{appendix}

We now make remarks, for the interested reader, on some of the behavior with respect to measures $\Pam_{c,N}$ at the boundary, when $c=c_0$.
\vskip .1cm

{\bf 1.}
Lemma \ref {Lem: mean variance mu N} does not hold for invariant measure $\Pam_{c_0,N}$.
In fact, under $\Pam_{c_0,N}$, the total number of particles explodes and the associated variance does not vanish in the limit.

\begin{Lem}
\label{app_1}
We have
\begin{equation} \label{eqn: c=c_0 diverge}
\dfrac { N_\beta} { N} 
\sum_{k=1}^{\infty}  E_{\Pam_{c_0,N}} (\eta(k))
=
\dfrac { N_\beta} { N} 
\sum_{k=1}^{\infty}  \rho_{c_0,k}
\to
\infty
\text{ as } N\to \infty.
\end{equation}
and 
\begin{equation}\label{eqn: c=c_0 variance not vanish}
\liminf_{N\to \infty} \dfrac {N_\beta^2} {N^2} 
\sum_{k=1}^{\infty}  \text{Var}_{ \Pam_{c_0,N}} (\eta(k))
=
\liminf_{N\to \infty} \dfrac { N_\beta^2} { N^2} 
\sum_{k=1}^{\infty}  (\rho_{k,c_0}^2 + \rho_{k,c_0})
>0.
\end{equation}
\end{Lem}

\begin{proof}
To verify these two claims, recall that $\rho_{k,c_0} =  \dfrac{c_0e^{-\beta \Eb_k-k/N}}{1-c_0e^{-\beta \Eb_k-k/N}}$ and $c_0 =\min_{k} e^{\beta \Eb_k}$.

When $\beta = 0$, \eqref{eqn: c=c_0 diverge} and \eqref{eqn: c=c_0 variance not vanish} follow from the limits,
$$
\dfrac { N_\beta} { N} 
\sum_{k=1}^{\infty}  \rho_{c_0,k}
=
\dfrac 1 N 
\sum_{k=1}^{\infty} \dfrac{e^{-k/N}}{1-e^{-k/N}}
=
\sum_{k=1}^{\infty} \dfrac{1}{N(e^{k/N}-1)}
\to
\infty,
$$
and
$$
 \dfrac { N_\beta^2} { N^2} 
\sum_{k=1}^{\infty} \rho_{k,c_0}^2
=
 \dfrac1 { N^2} 
\sum_{k=1}^{\infty} \Big( \dfrac{e^{-k/N}}{1-e^{-k/N}} \Big)^2
\geq
 \dfrac1 { N^2} 
 \Big( \dfrac{1}{e^{1/N} - 1} \Big)^2 \to 1.
$$

For the other two cases, when $\beta> 0$, let $k_0$ be an index where $c_0$ is realized, that is \,$c_0 = e^{\Eb_{k_0}}$. 
Now notice, as $N\rightarrow\infty$,
$$
\dfrac{1}{N} \rho_{k_0,c_0}= \dfrac{1}{N} \dfrac{e^{-k_0/N}}{1-e^{-k_0/N}}\to \frac{1}{k_0}.
$$
Then, both \eqref{eqn: c=c_0 diverge} and \eqref{eqn: c=c_0 variance not vanish} follow
from 
$$
\dfrac { N_\beta} { N} 
\sum_{k=1}^{\infty}  \rho_{k,c_0}
\geq
\dfrac{N_\beta}{N} \rho_{k_0,c_0},
\quad
\dfrac { N_\beta^2} { N^2} 
\sum_{k=1}^{\infty}  \rho^2_{k,c_0}
\geq
\dfrac{N_\beta^2}{N^2} \rho_{k_0,c_0}^2,
$$
and that $N_\beta \to \infty$ as $N\to \infty$. \end{proof}
\vskip .1cm

{\bf 2.}
We showed in Proposition \ref{prop: static}, when $c<c_0$ in the three regimes , that $\phi_c$ corresponds in a sense to the limit shape under the measures $\Pam_{c,N}$.  We now state the same happens when $c=c_0$.

\begin{Lem}
\label{app_2}
We have that the limit \eqref{static limit} holds when $c=c_0$.
\end{Lem}

\begin{proof}
A main tool in the proof of Proposition \ref{prop: initial convergence and mass for condition 1}, which applies under measures $\Pam_{c,N}$ when $c<c_0$, is the variance estimate in Lemma \ref{Lem:  mean variance mu N}, which as seen in Lemma \ref{app_1} above does not hold.  
However, since $G$ has compact support, it is enough to make estimates for $k\in [aN,bN]$, where the support of $G$ is contained in $[a,b]$ for $0<a<b$. 

We claim that in all the three regimes,
\begin{equation} \label {eqn: variance estimate for invariant on compact}
 \lim_{N\rightarrow \infty}\dfrac {N_\beta^2} {N^2} 
\sum_{aN\leq k\leq bN}  \text{Var}_{ \Pam_{c_0,N}} (\eta(k)) = 0.
\end{equation}
Indeed,
notice that $\text{Var}_{ \Pam_{c_0,N}} (\eta(k)) = \rho_{k,c_0}^2 + \rho_{k,c_0}$
where $\rho_{k,c_0} =  \dfrac{c_0e^{-\beta \Eb_k-k/N}}{1-c_0e^{-\beta \Eb_k-k/N}}$.
Since $N_\beta = o(N)$, the claim \eqref{eqn: variance estimate for invariant on compact} would follow from the bound 
$
\sup_N \sup_{aN\leq k\leq bN} N_\beta\rho_{k,c_0} <\infty$. Such a bound holds in fact by the proof of Lemma \ref{lem: site particle bound}.

Hence, under $\Pam_{c_0,N}$, we conclude $N_\beta N^{-1}\sum G(k/N)(\eta(k) - \rho_{k,c_0})\to 0$ in probability. To finish, we need only show that
\begin{equation} \label {eqn: convergence of claim static}
\lim_{N\to \infty}
\dfrac {1}{N} 
\sum_{k=1}^{\infty}
G\Big(\dfrac k N\Big) N_{\beta}  \rho_{k,c_0} 
=
\int_0^{\infty} G(x) \phi_{c_0}(x) dx,
\end{equation}
where, we note that the summation of $k$ above is actually on $aN\leq k\leq bN$.  Recall the formula for $\rho_{k,c_0}$ in \eqref{rho k c}.

When $\beta=0$, we have $N_\beta = 1$ and $c_0=1$.  Then,
\begin{equation*}
N_{\beta} \rho_{k,c_0} = \dfrac{e^{-k/N}}{1- e^{-k/N}}
\to \dfrac{ e^{-x}}{1- e^{-x}}=\phi_{c_0},
\text{ as } N\to \infty, \frac kN\to x.
\end{equation*}
Then,  \eqref{eqn: convergence of claim static} follows from dominated convergence.
 
However, when $\beta> 0$, note first $N_{\beta} \theta_{k,c_0} = c_0e^{-\beta (\Eb_k - \Eb_N)-k/N}$ and $\Eb_k - \Eb_N = u(\ln k ) - u(\ln N)$.
 By the mean value theorem,
$\Eb_k - \Eb_N \to \ln x \lim_{z\to \infty} u'(z)$ as $N\to \infty$ and $k/N \to x$.  Note also that $N_\beta =e^{\beta \Eb_N} \to \infty$ (cf. \eqref{Nbeta eqn}).
Then,
\begin{equation*}
N_{\beta}  \rho_{k,c_0} = \frac{N_\beta \theta_{k,c_0}}{1-\theta_{k,c_0}}\to c_0 e^{-\beta \ln x \lim_{z\to \infty} u'(z)} e^{-x}=\phi_{c_0}(x), \text{ as } N\to \infty, \frac kN\to x.
\end{equation*}
Again, by dominated convergence theorem, \eqref{eqn: convergence of claim static} follows.
\end{proof}

\medskip
{\bf Acknowledgements.}
This research was partly supported by ARO-W911NF-18-1-0311 and a Simons Foundations Sabbatical grant.





\end{document}